\documentclass[10pt]{article}
\usepackage{latexsym}
\usepackage{amssymb}
\usepackage{amsthm}
\usepackage{amsmath}
\usepackage{graphicx}
\usepackage[latin1]{inputenc}
\usepackage{hyperref}

\newtheorem{theorem}{Theorem}[section]
\newtheorem{prop}[theorem]{Proposition}
\newtheorem{coro}[theorem]{Corollary}
\newtheorem{lemma}[theorem]{Lemma}
\newtheorem{remark}[theorem]{Remark}

\newcommand{\R}{\mathbb{R}}             
\newcommand{\N}{\mathbb{N}}             
\newcommand{\C}{\mathbb{C}}             
\renewcommand{\S}{\mathbb{S}}             

\newcommand{\half}{\frac{1}{2}}

\newcommand{\Ga}{\Gamma^1}         
\newcommand{\Gb}{\Gamma^2}         
\newcommand{\Gc}{\Gamma^3}         
\newcommand{\Gd}{\Gamma^0}         

\newcommand{\Section}[1]{\section{#1} \setcounter{equation}{0}}

\author{Damien Gobin\footnote{D\'epartement de Math\'ematiques, Universit\'e de Nantes, 2, rue de la Houssini\`ere, BP
     92208, 44322 Nantes Cedex 03. Email adress: damien.gobin@univ-nantes.fr. Research supported by the French National Research Project
AARG, No. ANR-12-BS01-012-01.}}
\title{Inverse scattering at fixed energy for massive charged Dirac fields in de Sitter-Reissner-Nordstr\"{o}m black holes}
\date{\today}

\begin{document}

\maketitle


\begin{abstract}

In this paper, we consider massive charged Dirac fields propagating in the exterior region of de Sitter-Reissner-Nordstr\"{o}m black holes.
We show that the parameters of such 
black holes are uniquely determined by the partial knowledge of the corresponding scattering operator $S(\lambda)$ at a fixed energy $\lambda$.
More precisely, we consider the partial wave scattering operators 
$S(\lambda,n)$ (here $\lambda \in \mathbb{R}$ is the energy and $n \in \mathbb{N}^{\star}$ denotes the angular momentum) defined as the restrictions of the full scattering operator on a well chosen basis of 
spin-weighted spherical harmonics.
We prove that the knowledge 
of the scattering operators $S(\lambda,n)$, for all $n \in \mathcal{L}$, where $\mathcal{L}$ is a subset of $\mathbb{N}^{\star}$ that satisfies the M\"{u}ntz condition 
$\sum_{n \in \mathcal{L}} \frac{1}{n} = + \infty$, allows to recover the mass, the charge and the cosmological constant of a dS-RN black hole.
The main tool consists in the complexification of the angular momentum $n$ and in studying the analytic properties of the 
``unphysical'' corresponding data in the complex variable $z$.

\vspace{0.5cm}
\end{abstract}


\Section{Introduction and statement of the main result}

General Relativity was introduced by Einstein in 1915 and is one of the most important and beautiful theory of the twentieth century. This theory predicts the existence 
of black holes which are the objects of main interest of this paper. Although they are complicated objects to study from an astrophysical point of view, they are in fact quite simple to describe in theory. Indeed they only depend on a few physical parameters (see for instance \cite{Heu} and \cite{Lef} for spherically symmetric black holes and \cite{Wald} for rotating black holes): their mass, their electric charge, their angular momentum and possibly the cosmological constant of the universe. Inverse and direct scattering theory in black hole spacetimes are subjects of great interest. Direct scattering for Schwarzschild (static, uncharged and spherically symmetric), 
(de Sitter)-Reissner-Nordstr\"{o}m (static, charged and spherically symmetric) and Kerr (uncharged and rotating) black holes was studied for instance by Bachelot, Dimock, Kay, Nicolas, Jin, Melnik and H\"{a}fner in \cite{bach1,bach2,Dim,Kay,
Nic,Jin,Mel1,Mel2,Haf1,Haf2}. Among other reasons, these studies were motivated by the discovery unexpected phenomena such as the Hawking effect and the superradiance phenomenon. We refer for instance to Bachelot \cite{bach3,bach4}, H\"afner \cite{Haf3} and Melnyk \cite{Mel} for an application of scattering results in terms of the Hawking effect. Concerning the inverse scattering in (de Sitter)-Reissner-Nordstr\"{o}m black holes, this problem has been adressed in the serie of papers by Daud\'{e} and Nicoleau \cite{DN,DN2,DN4}. 

This work is a continuation of the papers \cite{DN,DN2,DN4} and deals with de Sitter-Reissner-Nordstr\"{o}m black holes. These are spherically symmetric and charged solutions of the Einstein-Maxwell equations that are completely characterized by three parameters: the mass 
$M > 0$ and the electric charge $Q \in \mathbb{R}$ of the black hole and the cosmological constant $\Lambda > 0$ of the universe. The object of this paper is to study an inverse scattering problem in a dS-RN black hole whose unknowns are its parameters. In fact, we shall see in the course of the proof that we are able to recover more than a few parameters (see below).


A peculiarity of scattering theory in black hole spacetimes is the fact that we have to deal with two \emph{asymptotic regions}. Indeed, we adopt the point of view of an observer located far away from the event and cosmological horizons and static with respect to them. It is well known that such an observer perceives both horizons as asymptotic regions. Hence we shall understand in the following de Sitter-Reissner-Nordstr\"om black holes as spherically symmetric spacetimes with two asymptotic ends (the cosmological and the event horizons). It is worth mentioning that the geometry at both horizons is of asymptotically hyperbolic type. The question we adress is the following: is there any way to characterize uniquely the parameters of the black hole by an inverse scattering experiment from the point of view of a static observer?

To reformulate our main problem we have to introduce the wave operators associated to massive charged Dirac fields evolving in the exterior region of the black hole. We denote by $W_{(-\infty)}^{\pm}$ the wave operators corresponding to the part of the massive and charged Dirac fields which scatters toward the event horizon and by $W_{(+\infty)}^{\pm}$ the wave operators corresponding the part of Dirac fields which scatters toward the cosmological horizon. Thanks to \cite{Dau, Mel} we know that the global wave operators defined by
\begin{equation}\label{globwav}
 W^{\pm} = W_{(-\infty)}^{\pm} \oplus W_{(+\infty)}^{\pm},
\end{equation}
exist and are asymptotically complete on the Hilbert space of scattering data. This allows to define a global scattering operator $S$ by the usual formula
\[ S = (W^+)^{\star}W^-.\]

The scattering operator is the main object of study of this paper. It contains all the scattering information as viewed by observers living far from the horizons of a dS-RN black hole. Thanks to this definition, we can reformulate our main question: is the knowledge of $S$ a sufficient information to uniquely characterize the parameters of a dS-RN black hole ?

The aim of this paper is to show that the parameters $M$, $Q$ and $\Lambda$ are uniquely characterized from the knowledge of the scattering operator at a fixed energy and more precisely, from the knowledge of the reflection operators at a fixed energy (see below). In fact, we mention that we are able to recover more than only the parameters of the black hole since we show the uniqueness (up to a certain diffeomorphism) of some scalar functions appearing in the Dirac equation. Note at last that, contrary to \cite{DN2}, we don't need the knowledge of the scattering operator on an interval of energy but only at \emph{a fixed energy} to recover the metric of the black hole. Our result is an adaptation to the case of massive and charged Dirac fields of a similar result given in \cite{DN} for massless and uncharged Dirac fields. In particular, the (physically relevant) addition of a mass term makes more complicated both the definition of the scattering matrix and the technical details of the proof of our main 
Theorem.

\subsection{de Sitter-Reissner-Nordstr\"{o}m black holes}\label{dS-RN}
\noindent
In Schwarzschild coordinates, the exterior region of a dS-RN black hole is described by the four-dimensional manifold
\[\mathcal{M} = \mathbb{R}_t \times ]r_-,r_+[_r \times \mathbb{S}_{\omega}^2,\]
equipped with the Lorentzian metric
\begin{equation}\label{metric}
  g = g_{\mu \nu} dx^{\mu} dx^{\nu} = F(r) dt^2 - F(r)^{-1}dr^2 - r^2 d \omega ^2,
\end{equation}
where
\begin{equation}\label{F}
 F(r) = 1 - \frac{2M}{r} + \frac{Q^2}{r^2} - \frac{\Lambda r^2}{3},
\end{equation}
and $d \omega^2 = d\theta^2 + \sin(\theta)^2 d\varphi^2$ is the euclidean metric on the sphere $\mathbb{S}^2$. De Sitter-Reissner-Nordstr\"{o}m (dS-RN) black holes are spherically symmetric electrically charged solutions of the Einstein equations
\[ G_{\mu \nu} = 8 \pi T_{\mu \nu},\]
where $G_{\mu \nu}$ is the Einstein tensor, $T_{\mu \nu}$ is the energy momentum tensor,
\[T_{\mu \nu} = \frac{1}{4 \pi} (F_{\mu \rho} F_{\nu}^{\rho} - \frac{1}{4} g_{\mu \nu} F_{\rho \sigma} F^{\rho \sigma}),\]
where $F_{\mu \nu}$ is the electromagnetic 2-form, solution of Maxwell's equations $\nabla^{\mu} F_{\nu \rho} = 0$, $\nabla_{[\mu}F_{\nu \rho ]} = 0$ and given here in terms of a global electromagnetic vector potential 
$F_{\mu \nu} = \nabla_{[\mu} A_{\nu ]}$, $A_{\nu}dx^{\nu} = - \frac{Q}{r} dt$. The quantities $M$ and $Q$ are interpreted as the mass and the charge of the Reissner-Nordstr\"om black hole and 
$\Lambda > 0$ is the cosmological constant of the universe.

Let us look at the singularities of the metric $g$. Firstly, $F$ is singular at the point $\{r = 0\}$. This is a curvature singularity meaning that some contraction of the Riemann tensor explodes when $r \to 0$.
Secondly, the spheres whose radii are the roots of $F$ are also singularities for the metric $g$ (the coefficient of the metric $g$ involving $F^{-1}$ blows up in this case). 
We assume here that the function $F(r)$ has three simple positive roots $0 < r_c < r_- < r_+$ and a negative one $r_n < 0$. This is always achieved if we suppose, for instance, that $Q^2 < \frac{9}{8} M^2$ and that $\Lambda M^2$ be small enough (see \cite{Lak}). 
The hypersurface $\{r = r_c\}$ is called the Cauchy horizon whereas the hypersurfaces $\{ r =r_-\}$ and $\{r = r_+\}$ are, respectively, the event and cosmological horizons. 
We shall only consider the exterior region of the black hole, i.e. the region $\{r_- < r < r_+\}$ lying between the event and cosmological horizons (remark that the function $F$ is positive here). 
Actually, the event and cosmological horizons which appear as singularities of the metric (\ref{metric}) are coordinates singularities and are due to our bad choice of coordinates system. Using appropriate coordinates, these horizons can be understood as regular null hypersurfaces that can be crossed one way but would require speeds greater than that of light to be crossed the other way. 

As mentionned previously, the point of view implicitely adopted throughout this work is that of static observers located far from the event and cosmological horizons of the black hole. We think typically of a telescope on earth aiming at the black hole or at the cosmological horizon. We understand these observers as living on world lines $\{ r = r_0 \}$ with $r_- \ll r_0 \ll r_+$. The variable $t$ of the 
Schwarzschild coordinates corresponds to their proper time. From the point of view of our observers, it is important to understand that the event and cosmological horizons are the natural \emph{boundaries} of the observable world. This can be more easily understood if we notice that the horizons are in fact never reached in a finite time $t$ by incoming and outgoing radial null geodesics, the trajectories followed by classical light-rays radially at the black hole and either at the cosmological horizon. Both horizons are thus perceived as {\em asymptotic regions} by our static observers.

Instead of working with the radial variable $r$, we describe the exterior region of the black hole by using the Regge-Wheeler (RW) radial variable. 
The RW variable $x$ is defined implicitely by 
\[\frac{dx}{dr} = F^{-1}(r),\]
or explicitely by
\begin{equation}\label{defx}
 x = \frac{1}{2\kappa_n} \ln(r-r_n) + \frac{1}{2\kappa_c} \ln(r-r_c) +\frac{1}{2\kappa_-} \ln(r-r_-) +\frac{1}{2\kappa_+} \ln(r-r_+) + C,
\end{equation}
where $C$ is any constant of integration and the quantities $\kappa_j$, $j=n,c,-,+$ are defined by
\begin{equation}\label{kappa}
 \kappa_n = \frac{1}{2} F'(r_n), \quad \kappa_c = \frac{1}{2} F'(r_c), \quad \kappa_- = \frac{1}{2} F'(r_-), \quad \kappa_+ = \frac{1}{2} F'(r_+).
\end{equation}
The constants $\kappa_- > 0$ and $\kappa_+ < 0$ are called the surface gravities of the event and cosmological horizons, respectively. Note from (\ref{defx}) that the event 
and cosmological horizons $\{ r = r_{\pm} \}$ are pushed away to the infinities $\{ x = \pm \infty \}$ using the RW variable $x$. Moreover, it can be shown easily that, in this new coordinates system, the incoming and outgoing null radial geodesics become the straight lines $\{ x = \pm t\}$ in the $t-x$ plane. Hence, working with the RW radial variable achieves in practice the fact that the event and cosmological horizons are \emph{asymptotic regions} for our observers.  

Finally, we note the presence of a constant of integration $C$ in the definition of $x$. The importance of such a constant is explained in Section 4.1.5, Proposition 4.12 of \cite{DN3}. In this Proposition, it is shown that there is a dependence of the scattering matrix 
under the choice of the constant of integration $C$. Since the exterior region of a dS-RN black hole can be described uniquely by any choice of the Regge-Wheeler variable $x$, we could identify all the possible forms of the reduced scattering matrices in the statement of our main result. However, for the sake of simplicity, we take $C = 0$ in our study (the case of $C \neq 0$ could be treated in the same way using the explicit dependence on $C$ given in \cite{DN3} for the scattering matrix).

\subsection{The scattering matrix and statement of the result}\label{statement}
\noindent
As in \cite{DN2}, we consider massive charged Dirac fields propagating in the exterior region of a dS-RN black hole. Scattering theory for these Dirac fields has been the object of the papers \cite{Dau, Mel}. We shall
 use the form of the Dirac equation obtained therein as the starting point of our study. We refer to Section \ref{Dirpb} for the details.

The considered massive charged Dirac fields are represented by 4-components spinors $\psi$ belonging to the Hilbert space
\[ L^2(\mathbb{R} \times \S^2 ; \mathbb{C}^4),\]
which satisfy the evolution equation
\begin{equation}\label{eqDir}
 i \partial_t \psi = (\Gamma^1 D_x + a(x) D_{\mathbb{S}^2} + b(x) \Gamma^0 + c(x)) \psi.
\end{equation}
The symbols $D_x$ stands for $-i\partial_x$ whereas $D_{\mathbb{S}^2}$ denotes the Dirac operator on $\S^2$ which in spherical coordinates, takes the form,
\begin{equation}\label{opeDirac}
 D_{\mathbb{S}^2} = -i \Gamma^2 \left( \partial_{\theta} + \frac{\cot(\theta)}{2} \right) - \frac{i}{\sin(\theta)} \Gamma^3 \partial_{\varphi}.
\end{equation}
The potentials $a$, $b$ and $c$ are scalar smooth functions given in terms of the metric (\ref{metric}) by
\begin{equation}\label{defpot}
 a(x) = \frac{\sqrt{F(r)}}{r}, \quad b(x) = m \sqrt{F(r)}, \quad c(x) = \frac{qQ}{r},
\end{equation}
where $m$ and $q$ respectively denotes the mass and the electric charge of the Dirac fields. Finally, the matrices $\Ga$, $\Gb$, $\Gc$ and $\Gd$ appearing in (\ref{eqDir}) and (\ref{opeDirac}) 
are usual $4 \times 4$ Dirac matrices that satisfy anticommutation relations (see (\ref{anticom})).
The equation (\ref{eqDir}) is spherically symmetric and in consequence can be separated into ODEs. The stationary scattering will be shown to be governed by a countable family of one-dimensional stationary Dirac equations that take the following form:
\begin{equation}\label{eqstat}
 \left( \Gamma^1 D_x - \left( l +\half \right)  a(x) \Gamma^2 + b(x) \Gamma^0 + c(x) \right) \psi(x,\lambda,l ) = \lambda \psi(x,\lambda,l),
\end{equation}
and which are restrictions of the full stationary equation to a well chosen basis of spin-weighted spherical harmonics (indexed here by $l = \half$, $\frac{3}{2},...$) invariant for the full equation. Here $\lambda$ is the energy of the 
considered waves and $n:= l + \half \in \mathbb{N}^{\star}$ is called the angular momentum.\\

Concerning the potentials, we know from \cite{DN2} that 
$a(x),\,b(x),\,a'(x),\,b'(x) = O(e^{\kappa_{\pm}x})$, $c(x) = O(1)$ and $c'(x) = O(e^{2\kappa_{\pm}x})$ as $x \to \pm \infty$ (see Lemma \ref{aspot} for precise asymptotics).
We make some remarks about these asymptotics. On one hand, although the Dirac fields are massive, they propagate asymptotically as in the massless case since $b(x) \to 0$ as $x \to \pm \infty$. This is due to the effects of the intense gravitation near the event and cosmological horizons of the black hole. On the other hand, it remains in the asymptotics an influence of the interaction between the electric charges $q$ and $Q$ since the potential $c(x)$ satisfies $c(x) \to c_{\pm}$ when $x\to \pm \infty$. Actually we shall see later that we can come down to the usual case of a one dimensional Dirac equation with a $L^1$ potential using a unitary transformation.

As we have already said in the introduction, the existence and the asymptotic completeness of the global 
wave operators associated to (\ref{eqstat}) defined for all angular momentum $n \in \N^*$ by
\[W_n^{\pm} = W^{\pm}_{n,(-\infty)} \oplus W^{\pm}_{n,(+\infty)},\]
proved in \cite{Dau,Mel} allows to define a scattering operator $S(n)$ on each spin-weighted spherical harmonics by the usual formula
\[S(n) = (W_n^+)^{\star} W_n^-.\] 
We refer to Section \ref{Dirpb} for the details of the definition and construction of these operators. Moreover this operator can be decomposed as
\[S(n) = \begin{pmatrix}
                   T_L(n) & R(n) \\ L(n) & T_R(n)
                  \end{pmatrix},\]
where the first two terms $T_R(n)$ and $T_L(n)$ are understood as transmission operators whereas the last two terms $L(n)$ and $R(n)$ are understood as reflection operators 
(see Section \ref{Dirpb} for more details).

In Section \ref{first}, using a conjugation by a unitary Fourier transform, we obtain a stationary representation of the scattering operator $S(n)$ as a direct integral of scattering matrices $S(\lambda,n)$. Since we separated the full Dirac equation into a countable family of 
one-dimensional stationary Dirac equations which are the restrictions of the full stationary equation to a well chosen basis of spin-weighted spherical harmonics, we can define the global scattering matrix by $S(\lambda) = \oplus_{n \in \mathbb{N}^{\star}} S(\lambda,n)$, 
where
\[S(\lambda,n) = \begin{pmatrix}
                   T_L(\lambda,n) & R(\lambda,n) \\ L(\lambda,n) & T_R(\lambda,n)
                  \end{pmatrix}.\]
Once again, we refer to Section \ref{Dirpb} for the details of this decomposition. The matrices $S(\lambda,n)$ are called the partial scattering matrices. What is important to keep in mind is that the knowledge of the scattering matrix $S(\lambda)$ is equivalent to the knowledge of the partial scattering matrices $S(\lambda,n)$ for all $n \in \mathbb{N}^{\star}$.

Roughly speaking the main result of this paper states that the partial knowledge of the scattering operator $S(\lambda)$ at a fixed energy $\lambda \in \mathbb{R}$ uniquely determines the mass $M$ and the charge $Q$ of the black hole as well as the cosmological constant $\Lambda$ of the universe. More precisely, it suffices to know one of the reflection operators of
the partial scattering operators $L(\lambda,n)$, $R(\lambda,n)$ at a fixed energy $\lambda$ on a subset $\mathcal{L} \subset \mathbb{N}^{\star}$ that satisfies the M\"{u}ntz condition $\sum_{n \in \mathcal{L}} \frac{1}{n} = \infty$ in order 
to prove the uniqueness of the parameters $M$, $Q$ and $\Lambda$.

Precisely, the main result of this paper is the following:

\begin{theorem}\label{mainthm}
 Let $(M,Q,\Lambda)$ and $(\tilde{M},\tilde{Q},\tilde{\Lambda})$ be the parameters of two dS-RN black holes. We denote by $S(\lambda,n)$ and $\tilde{S}(\lambda,n)$ the corresponding partial wave scattering operators at a fixed energy $\lambda \in \mathbb{R}$. Consider a subset 
 $\mathcal{L}$ of $\mathbb{N}^{\star}$ that satisfies the M\"{u}ntz condition $\sum_{n \in \mathcal{L}} \frac{1}{n} = \infty$ and assume that at a fixed energy $\lambda \in \mathbb{R}$
 one of the following assertions holds,
\[
 (i) \quad L(\lambda,n) = \tilde{L}(\lambda,n), \quad \forall n \in \mathcal{L},
\]
\[
  (ii) \quad R(\lambda,n) = \tilde{R}(\lambda,n), \quad \forall n \in \mathcal{L}.
\]
 Then, there exists a diffeomorphism $\psi : \mathbb{R} \to \mathbb{R}$ such that
 \[ \frac{c(\psi(x)) - \lambda}{a(\psi(x))} = \frac{\tilde{c}(x)-\lambda}{\tilde{a}(x)} \quad \and \quad \frac{b(\psi(x))}{a(\psi(x))} = \frac{\tilde{b}(x)}{\tilde{a}(x)}.\]
 As a consequence we get
 \[ M = \tilde{M}, \quad Q = \tilde{Q}, \quad \Lambda = \tilde{\Lambda}.\]
\end{theorem}

\begin{remark}
\begin{enumerate}
 \item The question of determining the parameters $M$, $Q$ and $\Lambda$ from the transmission coefficient at a fixed energy is still an open problem.
 \item Actually, it is enough to prove our uniqueness result to assume that the reflection coefficients are known up to an error $O(e^{-2nB})$ for some $B \in ]0, \min(A,\tilde{A})[$. Indeed, we can use the 
 idea of \cite{Pap1} which proves a local inverse scattering result at a fixed energy in spherically symmetric asymptotically hyperbolic manifolds.
 \item Moreover, it is also sufficient in order to prove our uniqueness result to assume that the reflection coefficients are known up to a constant unitary matrix.
Indeed in the proof of the uniqueness result under the knowledge of the reflection coefficient $R$ (see Section \ref{preuveR}), we come down to this case and we conclude under a little technical modification.
 \item We emphasize that in the proof of this Theorem, we first obtain, without using the explicit form of the potentials, the equalities
 \[ \quad \frac{c-\lambda}{a}(h(X)) = \frac{\tilde{c}-\lambda}{\tilde{a}}(\tilde{h}(X)), \quad \frac{b}{a}(h(X)) = \frac{\tilde{b}}{\tilde{a}}(\tilde{h}(X))\]
(see the following Section for the definitions of the Liouville variable $X$ and of the diffeomorphisms $h$ and $\tilde{h}$). Then, thanks to the explicit definition of the potentials, 
we obtain the equality of the parameters. We also observe that, if we suppose the equality of the partial scattering matrices at {\em two} fixed energies (possibly for different M\"untz sets), we can conclude to the equality of each potentials $a,b,c$ separatly up to a certain diffeomorphism. 
 \item Contrary to \cite{DN}, the case of a zero energy is not an obstruction for our inverse problem if we suppose that the electric charge $q$ of the Dirac fields is non vanishing. 
 \item Contrary to the massless and uncharged case, we determine exactly the charge and not the square of the charge. This is due to the presence of the potential of electric type $c$. 
\end{enumerate}
\end{remark}


\subsection{Overview of the proof}\label{overview}
\noindent
There are three steps which constitute the proof of Theorem \ref{mainthm}. The aim of this Subsection is to describe them.

The first step of the proof consists in getting rid of the long-range potential $c$ in the expression of the one-dimensional stationary Dirac operators (\ref{opeDirac}) 
in order to obtain a new Dirac operator with short-range potential. This new operator satisfies the usual framework of inverse scattering theory for one dimensional Dirac 
operators studied in \cite{AKM}. To do this we conjugate the one-dimensional selfadjoint operator 
\[ H =   \Gamma^1 D_x - n a(x) \Gamma^2 + b(x) \Gamma^0 + c(x),\]
by a well chosen unitary operator $U$ defined by
$$
  U = e^{-i \Gamma^1 C^{-}(x)},\quad C^-(x) = \int_{-\infty}^{x} (c(s)-c_-) \, \mathrm ds + c_- x, 
$$
so that,
\begin{equation*}
 A = U^{\star} H U,
\end{equation*}
can then be written as,
\[A = \Gamma^1 D_x + W(x),\]
where
\[W(x) = e^{i\Gamma^1 C^-(x)}(n \,a(x) \Gamma^2 + b(x) \Gamma^0) e^{-i\Gamma^1 C^-(x)}.\]
Thanks to the asymptotics of the potentials, we see that $W$ is exponentially decreasing near the horizons and thus belong to $L^1(\R)$. 
Thus, this new operator $A$ lies into the framework of the paper 
\cite{AKM} and we can define in a straightforward way the scattering data. Following \cite{AKM}, the scattering matrix $\hat{S}(\lambda,n)$ is defined in terms of stationary solutions 
with prescribed asymptotics at infinity, called Jost functions. These are $4 \times 4$ matrix-valued functions $\hat{F}_L$ and $\hat{F}_R$ solutions of 
\begin{equation}\label{eqA}
 A \psi(x,\lambda,n) = \lambda \psi(x,\lambda,n),
\end{equation}
having the asymptotics
\begin{equation}\label{asJG}
 \hat{F}_L(x,\lambda,n) = e^{i \Gamma^1 \lambda x}(I_4 + o(1)), \quad x \to +\infty,
\end{equation}
\begin{equation}\label{asJD}
\hat{F}_R(x,\lambda,n) = e^{i \Gamma^1 \lambda x}(I_4 + o(1)), \quad x \to -\infty,
\end{equation}
where $I_4$ is the identity matrix.
Both Jost solutions form fundamental matrices of (\ref{eqA}). In consequence, there exists a $4 \times 4$ matrix $\hat{A}_{L}(\lambda,n)$ depending only on the energy and the angular momentum $n$ such that the Jost functions are connected by 
\begin{equation}\label{defAL}
 \hat{F}_L(x,\lambda,n) = \hat{F}_R(x,\lambda,n) \hat{A}_{L}(\lambda,n).
\end{equation}
Similarly, we define the $4 \times 4$ matrix $\hat{A}_{R}(\lambda,n)$ by 
\begin{equation}\label{defAR}
\hat{F}_L(x,\lambda,n) \hat{A}_{R}(\lambda,n) = \hat{F}_R(x,\lambda,n).
\end{equation}
The coefficients of the matrices $\hat{A}_L$ and $\hat{A}_R$ contain all the scattering information of the equation (\ref{eqA}). In particular, using the notations
\begin{equation}\label{notALAR}
 \hat{A}_L(\lambda,n) = \begin{pmatrix}
                   \hat{A}_{L1}(\lambda,n) & \hat{A}_{L2}(\lambda,n) \\ \hat{A}_{L3}(\lambda,n) & \hat{A}_{L4}(\lambda,n)
                  \end{pmatrix},\quad \hat{A}_R(\lambda,n) = \begin{pmatrix}
                   \hat{A}_{R1}(\lambda,n) & \hat{A}_{R2}(\lambda,n) \\ \hat{A}_{R3}(\lambda,n) & \hat{A}_{R4}(\lambda,n)
                  \end{pmatrix},
\end{equation}
where $\hat{A}_{Lj}$ and $\hat{A}_{Rj}$ are $2 \times 2$ matrices, the partial wave scattering matrix $\hat{S}(\lambda,n)$ is then defined by 
\begin{equation}\label{defS}
 \hat{S}(\lambda,n) = \begin{pmatrix}
                   \hat{T}_L(\lambda,n) & \hat{R}(\lambda,n) \\ \hat{L}(\lambda,n) & \hat{T}_R(\lambda,n)
                  \end{pmatrix},
\end{equation}
where (see \cite{AKM}, eqs. (3.6)-(3.7))
\begin{equation}\label{defT}
 \hat{T}_L(\lambda,n) = \hat{A}_{L1}(\lambda,n)^{-1}, \quad \quad \hat{T}_{R}(\lambda,n) = \hat{A}_{R4}(\lambda,n)^{-1},
\end{equation}
\begin{equation}\label{defL}
 \hat{L}(\lambda,n) = \hat{A}_{L3}(\lambda,n) \hat{A}_{L1}(\lambda,n)^{-1} = -\hat{A}_{R4}(\lambda,n)^{-1} \hat{A}_{R3}(\lambda,n)
\end{equation}
and
\begin{equation}\label{defR}
 \hat{R}(\lambda,n) = - \hat{A}_{L1}(\lambda,n)^{-1} \hat{A}_{L2}(\lambda,n) = \hat{A}_{R2}(\lambda,n) \hat{A}_{R4}(\lambda,n)^{-1}.
\end{equation}

\begin{remark}
We emphasize that the Dirac fields studied in this paper are \emph{massive and electrically charged}. This entails several technical complications with respect to the case of massless Dirac fields studied in \cite{DN}. 
First, the scattering data $A_{Lj}(\lambda,n)$ in the massless and uncharged case turn out to be complex numbers and not $2 \times 2$ matrices. Similarly, the Jost functions are $2 \times 2$ matrix-valued functions and not $4 \times 4$ matrix-valued functions. This is due to the fact that the additional mass term prevents us from using $2$-components spinors in the description of the Dirac fields. However, if we take $m=q=0$ in our study, we can decouple the $4$ components of our Dirac spinors into two $2$-components Dirac spinors that satisfy exactly the massless Dirac Equation (1.7) of \cite{DN}. Hence we shall recover the results of \cite{DN} as a particular case of our result.  
\end{remark}
\noindent
At the end of Section 3, we obtain systems of second order differential equations (in the variable $x$) satisfied by the components of the Jost matrices related to the unconjugate operator $H$. These systems of ODEs will be important in the next analysis in order to obtain some refined estimates on the components of the Jost functions. To obtain these systems of ODEs, we use the equations satisfied by the components 
of the Jost matrices for the conjugate operator and the link between these components and the components of the Jost matrices for the operator $H$.


The second step of the proof is the most important one. Following \cite{DN}, the main idea of this paper is to complexify the angular momentum $n = l + \half$ and study the 
analytic properties of the ``unphysical'' corresponding scattering data with respect to the variable $z = n \in \mathbb{C}$. The general idea to consider complex angular 
momentum originates in a paper by Regge \cite{Reg} as a tool in the 
analysis of the scattering matrix of Schr\"{o}dinger operators in $\mathbb{R}^3$ with spherically symmetric potentials (see also \cite{New, Cha} for a detailed account of 
this approach). To understand the analytic properties in the complex plane of the scattering data, we need good asymptotics of the Jost functions when the complex angular 
momentum $z$ becomes large.

\begin{remark}
 Actually we don't need the asymptotics for all large $z$ in the complex plane. Indeed it is sufficient in the proof of Theorem \ref{mainthm} (see Section 6) to have good asymptotics only on the real axis. Precisely, we need to know that the components of the Jost functions are bounded on $i \mathbb{R}$ and to show that these components are of exponential type. However it's not so much harder to obtain the asymptotics 
 on $\mathbb{C}$. We hope these asymptotics could be useful in a future work.
\end{remark}


\begin{remark}
 The presence of the mass term has a fondamental consequence in the research of the asymptotics of the components of the Jost matrices. Precisely, it is shown in the 
massless case (see \cite{DN}) that the components of the Jost matrices have power series expansions in the variable $z$ and satisfy uncoupled ODEs. From these ODEs, it 
is straightforward to see that the components of the Jost matrices are perturbations of the modified Bessel functions from which we can deduce their asymptotics easily. 
In the massive charged case, the components of the Jost matrices satisfy \emph{systems} of ODEs that makes the analysis of \cite{DN} much more involved. Nevertheless, 
using a perturbative argument and good estimates on the Green kernels, we are able to show that they are still perturbations of modified Bessel functions from which we 
can obtain their asymptotics.\\
The presence of the mass has also consequences on the symmetries of the scattering data. Indeed, as it is shown in \cite{DN} (see Lemmas 3.1 and 3.3), the components of the Faddeev matrices (see Section 3 for the definition) are odd or even in $z$ and there is a relation of conjugation between their first and fourth components, and their second and third components. The same properties are also true for the coefficients of the matrix of scattering data. These symmetries are no longer true in the case of massive Dirac fields. As a consequence, we shall often have to use different strategies to prove the corresponding results of \cite{DN}. 
However, since for large $z$, the most important term of the potential is $za(x)$, the mass should have no influence in the regime $z$ large. We expect then to find the symmetries given in \cite{DN} in the asymptotics of the scattering data as $|z| \to \infty$. 
\end{remark}
\noindent
As mentionned in the previous remark, the first step to obtain the asymptotics consists in showing that the components of the Jost functions are perturbations of modified Bessel 
functions. To do this, it is convenient to introduce a new radial variable $X$ by the Liouville transformation
\begin{equation}\label{defX}
  X = g(x) = \int_{-\infty}^{x} a(s) \, \mathrm ds.
\end{equation}
Note that $g$ is a diffeomorphism from $\mathbb{R}$ to $]0,A[$ where $A = \int_{\mathbb{R}} a(s) \, \mathrm ds$. Let us denote by 
$h(X)$ the inverse transformation. The reason why we introduce such a variable lies in the observation that, thanks to the coupled differential equations obtained in Section 
\ref{ED} for the Jost functions, we can show that the components $F_{Li,j}(h(X),\lambda,z)$ and $F_{Ri,j}(h(X),\lambda,z)$ of the Jost matrices satisfy then systems of second order differential equations of the form
\begin{equation}\label{eqsturm}
 f''(X) + q(X) f(X) = z^2 f(X) + r(X), \quad \quad X \in ]0,A[.
\end{equation}
Here, the potential $q(X)$ have quadratic singularities at the boundaries $0$ and $A$ and the term $r(X)$ is the coupling term between the components of the Jost functions and can be shown to be bounded 
at the boundaries. We note that the angular momentum (or coupling constant) $z = l+\half$ has now become the spectral parameter of the uncoupled part of this new equation.
We remark that this equation is an approximation of a modified Bessel equation. Using perturbation theory, we can then prove that the components of the Jost functions are perturbations of well chosen modified Bessel functions. Finally, using the well known (see \cite{Leb}) asymptotics of the modified Bessel functions, we are able to prove estimates on the components of Jost functions of the type $|F_{L/R}(X,\lambda, z)| \leq C e^{\vert \mathrm{Re}(z) \vert X}$. 

\begin{remark}
 This last step is a consequence of a fundamental difference with the massless case. Indeed, thanks to the symmetries, it is not so hard in \cite{DN} to prove that the 
 components of Jost functions are less than $e^{\vert \mathrm{Re}(z) \vert X}$ (see Lemma $3.4$). However, in our study there is no symmetry (due to the mass term) and it 
 is harder to obtain these estimates. To prove them, we use the Phragm\'en-Lindel\"{o}f's Theorem. Indeed, it is quite easy (see Section 3) to prove 
 that the components of Jost functions are less than 
 $e^{\vert z \vert X}$ and we can prove, using a generalized Duhamel's formula and precise estimates of modified Bessel functions (see Section \ref{partDuh}) that the 
 components of the Jost functions are bounded on $i\mathbb{R}$. 
\end{remark}

\noindent
As a consequence of the asymptotics, we show that these coefficients are in the class of analytic functions called Nevanlinna's class.
Let us define the Nevanlinna class $N(\Pi^+)$ as the set of all analytic 
functions $f(z)$ on the right half plane $\Pi^+  = \{ z \in \mathbb{C}: \, \mathrm{Re}(z) > 0 \}$ that satisfy
\[\sup_{0 < r <1} \int_{-\pi}^{\pi} \ln^+ \left| f \left( \frac{1-re^{i \varphi}}{1+re^{i \varphi}} \right) \right| \, \mathrm d\varphi < \infty,\]
where $\ln^{+}(x) = \ln(x)$ if $\ln(x) \geq 0$ and $\ln^{+}(x) = 0$ if $\ln(x) < 0$. The reason why we use this kind of functions is that functions in this class
are uniquely determined by their values on any subset $\mathcal{L} \subset \mathbb{N}^{\star}$ that 
satisfies the M\"{u}ntz condition $\sum_{n \in \mathcal{L}} \frac{1}{n} = \infty$ (see for instance \cite{Ram} and \cite{Rudin}). 
Thus, with a little more work, we are able to conclude from this that the equality between the reflection 
coefficients not only holds for the integers $n \in \N^*$, but for almost all $z$ in the complex plane (except the poles of the reflection coefficients). We enlarge in 
this way considerably the information at our disposal to determine the metric of the black hole. 

The last step to conclude the proof of the main Theorem \ref{mainthm} is an idea borrowed from \cite{FY}. 
Consider two black holes with parameters $M$, $Q$, $\Lambda$ and $\tilde{M}$, $\tilde{Q}$, $\tilde{\Lambda}$, respectively. 
We add a ``$\sim$'' to all quantities corresponding to the black hole with parameters $\tilde{M}$, $\tilde{Q}$ and $\tilde{\Lambda}$. Using the asymptotics of the components of the matrix $\hat{A}_L(\lambda,z)$, we first prove that
\[ A = \int_{\mathbb{R}} a(s) \, \mathrm ds = \int_{\mathbb{R}} \tilde{a}(s) \, \mathrm ds =  \tilde{A}.\]
We thus define a matrix-valued function $P(X,\lambda,z)$ which makes the link between the Jost function for the first black hole and the Jost function for the second one:
\[P(X,\lambda,z)  \tilde{F}_R(\tilde{h}(X),\lambda,z) = F_R(h(X),\lambda,z).\]
 Using the previously obtained asymptotics and the equality between the reflection coefficients on $\C$, we show that 
this matrix is (up to a sign) the identity matrix. 

\begin{remark}
In this step appears a last consequence of the presence of the mass term in our problem. In the massless case, we can perform explicit calculations for the scattering data 
in the case $z=0$ because the one-dimensional Dirac operator is simply $\Gamma^1 D_x$. These explicit expressions for the scattering data were used to prove that $P$ is 
(up to a sign) the identity once we have shown that $P$ is constant with respect to the $z$ variable. Because of the mass term, we cannot do explicit calculations in the 
massive charged case and we shall thus need a different strategy to obtain it. 
\end{remark}
\noindent
The Jost functions of the two black holes being so tightly linked, we conclude easily from this that 
\[ \forall X \in ]0,A[, \quad \frac{c-\lambda}{a}(h(X)) = \frac{\tilde{c}-\lambda}{\tilde{a}}(\tilde{h}(X)), \quad \frac{b}{a}(h(X)) = \frac{\tilde{b}}{\tilde{a}}(\tilde{h}(X)).\]
Finally, thanks to these equalities and the precise form of the potentials $a$, $b$ and $c$, we conclude
\[ M = \tilde{M}, \quad Q = \tilde{Q}, \quad \Lambda = \tilde{\Lambda}.\]

This paper is organized as follows. In Section 2 we recall the direct scattering results of \cite{Dau, Mel} useful for the later analysis. In Section 3 we first show how our model can be transformed in order to fit the 
framework of \cite{AKM} and we find systems of second order diffential equations satisfied by the components of the Jost functions in the variable $x$. 
In Section 4 we introduce the Liouville variable $X$ to obtain perturbed modified Bessel equations in this variable and we use estimates of the Green kernel of modified 
Bessel functions and some estimates of the components of the Jost functions to calculate precise asymptotics of the components of the Jost functions 
and of the matrix $\hat{A}_L$ for large $z$ in the complex plane. In Section 5 we use the analytic properties 
of the Jost functions and of the matrix of scattering data and a uniqueness result for functions in the Nevanlinna class to extend the range of validity of our hypothesis to almost the whole complex plane. Finally, in Section 6, we prove our main Theorem \ref{mainthm}. 

\section{Direct scattering problem}\label{Dirpb}
\noindent
In this Section, we first recall the expression of the Dirac equation in dS-RN black holes as well as the direct scattering theory obtained in \cite{Dau, Mel}.

As explained in Section \ref{dS-RN}, we describe the exterior region of a dS-RN black hole using the Regge-Wheeler variable $x$ defined by (\ref{defx}). We thus work on the manifold $\mathcal{B} = \mathbb{R}_t \times \Sigma$ with 
$\Sigma = \mathbb{R}_x \times S_{\theta,\varphi}^2$, equipped with the metric
\[g = F(r) (dt^2 - dx^2) - r^2 d \omega^2,\]
where $F$ is given by (\ref{F}) and $d\omega^2 = d\theta^2 + \sin^2(\theta) d\varphi^2$ the Euclidean metric on $\S^2$.

\subsection{Dirac equation and direct scattering results}
\noindent
Scattering theory for massive charged Dirac fields on the spacetime $\mathcal{B}$ has been the object of the papers \cite{Dau, Mel}. We briefly recall here the main results of these papers. In particular, we use the form of the Dirac equation obtained therein. 

We first write the evolution equation satisfied by massive charged Dirac fields in $\mathcal{B}$ under the Hamiltonian form
\begin{equation}\label{eqH}
 i \partial_t \psi = H \psi,
\end{equation}
where $\psi$ is a 4-components spinor belonging to the Hilbert space
\[ \mathcal{H} = L^2(\mathbb{R} \times \S^2 ; \mathbb{C}^4),\]
and the Hamiltonian $H$ is given by
\begin{equation}\label{ham}
 H = \Gamma^1 D_x + a(x) D_{\S^2} + b(x) \Gamma^0 + c(x).
\end{equation}
The symbols $D_x$ stands for $-i\partial_x$ whereas $D_{\S^2}$ denotes the Dirac operators on $\S^2$ which, in spherical coordinates, takes the form 
given in (\ref{opeDirac}).
The potentials $a$, $b$ and $c$ are scalar smooth functions given in terms of the metric (\ref{metric}) by (\ref{defpot}).
Finally, the matrices $\Gamma^1$, $\Gamma^2$, $\Gamma^3$ and $\Gamma^0$ appearing in
(\ref{ham}) and (\ref{opeDirac}) are usual $4 \times 4$ Dirac matrices that satisfy the anticommutation relations
\begin{equation}\label{anticom}
 \Gamma^i \Gamma^j + \Gamma^j \Gamma^i = 2 \delta_{ij} Id, \quad \quad \forall i,j = 0,1,2,3.
\end{equation}

We now use the spherical symmetry of the equation to simplify the expression of the Hamiltonian $H$. We can diagonalize the Dirac operator on $\mathbb{S}^2$ into an infinite sum of 
matrix-valued multiplication operators by decomposing it onto a basis of spin-weighted spherical harmonics (see \cite{Gel}). 
Precisely there is a family of spinors $F_k^l$ with the indices $(l,k)$ running in the set $\mathcal{I} = \{ (l,k),\, l- \half \in \mathbb{N}, \, l - |k| \in \mathbb{N} \}$ which forms a Hilbert basis of $L^2(\S^2 ; \mathbb{C}^4)$ with the following 
property. The Hilbert space $\mathcal{H}$ can then be decomposed into the infinite direct sum
\[ \mathcal{H} = \underset{(l,k) \in \mathcal{I}}{\oplus} (L^2(\mathbb{R}_x ; \mathbb{C}^4) \otimes F_k^l) = \underset{(l,k) \in \mathcal{I}}{\oplus} \mathcal{H}_{lk},\]
where $\mathcal{H}_{lk} = L^2(\mathbb{R}_x ; \mathbb{C}^4) \otimes F_k^l$ is identified with $L^2(\mathbb{R};\mathbb{C}^4)$ and more important, we obtain the orthogonal decomposition for the Hamiltonian 
$H$
\[ H = \underset{(l,k) \in \mathcal{I}}{\oplus} H^{lk},\]
with
\begin{equation}\label{hampart}
 H^{lk}:= H_{\vert \mathcal{H}_{lk}} = \Gamma^1 D_x + a_l(x) \Gamma^2 + b(x) \Gamma^0 + c(x),
\end{equation}
and $a_l(x) = -\left( l + \frac{1}{2} \right) a(x)$. Note that the Dirac operator $D_{\S^2}$ has been replaced in the expression of $H^{lk}$ by $-\left(l + \frac{1}{2} \right)\Gamma^2$ thanks to the good properties of the 
spin-weighted spherical harmonics $F_k^l$. The operator $H^{lk}$ is a selfadjoint operator on $\mathcal{H}_{lk}$ with domain $D(H^{lk}) = H^1(\mathbb{R};\mathbb{C}^4)$. Finally we use the following representation for the Dirac matrices 
$\Gamma^1$, $\Gamma^2$ and $\Gamma^0$ appearing in (\ref{hampart})
\begin{equation}\label{matDir}
\Gamma^1 = \begin{pmatrix}
               1 & 0& 0& 0 \\ 0&1&0&0 \\ 0&0&-1&0 \\ 0&0&0&-1 
              \end{pmatrix}, \quad
              \Gamma^2 = \begin{pmatrix}
                          0&0&0&1 \\ 0&0&-1&0 \\ 0&-1&0&0 \\ 1&0&0&0
                         \end{pmatrix}, \quad
                         \Gamma^0 = \begin{pmatrix}
                                     0&0&-i&0 \\ 0&0&0&i \\ i&0&0&0 \\ 0&-i&0&0
                                    \end{pmatrix}.
\end{equation}

In what follows, for all $n \in \N^*$, we shall write $H_n$ instead of $H^{lk}$, where $n = l + \half$.

We now recall the direct scattering results obtained in \cite{Dau, Mel} for the Dirac Hamiltonians $H_n$, restrictions of $H$ on each spin weighted spherical harmonics. 
Using essentially a Mourre theory (see \cite{Mou}), it was shown in \cite{Dau, Mel} that
\[\sigma_{pp}(H_n) = \emptyset, \quad \quad \sigma_{\mathrm{sing}}(H_n) = \emptyset.\]
In other words, the spectrum of $H_n$ is purely absolutely continuous. Thus, on each spin weighted spherical harmonics, massive charged Dirac fields scatter toward the two asymptotic regions at late times and they are expected to obey 
simpler equations there. This is one of the main information encoded in the notion of wave operators that we introduce now.

First, we need to calculate the asymptotics of the potentials. According to (\ref{F}) and (\ref{defx}), the potentials $a$, $b$ and $c$ have the following asymptotics (see \cite{DN3} equations $(3.17)$ and $(3.18)$), as $x \to \pm \infty$.

\begin{lemma}\label{aspot}
There exists constants $a_{\pm}$, $b_{\pm}$, $c_{\pm}$ and $c'_{\pm}$ such that as $x \to \pm \infty$,
 \[a(x) = a_{\pm} e^{\kappa_{\pm} x} + O(e^{3 \kappa_{\pm} x}), \quad a'(x) =  a_{\pm} \kappa_{\pm} e^{\kappa_{\pm} x} + O(e^{3 \kappa_{\pm} x}),\]
 \[b(x) = b_{\pm} e^{\kappa_{\pm} x} + O(e^{3 \kappa_{\pm} x}), \quad b'(x) = b_{\pm} \kappa_{\pm} e^{\kappa_{\pm} x} + O(e^{3 \kappa_{\pm} x}),\]
 \[c(x) = c_{\pm} + c_{\pm}' e^{2\kappa_{\pm} x} + O(e^{4 \kappa_{\pm} x}), \quad c'(x) = 2c_{\pm}' \kappa_{\pm} e^{2\kappa_{\pm} x} + O(e^{4 \kappa_{\pm} x}),\]
where
\[ c_- = \frac{qQ}{r_-} \quad \quad \mathrm{and} \quad \quad c_+ = \frac{qQ}{r_+}.\]
\end{lemma}

\begin{proof}
 For instance, we can give the proof for the potential $a$. We first give asymptotics for the RW variable $x$ when $r \to r_{\pm}$. From (\ref{defx}) we have,
\[ x = \frac{1}{2\kappa_{\pm}} \ln (|r-r_{\pm}|) + C_{\pm} + O(|r-r_{\pm}|), \quad r \to r_{\pm},\]
where
\[C_{\pm} = \ln \left( (r_{\pm} - r_n)^{\frac{1}{2\kappa_n}} (r_{\pm} - r_c)^{\frac{1}{2\kappa_c}} (r_+-r_-)^{\frac{1}{2\kappa_{\mp}}} \right).\]
Thus,
\begin{eqnarray*}
e^{\kappa_{\pm} (x-C_{\pm})} = \sqrt{r-r_{\pm}} ( 1 + O(|r-r_{\pm}|)), \quad r \to r_{\pm}.
\end{eqnarray*}
Finally,
\begin{eqnarray*}
 a(x) = \frac{\sqrt{F(r)}}{r} = \frac{\sqrt{\mp 2 \kappa_{\pm}}}{r_{\pm}} \sqrt{|r-r_{\pm}|} + O(|r-r_{\pm}|^{\frac{3}{2}}), \quad r \to r_{\pm}.
\end{eqnarray*}
\end{proof}

Hence the potentials $a$ and $b$ are short-range when $x \to \pm \infty$ and $c-c_-$ and $c-c_+$ are short-range when $x \to -\infty$ and $x \to +\infty$, respectively.
Thus, the comparison dynamic we choose at the event horizon is the one generated by the Hamiltonian $H_- = \Ga D_x + c_-$ while, at the cosmological horizon, we 
choose the asymptotic dynamic generated by the Hamiltonian $H_+ = \Ga D_x + c_+$. The Hamiltonians $H_-$ and $H_+$ are selfadjoint 
operators on $\mathcal{H}$ and their spectra are exactly the real line, i.e $\sigma(H_-) = \sigma(H_+) = \mathbb{R}$.
The asymptotic velocity operator associated to the asymptotic Hamiltonians $H_\pm$ is simply $\Ga$. Let us denote the projections onto the positive and negative spectrum of the asymptotic velocity operator $\Gamma^1$ by
\[P_{\pm} = \textbf{1}_{\mathbb{R}^{\pm}}(\Gamma^1).\]
As shown in \cite{Dau}, we can use these projections to separate the part of the Dirac fields that scatter toward the event and cosmological horizons.

We are now in position to introduce the wave operators. At the event horizon, we define
\begin{equation}\label{opeondemoins}
 W_{n,(-\infty)}^{\pm} = s - \underset{t \to \pm \infty}{\lim} e^{itH_n} e^{-itH_-} P_{\mp}
\end{equation}
and at the cosmological horizon, we define
\begin{equation}\label{opeondeplus}
 W_{n,(+\infty)}^{\pm} = s - \underset{t \to \pm \infty}{\lim} e^{itH_n} e^{-itH_+} P_{\pm}.
\end{equation}
Finally, the global wave operators are given by
\begin{equation}\label{ondeglob}
W_n^{\pm} = W_{n,(-\infty)}^{\pm} + W_{n,(+\infty)}^{\pm}.
\end{equation}
The main result of \cite{Dau, Mel} is

\begin{theorem}\label{thmMel}
 The wave operators $W_{n,(-\infty)}^{\pm}$, $W_{n,(+\infty)}^{\pm}$ and $W_n^{\pm}$ exist on $\mathcal{H}_{lk}$. Moreover, the global wave operators $W_n^{\pm}$ are isometries on $\mathcal{H}_{lk}$. 
In particular, $W_n^{\pm}$ are asymptotically complete, i.e. $\mathrm{Ran} W_n^{\pm} = \mathcal{H}_{lk}$.
\end{theorem}

Thanks to Theorem \ref{thmMel}, we can define the scattering operator by
\[ S(n) = (W_n^+)^{\star}W_n^-\] 
which is a unitary operator on $\mathcal{H}_{lk}$. 
As mentionned in Section \ref{statement}, this operator can be decomposed as
\[S(n) = \begin{pmatrix}
                   T_L(n) & R(n) \\ L(n) & T_R(n)
                  \end{pmatrix}.\]
where,
\[ T_L(n) = (W_{n,+\infty}^+)^{\star}W_{n,-\infty}^-, \quad \quad T_R(n) = (W_{n,-\infty}^+)^{\star}W_{n,+\infty}^-\]
and
\[R(n) = (W_{n,+\infty}^+)^{\star}W_{n,+\infty}^- \quad \mathrm{and} \quad L(n) = (W_{n,-\infty}^+)^{\star}W_{n,-\infty}^-.\]
It follows from our definition of the wave operators (\ref{opeondemoins}) and (\ref{opeondeplus}) that the previous quantities can be interpreted in terms of transmission and reflection operators
between the event horizon of the black hole $\{ x = -\infty \}$ and the cosmological horizon $\{ x = +\infty \}$. For instance, $T_R(n)$ corresponds to the part of a signal transmitted 
from $\{ x = +\infty \}$ to $\{ x = -\infty \}$ in a scattering process whereas the term $T_L(n)$ corresponds to the part of a signal transmitted from $\{ x = -\infty \}$ to
$\{ x = +\infty \}$. Hence $T_R(n)$ stands for ``transmitted from the right'' and $T_L(n)$ for ``transmitted from the left''. Conversely, $L(n)$ corresponds to the part of a signal reflected from 
$\{ x = -\infty \}$ to $\{ x = -\infty \}$ in a scattering process whereas the term $R(n)$ corresponds to the part of a signal reflected from $\{ x = +\infty \}$ to $\{ x = +\infty \}$.

\section{Simplification of the framework and differential equations for the Jost functions}\label{ED}
\noindent
In this Section we first follow the idea of \cite{DN2}. By a series of simplifications of our model, which finally reduces to the framework studied in \cite{AKM}, we recall and
explain the stationary representation of the scattering operator $S(n)$ expressed in terms of the usual transmission and reflection coefficients (here matrices). In a second time, using \cite{AKM}, we 
obtain second order differential equations satisfied by the Jost functions of our original model.

\subsection{First simplifications} \label{first}
\noindent
In this Section we follow the work of \cite{DN2} Section $4.2$. We recall that the scattering operator $S(n)$ is defined by
\[ S(n) = (W_n^+)^{\star} W_n^-,\]
where the global wave operator $W_n^{\pm}$ are given by (\ref{opeondemoins})-(\ref{ondeglob}).
We introduce the unitary transform $U$ on $\mathcal{H}_{lk}$
\[U = e^{-i \Gamma^1 C^{-}(x)},\quad C^-(x) = \int_{-\infty}^{x} (c(s)-c_-) \, \mathrm ds + c_- x, \]
 and the selfadjoint operators on $\mathcal{H}_{lk}$
\begin{equation}\label{defA}
 A_0 = \Gamma^1 D_x \quad \quad \mathrm{and} \quad \quad A_n = U^{\star} H U = \Gamma^1 D_x + W(x,n).
\end{equation}
Let us denote by $\hat{S}(A_n,A_0)$ the scattering 
operator associated to the operators $A_n$ and $A_0$, i.e.
\[ \hat{S}(A_n,A_0) = (W^+(A_n,A_0))^{\star} W^-(A_n,A_0),\]
where the wave operators are defined, as in \cite{Dau,Dem,Mel}, by
\[W^{\pm}(A_n,A_0) = s- \lim \limits_{t \to \pm \infty} e^{itA_n} e^{-itA_0}.\]
The couple of operators $(A_n,A_0)$ acting on $\mathcal{H}_{lk}$ turns out to fit the framework studied in \cite{AKM}. They are given by $A_0 = \Ga D_x$ and 
$A_n = A_0 + W(x,n)$ where the potential $W(x,n) = e^{i\Ga C^-(x)} (na(x) \Gamma^2 + b(x) \Gamma^0)e^{-i\Ga C^-(x)}$ is the $4 \times 4$ matrix valued function
\[ W(x,n) = \begin{pmatrix}
           0&k(x,n) \\ k(x,n)^{\star} & 0
          \end{pmatrix}, \quad k(x,n) = e^{2i C^-(x)}
          \begin{pmatrix}
           -ib(x) & na(x) \\ -na(x) & ib(x)
          \end{pmatrix},\]
where $n = l + \half$. Here $k(x,n)^{\star}$ denotes the transpose conjugate of the matrix valued function $k(x,n)$. Moreover the entries of $W(.,n)$ are in $L^1(\mathbb{R})$ as in \cite{AKM} (note that our 
potential $W$ is better since it is exponentially decreasing at both ends $x \to \pm \infty$). Thus, we can use the following stationary representation 
of $\hat{S}(A_n,A_0)$. Let us introduce the unitary transform $\mathcal{F}$ on $\mathcal{H}_{lk}$ defined by
\begin{equation}\label{fourier}
 \mathcal{F} \psi(\lambda) = \frac{1}{\sqrt{2\pi}} \int_{\mathbb{R}} e^{-i\Ga x \lambda} \psi(x) \, \mathrm dx,
\end{equation}
then we have (see \cite{AKM} p. 143)
\[\hat{S}(A_n,A_0) = \mathcal{F}^{\star} \hat{S}_0(\lambda,n) \mathcal{F},\]
where the scattering matrix $\hat{S}_0(\lambda,n)$ takes the form
\[\begin{pmatrix}
           \hat{T}_L(\lambda,n) & \hat{R}(\lambda,n) \\ \hat{L}(\lambda,n) & \hat{T}_R(\lambda,n)
          \end{pmatrix}.\]
Here $\hat{T}_L(\lambda,n)$ and $\hat{T}_R(\lambda,n)$ are $2 \times 2$ matrix valued functions which correspond to the usual transmission coefficients of $\hat{S}$ whereas $\hat{L}(\lambda,n)$ and $\hat{R}(\lambda,n)$ 
are $2 \times 2$ matrix valued functions which correspond to the usual reflection coefficients of $\hat{S}$. 

The definition of $\hat{S}_0(\lambda,n)$ obtained so far comes from a time dependent approach to the scattering theory for $(A_n, A_0)$.  
As mentionned in Section \ref{overview} there is an alternative but equivalent definition of $\hat{S}_0(\lambda,n)$ that will be easier to manipulate and is
provided by the stationary approach to the scattering theory for $(A_n, A_0)$.
We thus recall that the reflection and transmission coefficients can be expressed in terms of 
Jost functions which are solutions of Equation (\ref{eqA}) with specific asymptotics at infinity given by (\ref{asJG}) and (\ref{asJD}).
In the following we need some relations between the different quantities of the stationary approach that we summarize here. 
On the Jost functions, the following relation can be proved (see \cite{AKM} eq. (2.5))
\begin{eqnarray}\label{eqJ}
 \hat{F}_{R}(x,\lambda,n)^{\star} \Gamma^1 \hat{F}_R(x,\lambda,n) = \Gamma^1.
\end{eqnarray}
Using the notations introduced in (\ref{notALAR}) we also have some useful relations obtained in \cite{AKM} on the matrix $A_L$ and $A_R$ defined by (\ref{defAL}) and 
(\ref{defAR}). Indeed, for $\lambda \in \mathbb{R}$ and for all $n \in \N^*$,
\begin{eqnarray}\label{lienALJ}
 \hat{F}_{R}(x,\lambda,n)^{\star} \Gamma^1 \hat{F}_L(x,\lambda,n) = \hat{A}_{R}(\lambda,n)^{\star} \Gamma^1 =  \Gamma^1 \hat{A}_{L}(\lambda,n)^{\star}, 
\end{eqnarray}
\begin{eqnarray}\label{eqALR}
 \hat{A}_{R2}(\lambda,n) = \hat{A}_{L3}(\lambda,n)^{\star}, \quad \hat{A}_{R3}(\lambda,n) = -\hat{A}_{L2}(\lambda,n)^{\star},
\end{eqnarray}
\begin{eqnarray}\label{eqALR2}
 \hat{A}_{L1}(\lambda,n) = \hat{A}_{R1}(\lambda,n)^{\star}, \quad \hat{A}_{R4}(\lambda,n) = \hat{A}_{L4}(\lambda,n)^{\star},
\end{eqnarray}
\begin{eqnarray}\label{eqAL}
 \hat{A}_{L1}(\lambda,n)^{\star} \hat{A}_{L1}(\lambda,n) = I_2 + \hat{A}_{L3}(\lambda,n)^{\star} \hat{A}_{L3}(\lambda,n),
\end{eqnarray}
\begin{eqnarray}\label{equnit}
 \hat{A}_{L1}(\lambda,n) \hat{A}_{R1}(\lambda,n) + \hat{A}_{L2}(\lambda,n) \hat{A}_{R3}(\lambda,n) = I_2  = \hat{A}_{R1}(\lambda,n) \hat{A}_{L1}(\lambda,n) + \hat{A}_{R2}(\lambda,n) \hat{A}_{L3}(\lambda,n).
\end{eqnarray}


Finally, we make the link between the stationary representation $S(\lambda,n)$ of the scattering operator $S$ and the scattering matrix $\hat{S}_0(\lambda,n)$. Let us define two unitary transforms $\mathcal{F}_{\pm}$ on $\mathcal{H}$ by
 \begin{equation}\label{defF+}
  F_+ \psi(\lambda) = \mathcal{F} \begin{pmatrix}
           e^{ic_+ x}&0 \\ 0 & e^{-ic_- x}
          \end{pmatrix} \psi(\lambda) =  \frac{1}{\sqrt{2\pi}}  \int_{\mathbb{R}} \begin{pmatrix}
           e^{-ix\lambda + ic_+ x} & 0 \\ 0 & e^{ix \lambda-ic_- x}
          \end{pmatrix} \psi(x) \, \mathrm dx,
 \end{equation}
and
\begin{equation}\label{defF_-}
 F_- \psi(\lambda) = \mathcal{F} \begin{pmatrix}
           e^{ic_- x}&0 \\ 0 & e^{-ic_+ x}
          \end{pmatrix} \psi(\lambda) =  \frac{1}{\sqrt{2\pi}}  \int_{\mathbb{R}} \begin{pmatrix}
           e^{-ix\lambda + ic_- x} & 0 \\ 0 & e^{ix \lambda-ic_+ x}
          \end{pmatrix} \psi(x) \, \mathrm dx,
\end{equation}
and recall the following Proposition given in \cite{DN2}.

\begin{prop}\label{LienSS}
 The scattering operator $S(n)$ has the following stationary representation. If $F_{\pm}$ are the unitary transforms defined in (\ref{defF+}) and (\ref{defF_-}), then
 \[ S(n) = F_+^{\star} S(\lambda,n) F_-,\]
 where the $4 \times 4$ scattering matrix $S(\lambda,n)$ is given by
 \[ S(\lambda,n) = \begin{pmatrix}
           e^{-i\beta} \hat{T}_L(\lambda,n) & e^{-2i\beta} \hat{R}(\lambda,n) \\ \hat{L}(\lambda,n) & e^{-i\beta} \hat{T}_R(\lambda,n)
          \end{pmatrix},\]
where
\begin{eqnarray}\label{beta}
&& \beta = \int_{-\infty}^{0} (c(s)-c_-) \, \mathrm ds + \int_{0}^{+\infty} (c(s)-c_+) \, \mathrm ds
\end{eqnarray}
 and the quantities $\hat{T}_L$, $\hat{T}_R$, $\hat{L}$ and $\hat{R}$ are the $2 \times 2$ matrices corresponding to the transmission and reflection matrices of
 the pair $(A_n,A_0)$.
\end{prop}

In the next Subsection we use some integral representations of the blocks of the Faddeev matrix from the right defined by
\[\hat{M}_R(x,\lambda,n) = \hat{F}_R(x,\lambda,n) e^{-i\Ga \lambda x}.\]
It is easy to show from (\ref{eqA}) that $\hat{M}_R(x,\lambda,n)$ satisfies the integral equation
\begin{equation}\label{eqinteML}
 \hat{M}_R(x,\lambda,n) = I_4 + i \Ga \int_{x}^{+\infty} e^{-i\Ga \lambda(y-x)} W(y,n) \hat{M}_R(y,\lambda,n) e^{i\Ga \lambda(y-x)} \, \mathrm dy,
\end{equation}
Actually, using once again $2 \times 2$ block matrix notation 
\[\hat{M}_R(x,\lambda,n) = \begin{pmatrix}
           \hat{M}_{R1}(x,\lambda,n) & \hat{M}_{R2}(x,\lambda,n) \\ \hat{M}_{R3}(x,\lambda,n) & \hat{M}_{R4}(x,\lambda,n)
          \end{pmatrix},\]
and iterating (\ref{eqinteML}) once, we get the uncoupled system of integral equations for the blocks $\hat{M}_{Ri}(x,\lambda,n)$, $i \in \{1,2,3,4\}$. For instance,
\begin{equation}\label{eqinteMRun}
 \hat{M}_{R1}(x,\lambda,n) = I_{2} + \int_{-\infty}^{x} \left( \int_{-\infty}^{y} e^{2i\lambda (t-y)} k(y,n)k(t,n)^{\star} \hat{M}_{R1}(t,\lambda,n) \, \mathrm dt \right) \, \mathrm dy
\end{equation}
and
\begin{equation}\label{eqinteMRdeux}
 \hat{M}_{R2}(x,\lambda,n) = i \int_{-\infty}^{x} e^{2i\lambda (x-y)} k(y,n) \, \mathrm dy  + \int_{-\infty}^{x} \left( \int_{-\infty}^{y} e^{2i\lambda (x-y)} k(y,n)k(t,n)^{\star} \hat{M}_{R2}(t,\lambda,n) \, \mathrm dt \right) \, \mathrm dy.
\end{equation}
Similar results hold for the Faddeev matrix from the left $\hat{M}_L$ which is defined by
\[\hat{M}_L(x,\lambda,n) = \hat{F}_L(x,\lambda,n) e^{-i\Ga \lambda x}.\]
The equation (\ref{eqinteMRun}) is an integral equation of Volterra type that can be solved by iteration. Moreover we can deduce from this method some important estimates.

\begin{lemma}\label{estimM}
 For all $x_0 \in \mathbb{R}$, there exists a constant $C \in \mathbb{R}$ such that for all $x \in \R$, for all $n \in \mathbb{N}$, we have
 \[ \Vert \hat{M}_{Ri}(x,\lambda,n) \Vert \leq C \exp \left( n \int_{-\infty}^{x} a(s) \, \mathrm ds \right),\quad \forall i \in \{1,2,3,4\},\]
and
 \[ \Vert \hat{M}_{Li}(x,\lambda,n) \Vert \leq C \exp \left( n \int_{x}^{+\infty} a(s) \, \mathrm ds \right),\quad \forall i \in \{1,2,3,4\}.\]
\end{lemma}

\begin{proof}
 We can for instance prove Lemma \ref{estimM} for $\hat{M}_{R1}(x,\lambda,n)$ since the proof for the other blocks is similar. To obtain the estimate of Lemma \ref{estimM} we use an iterative method.
 We use the integral equation (\ref{eqinteMRun}) and we define
$\hat{M}_{R1}^0 (x,\lambda,n) = I_2$ and 
 \[ \hat{M}_{R1}^{k+1}(x,\lambda,n) = \int_{-\infty}^{x} \left( \int_{-\infty}^{y} e^{2i\lambda (t-y)} k(y,n)k(t,n)^{\star} \hat{M}_{R1}^k(t,\lambda,n) \, \mathrm dt \right) \, \mathrm dy, \quad \forall k \in \mathbb{N}.\]
 We can prove by induction the estimate
\[ \Vert \hat{M}_{R1}^{k} (x,\lambda,n) \Vert \leq \frac{1}{(2k)!} \left( \int_{-\infty}^x \Vert  k(s,n) \Vert \, \mathrm ds \right)^{2k}, \quad \forall k \in \mathbb{N}.\]
Thus, we can write
\begin{eqnarray}\label{serieM_R1}
 \hat{M}_{R1}(x,\lambda,n) = \sum_{k=0}^{+\infty} \hat{M}_{R1}^k (x,\lambda,n)
\end{eqnarray}
and we obtain
\[  \Vert \hat{M}_{R1} (x,\lambda,n) \Vert \leq \exp \left( \int_{-\infty}^x  \Vert k(s,n) \Vert  \, \mathrm ds \right).\]
We recall that
\[ k(x,n) = e^{2i C^-(x)}
          \begin{pmatrix}
           -ib(x) & na(x) \\ -na(x) & ib(x)
          \end{pmatrix}.\]
Therefore $\Vert k(x,n) \Vert = na(x) + b(x)$ (we choose $\Vert . \Vert = \Vert . \Vert_1$) and
\[  \Vert \hat{M}_{R1} (x,\lambda,n) \Vert \leq \exp \left( \int_{-\infty}^{+\infty}  b(s)  \, \mathrm ds \right) \exp \left( n \int_{-\infty}^x a(s)  \, \mathrm ds \right).\]
\end{proof}

%
%

\subsection{Second order differential equations for the Jost functions}
\noindent
The aim of this Subsection is to find second order differential equations satisfied by the components of Jost functions. To do this, we first give differential equations satisfied by the blocks of 
the Faddeev matrix using the integral equations $(2.13)-(2.16)$ and $(2.18)-(2.21)$ of \cite{AKM}. Next, using the link between the Faddeev matrices 
of the reduced operator $A$ and the Jost functions of the Hamiltonian $H$, we obtain second order differential equations for the Jost functions.

We first prove easily the following Proposition:

\begin{prop}\label{eqbloc}
 The blocks of the Faddeev matrices satifsfy the following differential equations in the variable $x$.
 \begin{enumerate}
  \item The blocks $\hat{M}_{L1}(.,\lambda,n)$ and $\hat{M}_{R1}(.,\lambda,n)$ satisfy the differential equation
\begin{equation}\label{eqMR1}
 A''(x,\lambda,n) + (2i\lambda  I_{2} - k'(x,n)k(x,n)^{-1}) A'(x,\lambda,n) - k(x,n) k(x,n)^{\star} A(x,\lambda,n) = 0.
\end{equation}
  \item The blocks $\hat{M}_{L2}(.,\lambda,n)$ and $\hat{M}_{R2}(.,\lambda,n)$ satisfy
\begin{eqnarray*}\label{eqMR2}
  A''(x,\lambda,n) &-& (2i\lambda  I_{2} + k'(x,n)k(x,n)^{-1}) A'(x,\lambda,n) \\
  &+& (2i\lambda k'(x,n)k(x,n)^{-1}-k(x,n) k(x,n)^{\star}) A(x,\lambda,n) = 0.
\end{eqnarray*}
  \item The blocks $\hat{M}_{L3}(.,\lambda,n)$ and $\hat{M}_{R3}(.,\lambda,n)$ satisfy
\begin{eqnarray*}\label{eqMR3}
  A''(x,\lambda,n) &+& (2i\lambda  I_{2} - k'(x,n)^{\star}(k(x,n)^{\star})^{-1}) A'(x,\lambda,n)\\
  & -& (2i\lambda k'(x,n)^{\star} (k(x,n)^{\star})^{-1}+k(x,n)^{\star} k(x,n)) A(x,\lambda,n) = 0.
\end{eqnarray*}
  \item The blocks $\hat{M}_{L4}(.,\lambda,n)$ and $\hat{M}_{R4}(.,\lambda,n)$ satisfy
  \[A''(x,\lambda,n) - (2i\lambda  I_{2} + k'(x,n)^{\star}(k(x,n)^{\star})^{-1}) A'(x,\lambda,n) - k(x,n)^{\star} k(x,n) A(x,\lambda,n) = 0.\]
 \end{enumerate}
\end{prop}

\begin{remark}
 \begin{enumerate}
  \item The analysis is the same for the Jost function from the left and the Jost function from the right. We choose to work with the Jost function from the right.
  \item There is an analogy between the differential equations for the blocks $\hat{M}_{R1}$ and $\hat{M}_{R4}$. Indeed, the corresponding integral equations of \cite{AKM} are the same changing 
  $\lambda$ in $-\lambda$ and $k$ in $k^{\star}$. Thus in the following, we just give results for the block $\hat{M}_{R1}$. The same analogy is true for the blocks $\hat{M}_{R2}$ and $\hat{M}_{R3}$ and we 
  just give results for the block $\hat{M}_{R2}$ in the following.
 \end{enumerate}
\end{remark}

Using Proposition \ref{eqbloc}, we easily obtain second order differential equations for the coefficients of the Faddeev matrices. After that we use the explicit link between the Faddeev 
matrices and the Jost fonctions corresponding to the Hamiltonian $H$. The Jost functions of the Hamiltonian $A$ are given by
\[\hat{F}_R(x,\lambda,n) = \hat{M}_R(x,\lambda,n) e^{i \lambda \Gamma^1 x}.\]
We recall that
\[ H = e^{-i \Gamma^1 C^{-}(x)} A e^{i \Gamma^1 C^{-}(x)}.\]
Hence, the Jost function from the right of the Hamiltonian $H$ is defined by
\begin{eqnarray}\label{lienFMR}
 F_{R} (x,\lambda,n) = e^{-i \Gamma^1 C^{-}(x)} \hat{M}_R(x,\lambda,n) e^{i \lambda \Gamma^1 x}
\end{eqnarray}
and similarly the Jost function from the left of the Hamiltonian $H$ is defined by
\begin{eqnarray}\label{lienFML}
F_{L} (x,\lambda,n) = e^{-i \Gamma^1 C^{-}(x)} \hat{M}_L(x,\lambda,n) e^{i \lambda \Gamma^1 x}.
\end{eqnarray}
We use the following notation
\[ F_R(x,\lambda,n) = \begin{pmatrix}
           F_{R1}(x,\lambda,n) & F_{R2}(x,\lambda,n) \\ F_{R3}(x,\lambda,n) & F_{R4}(x,\lambda,n)
          \end{pmatrix} \quad \textrm{where} \quad F_{Ri}(x,\lambda,n) = \begin{pmatrix}
           F_{Ri,1}(x,\lambda,n) & F_{Ri,2}(x,\lambda,n) \\ F_{Ri,3}(x,\lambda,n) & F_{Ri,4}(x,\lambda,n)
          \end{pmatrix}\]
and the corresponding notations for the Jost function from the left.
Using the previous equalities and Proposition \ref{eqbloc} we prove the following Proposition.

\begin{prop}\label{eqJost1}
The coefficients of the Faddeev matrices satisfy the following coupled equations in the variable $x$.
 \begin{enumerate}
  \item For $i \in \{1,2\}$ and $(j,k) \in \{(1,3),(3,1),(2,4),(4,2)\}$ there exists a function $f_{i,j,k}(x,\lambda,n)$ such that the component $F_{Ri,j}$ of the Jost function 
  from the right and the component $F_{Li,j}$ of the Jost function from the left satisfy the coupled differential equation
\begin{eqnarray*}
 &&u''(x,\lambda,n) - \frac{a'(x)}{a(x)} u'(x,\lambda,n) \\
&&+ \left( ic'(x) + (c(x)-\lambda)^2 - i(c(x)-\lambda) \frac{a'(x)}{a(x)} - (n^{2}a(x)^2 + b(x)^2) \right)u(x,\lambda,n) = f_{i,j,k}(x,\lambda,n),
\end{eqnarray*}
where,
\[ f_{i,j,k} = c_{1,i,j} F_{Ri,j}'(x,\lambda,n) + c_{2,i,k} F_{Ri,k}'(x,\lambda,n) + c_{3,i,j} F_{Ri,j}(x,\lambda,n) + c_{4,i,k} F_{Ri,k}(x,\lambda,n),\]
with,
\[c_{r,i,t}(x,\lambda,n) = O \left( \frac{1}{n} \right), \quad n \to \infty, \quad \forall r \in \{1,2,3,4\} \quad \mathrm{and} \quad t \in \{j,k\}.\]
  \item For $i \in \{3,4\}$ and $(j,k) \in \{(1,3),(3,1),(2,4),(4,2)\}$ there exists a function $f_{i,j,k}(x,\lambda,n)$ such that the component $F_{Ri,j}$ of the Jost function 
  from the right and the component $F_{Li,j}$ of the Jost function from the left satisfy the coupled differential equation
\begin{eqnarray*}
 &&u''(x,\lambda,n) - \frac{a'(x)}{a(x)} u'(x,\lambda,n) \\
&&+ \left( -ic'(x) + (c(x)-\lambda)^2 + i(c(x)-\lambda) \frac{a'(x)}{a(x)} - (n^{2}a(x)^2 + b(x)^2) \right)u(x,\lambda,n) = f_{i,j,k}(x,\lambda,n),
\end{eqnarray*}
where,
\[ f_{i,j,k} = c_{1,i,j} F_{Ri,j}'(x,\lambda,n) + c_{2,i,k} F_{Ri,k}'(x,\lambda,n) + c_{3,i,j} F_{Ri,j}(x,\lambda,n) + c_{4,i,k} F_{Ri,k}(x,\lambda,n),\]
with,
\[c_{r,i,t}(x,\lambda,n) = O \left( \frac{1}{n} \right), \quad n \to \infty, \quad \forall r \in \{1,2,3,4\} \quad \mathrm{and} \quad t \in \{j,k\}.\]
 \end{enumerate}
\end{prop}

\begin{proof}
 To prove it, for instance for the coefficient $F_{R1,1}$, we use the equality
 \[\hat{M}_{R1} (x,\lambda,n) = e^{i(C^{-}(x) -\lambda x)} F_{R1}(x,\lambda,n).\]
 and the differential equation obtained for the first block of the Faddeev matrix in Proposition \ref{eqbloc}. We easily obtain
 \begin{eqnarray*}
 &&F_{R1,1}''(x,\lambda,n) - \frac{a'(x)}{a(x)} F_{R1,1}'(x,\lambda,n) \\
&&+ \left( ic'(x) + (c(x)-\lambda)^2 - i(c(x)-\lambda) \frac{a'(x)}{a(x)} - (n^{2}a(x)^2 + b(x)^2) \right) F_{R1,1}(x,\lambda,n) = f_{1,1,3}(x,\lambda,n) 
\end{eqnarray*}
where
\begin{eqnarray}\label{deffij}
f_{1,1,3}(x,\lambda,n) &=& c_1(x,\lambda,n) F_{R1,1}'(x,\lambda,n) + c_2(x,\lambda,n) F_{R1,3}'(x,\lambda,n)\\
& & +c_3(x,\lambda,n) F_{R1,1}(x,\lambda,n)+ c_4(x,\lambda,n) F_{R1,3}(x,\lambda,n)
\end{eqnarray}
and
\begin{eqnarray}\label{defci}
c_1(x,\lambda,n) &=& \frac{a(x)^2b(x)b'(x)- a'(x)a(x)b(x)^2}{(n^2 a(x)^2+b(x)^2)a(x)^2} \\
c_2(x,\lambda,n) &=& - \frac{in(-a(x)b'(x) + a'(x)b(x))}{n^2 a(x)^2 + b(x)^2}\\
c_3(x,\lambda,n) &=& i(c(x)-\lambda)\left(\frac{a(x)^2b(x)b'(x)- a'(x)a(x)b(x)^2}{(n^2 a(x)^2+b(x)^2)a(x)^2} \right)\\
c_4(x,\lambda,n) &=& (c(x)-\lambda) \left(\frac{n(-a(x)b'(x) + a'(x)b(x))}{n^2 a(x)^2 + b(x)^2} \right).
\end{eqnarray}
We prove the statement for the other components of the Jost functions similarly. For instance, we note that 
\begin{eqnarray*}
f_{1,3,1}(x,\lambda,n) &=& c_1(x,\lambda,n) F_{R1,3}'(x,\lambda,n) + c_2(x,\lambda,n) F_{R1,1}'(x,\lambda,n)\\
& & +c_3(x,\lambda,n) F_{R1,3}(x,\lambda,n)+ c_4(x,\lambda,n) F_{R1,1}(x,\lambda,n).
\end{eqnarray*}
\end{proof}


\begin{remark}
\begin{enumerate}
 \item If $b=c=0$ we obtain that $F_{Ri,j}$ and $F_{Li,j}$, where $i \in \{1,2\}$ and $j \in \{1,2,3,4\}$, satisfy the uncoupled differential equation
\[u''(x,\lambda,n) - \frac{a'(x)}{a(x)} u'(x,\lambda,n) + \left( \lambda^2 + i \lambda \frac{a'(x)}{a(x)}-n^2 a(x)^2 \right) u(x,\lambda,n) =0, \]
whereas $F_{Ri,j}$ and $F_{Li,j}$, where $i \in \{3,4\}$ and $j \in \{1,2,3,4\}$, satisfy the differential equation
\[u''(x,\lambda,n) - \frac{a'(x)}{a(x)} u'(x,\lambda,n) + \left( \lambda^2 - i \lambda \frac{a'(x)}{a(x)} -n^2 a(x)^2 \right) u(x,\lambda,n) =0.\]
 These are the equations $(3.21)$ and $(3.22)$ obtained in the massless and uncharged case studied in \cite{DN}.
 \item It is important to note that in the massive case the components of Jost functions satisfy differential equations coupled two by two whereas in 
 the massless case these components satisfy independant ordinary differential equations. This structure of the differential equations will be fundamental in Section 
 \ref{partDuh}.
\end{enumerate}
\end{remark}

\section{Complexification of the angular momentum and asymptotics for large angular momentum}\label{asym}
\noindent
In this Section, we allow the angular momentum to be complex. 
As in Section 4 of \cite{DN}, we shall obtain the asymptotics of the Jost functions and of the matrix $\hat{A}_L(\lambda,z)$ when $|z| \to \infty$, $z \in \mathbb{C}$. Here we shall crucially use the exponential decay of the potentials $a(x)$ and $b(x)$ at both horizons and thus the asymptotically hyperbolic nature of the geometry. The main tools to obtain these asymptotics 
are a simple change of variable $X = g(x)$, called the Liouville transformation and a perturbative argument. 

\subsection{The Liouville variable and the Bessel equations}
\noindent
We follow the strategy of \cite{Cha, Cha2, DN}. Considering the differential equations given in Proposition \ref{eqJost1}, we use a Liouville transformation, i.e. a change of variable $X = g(x)$, that transforms these equations into singular Sturm-Liouville equations in which the complex angular momentum $z$ becomes the spectral parameter.

Let us define precisely this Liouville transformation. We denote
\[ X = g(x) = \int_{-\infty}^x a(t) \, \mathrm dt.\]
Clearly, $g = \mathbb{R} \rightarrow ]0,A[ $ is a $C^1$-diffeomorphism where
\[ A = \int_{\mathbb{R}} a(t) \, \mathrm dt.\]
For the sake of simplicity, we denote $h = g^{-1}$ the inverse diffeomorphism of $g$ and we use the notations $f'(X) = \frac{\partial f}{\partial X}(X)$, $F_L(X,\lambda,z) = F_L(h(X),\lambda,z)$ and $F_R(X,\lambda,z) = F_R(h(X),\lambda,z)$.\\
We begin with an elementary Lemma which states that, in the variable $X$, the Jost functions satisfy Sturm-Liouville equations with potentials having quadratic singularities at 
the boundaries.

\begin{lemma}\label{eqJostX}
For $i \in \{1,2\}$ and $(j,k) \in \{(1,3),(3,1),(2,4),(4,2)\}$ the component $F_{Ri,j}$ of the Jost function 
  from the right and the component $F_{Li,j}$ of the Jost function from the left satisfy the coupled differential equation
\[ u''(X,\lambda,z) + q(X,\lambda) u(X,\lambda,z) = z^2 u(X,\lambda,z) + \frac{f_{i,j,k}(X,\lambda,z)}{a(X)^2},\]
 whereas for $i \in \{3,4\}$ and $(j,k) \in \{(1,3),(3,1),(2,4),(4,2)\}$ the component $F_{Ri,j}$ of the Jost function 
  from the right and the component $F_{Li,j}$ of the Jost function from the left satisfy the coupled differential equation 
\[ u''(X,\lambda,z) + \overline{q(X,\lambda)} u(X,\lambda,z) = z^2 u(X,\lambda,z) + \frac{f_{i,j,k}(X,\lambda,z)}{a(X)^2},\]
 where
 \[q(X,\lambda) = i\frac{c'(X)}{a(X)^2} + \frac{(c(X)-\lambda)^2}{a(X)^2} - i(c(X)-\lambda) \frac{a'(X)}{a(X)^3} - \frac{b(X)^2}{a(X)^2}.\]
Here the functions $f_{i,j,k}$ are the functions appearing in Proposition \ref{eqJost1}.
Moreover
 \begin{equation}\label{quadrasing}
 q(X,\lambda) = \frac{\omega_-}{X^2} + q_-(X,\lambda), \quad \mathrm{with} \quad \omega_- = \frac{(c_--\lambda)^2}{\kappa_-^2} -i \frac{(c_--\lambda)}{\kappa_-} \quad \mathrm{and} \quad q_-(X,\lambda) = O(1), \quad X \to 0,
\end{equation}
and
 \begin{equation}\label{quadrasingA}
 q(X,\lambda) = \frac{\omega_+}{(A-X)^2} + q_+(X,\lambda), \quad \mathrm{with} \quad \omega_+ = \frac{(c_+-\lambda)^2}{\kappa_+^2} -i \frac{(c_+-\lambda)}{\kappa_+} \quad \mathrm{and} \quad q_+(X,\lambda) = O(1), \quad X \to A.
\end{equation}
\end{lemma}

\begin{remark}
 In the case $b=c=0$ we obtain that $F_{Ri,j}(X,\lambda,z)$ and $F_{Li,j}(X,\lambda,z)$, where $i \in \{1,2\}$ and $j \in \{1,2,3,4\}$, satisfy the uncoupled Sturm-Liouville equation
\[u''(X,\lambda,z) + \left( \frac{\lambda^2}{a(X)^2} + i\lambda \frac{a'(X)}{a(X)^3} \right) u(X,\lambda,z) = z^2 u(X,\lambda,z) \]
whereas $F_{Ri,j}(X,\lambda,z)$ and $F_{Li,j}(X,\lambda,z)$, where $i \in \{3,4\}$ and $j \in \{1,2,3,4\}$, satisfy the Sturm-Liouville equation
\[u''(X,\lambda,z) + \left( \frac{\lambda^2}{a(X)^2} - i\lambda \frac{a'(X)}{a(X)^3} \right) u(X,\lambda,z) = z^2 u(X,\lambda,z). \]
 These are the equations $(4.5)$ and $(4.6)$ obtained in Lemma 4.1 of \cite{DN}. Moreover, we obtain that
 \[ \omega_- = \frac{\lambda}{\kappa_-^2} + i \frac{\lambda}{\kappa_-}\]
and
 \[ \omega_+ = \frac{\lambda}{\kappa_+^2} + i \frac{\lambda}{\kappa_+}\]
 which are the equations $(4.7)$ and $(4.8)$ of Lemma 4.1 of \cite{DN}.
\end{remark}

\begin{proof}
 Since the proofs are the same for the other components, we just prove Proposition \ref{eqJostX} for $F_{R1,1}$. Using the Liouville transformation and 
 the notations $a(X) = a(h(X))$, $b(X) = b(h(X))$ and $c(X) = c(h(X))$, we obtain that $F_{R1,1}$ satisfies the Sturm-Liouville equation:
\[ u''(X,\lambda,z) + q(X,\lambda) u(X,\lambda,z) = z^2 u(X,\lambda,z) + \frac{f_{1,1,3}(X,\lambda,z)}{a(X)^2}\]
where
\[q(X,\lambda) = i\frac{c'(X)}{a(X)^2} + \frac{(c(X)-\lambda)^2}{a(X)^2} - i(c(X)-\lambda) \frac{a'(X)}{a(X)^3} - \frac{b(X)^2}{a(X)^2}\]
and
\begin{eqnarray*}
 f_{1,1,3}(X,\lambda,z) &=& \left(\frac{a(X)^2b(X)b'(X)- a'(X)a(X)b(X)^2}{(z^2 a(X)^2+b(X)^2)a(X)} \right) F_{R1,1}'(X,\lambda,z)\\
&& + i(c(X)-\lambda)\left(\frac{a(X)^2b(X)b'(X)- a'(X)a(X)b(X)^2}{(z^2 a(X)^2+b(X)^2)a(X)^2} \right)F_{R1,1}(X,\lambda,z)\\
& &- \frac{iz(-a(X)b'(X) + a'(X)b(X))a(X)}{z^2 a(X)^2 + b(X)^2} F_{R1,3}'(X,\lambda,z)\\
&&+\frac{z(-a(X)b'(X) + a'(X)b(X))}{z^2 a(X)^2 + b(X)^2}(c(X)-\lambda) F_{R1,3}(X,\lambda,z).
\end{eqnarray*}
To show that the potential $q$ has quadratic singularities given by $(\ref{quadrasing})$ we use the following Lemma:

\begin{lemma}\label{aspotX}
When $X \to 0$, the potentials satisfy:
 \[a(X) = \kappa_- X + O(X^3), \quad a'(X) = \kappa_-^2 X + O(X^3),\]
 \[b(X) = \frac{b_- \kappa_-}{a_-} X + O(X^3), \quad b'(X) = \frac{b_- \kappa_-^2}{a_-} X + O(X^3),\]
 \[c(X) = c_- + \frac{c_{-}' \kappa_-}{a_-} X^2 + O(X^4), \quad c'(X) = \frac{c_{-}' \kappa_-^2}{a_-} X^2 + O(X^4).\]
\end{lemma}

\begin{proof}
 We know that
 \[ X = \int_{-\infty}^x a(t) \, \mathrm dt = \frac{a_-}{\kappa_-} e^{\kappa_- x} + O(e^{3\kappa_- x}).\]
Thus, $e^{\kappa_- x} = O(X)$ and according to Lemma \ref{aspot}
\[ a(X) = \kappa_- X + O(X^3).\]
Similarly, we obtain 
 \[a'(X) = \kappa_-^2 X + O(X^3).\]
 Using these asymptotics and Lemma \ref{aspot}, we obtain similarly the results for the potentials $b$ and $c$.
\end{proof}
\noindent
Finally, using Lemma \ref{aspotX}, we easily obtain
 \begin{eqnarray*}
  q(X,\lambda) &=& i\frac{c'(X)}{a(X)^2} + \frac{(c(X)-\lambda)^2}{a(X)^2} - i(c(X)-\lambda) \frac{a'(X)}{a(X)^3} - \frac{b(X)^2}{a(X)^2}\\
&=&  \frac{(c_--\lambda)^2 - i(c_- - \lambda) + O(X^2)}{(\kappa_- X +O(X^2))^2} \\
&=& \frac{\omega_-}{X^2} + q_-(X,\lambda).
 \end{eqnarray*}
The proof is the same when $X \to A$.
\end{proof}

In the next Subsections, we shall solve the equations given in Lemma \ref{eqJostX} by a perturbative argument. We need the following Lemma to give a sense of the 
perturbative terms and to compute the asymptotics of the Jost functions from the right (respectively the left) at $X \to 0$ (respectively $X \to A$).

\begin{lemma}\label{termesourceb}
 For each fixed $z \in \mathbb{C}$, for $i \in \{1,2\}$ and for $(j,k) \in \{(1,3),(2,4),(3,1),(4,2)\}$,
 \[h_{i,j,k}(X,\lambda,z) = \frac{f_{i,j,k}(X,\lambda,z)}{a(X)^2} = O(1)\, \quad \text{when} \quad X \to 0 \quad \text{and} \quad X \to A.\]
\end{lemma}
\begin{proof}
This is an easy consequence of Lemma \ref{aspotX} and the fact that the components of 
the Jost functions are bounded when $X \to 0$ and $X \to A$.
\end{proof}

The homogeneous parts of the differential equations given in Lemma \ref{eqJostX} are simple modified Bessel equations (see 5.4.11 in \cite{Leb})
\[ u'' + \frac{1-2\alpha}{X} u' + \left( (\beta \gamma X^{\gamma-1})^2 + \frac{\alpha^2 - \nu^2 \gamma^2}{X^2} \right) u = 0.\]
In our case the homogeneous equations are
\[ u'' + \frac{\omega_-}{X^2} u = z^2 u\]
and
\[ u'' + \frac{\overline{\omega_-}}{X^2} u = z^2 u,\]
where
\[\omega_{-} = \frac{(c_--\lambda)^2}{\kappa_-^2} -i \frac{(c_--\lambda)}{\kappa_-}.\]
Thus, we choose $\alpha = \frac{1}{2}$, $\gamma = 1$, $\beta = iz$ and
\[\nu_- = \frac{1}{2} - i \frac{(\lambda-c_-)}{\kappa_-}, \quad (\nu_- \notin \mathbb{Z}),\]
for the first equation and
\[\mu_- = \overline{\nu_-} = \frac{1}{2} + i \frac{(\lambda-c_-)}{\kappa_-},\]
for the second one.

\begin{remark}
 If $b=c=0$ we obtain 
\[\nu_- = \frac{1}{2} - i \frac{\lambda}{\kappa_-}\]
and
\[\mu_- = \frac{1}{2} + i \frac{\lambda}{\kappa_-}\]
which are the choices of \cite{DN}.
\end{remark}

\subsection{Estimates on the Green kernels and on the Jost functions}\label{Green}
\noindent
Singular Sturm-Liouville equations such as in Proposition \ref{eqJostX} have been studied in \cite{FY} by Freiling and Yurko. 
We follow the spirit of \cite{DN} but, because of the fact that we have no series expansion for our Jost functions, our proof is a little bit different.
We use the fact that modified Bessel functions form a fundamental system of solutions of the equation
\[u'' + \frac{\omega}{X^2} u = z^2 u,\]
as well as the known asymptotics of these modified Bessel functions given in \cite{Leb} and good estimates of the Green kernel to obtain the asymptotics of the Jost functions as $|z| \to \infty$. 
We recall the definition of the modified Bessel functions given in \cite{Leb}:
\begin{equation}\label{defI}
 I_{\nu}(z) = \sum_{k=0}^{\infty} \frac{\left( \frac{z}{2} \right)^{\nu + 2k}}{\Gamma(k+\nu+1) k!}, \quad \vert z \vert < \infty, \quad \vert \arg(z) \vert < \pi, \quad \nu \in \mathbb{C},
\end{equation}
\begin{equation}\label{defK}
  K_{\nu}(z) = \frac{\pi}{2} \frac{I_{-\nu}(z) - I_{\nu}(z)}{\sin{\nu \pi}}, \quad \vert \arg(z) \vert < \pi, \quad \nu \notin \mathbb{Z}.
\end{equation}
The aim of this Section is to prove the following Theorem.

\begin{theorem}\label{asJostdroite}
 We set
\[ \alpha_R(z) = \left( \frac{z}{2} \right)^{\nu_-} \left( \frac{\kappa_-}{a_-} \right)^{i \left( \frac{\lambda-c_-}{\kappa_-} \right)} \Gamma(1-\nu_-)\]
and
\[ \beta_R(z) = \left( \frac{z}{2} \right)^{\mu_-} \left( \frac{\kappa_-}{a_-} \right)^{-i \left( \frac{\lambda-c_-}{\kappa_-} \right)} \Gamma(1-\mu_-).\]
For large $z$ in the complex plane (see the following remark), the Jost functions from the right satisfy the following asymptotics, uniformly on each compact subset of $]0,A[$,
\[ \left| \left| F_{R1}(X,\lambda,z) - \alpha_R(z) \sqrt{X} I_{-\nu_-}(zX) \begin{pmatrix}  1 & 0 \\ 0 & 1   \end{pmatrix} \right|\right| = O \left( \frac{e^{| \mathrm{Re}(z)| X}}{\vert z \vert} \right),\]
\[ \left| \left| F_{R2}(X,\lambda,z) -i \beta_R(z) \sqrt{X} I_{\nu_-}(zX) \begin{pmatrix}  0 &1 \\ -1 & 0   \end{pmatrix} \right| \right| = O \left( \frac{e^{| \mathrm{Re}(z) | X}}{|z|} \right) ,\]
\[ \left| \left| F_{R3}(X,\lambda,z) -i \overline{\beta_R(\bar{z})} \sqrt{X} I_{\mu_-}(zX) \begin{pmatrix}  0 &1 \\ -1 & 0   \end{pmatrix} \right| \right| = O \left( \frac{e^{| \mathrm{Re}(z) | X}}{|z|} \right) ,\]
\[ \left| \left| F_{R4}(X,\lambda,z) - \overline{\alpha_R(\bar{z})} \sqrt{X} I_{-\mu_-}(zX) \begin{pmatrix}  1 & 0 \\ 0 & 1   \end{pmatrix} \right|\right| = O \left( \frac{e^{| \mathrm{Re}(z)| X}}{\vert z \vert} \right).\]
\end{theorem}

\begin{remark}
We have to be carefull in the definition of the asymptotics in the whole complex plane.
We recall the asymptotic given in \cite{Leb}: for large $z$ such that $\vert \text{arg}(z) \vert \leq \pi - \delta$, $\delta > 0$,
\begin{equation}\label{asInu}
 I_{\nu_-}(z) = \frac{e^z}{(2\pi z)^\frac{1}{2}} \left( 1 + O \left( \frac{1}{z} \right) \right) + \frac{e^{-z +\mathrm{sg}(\mathrm{Im}(z)) i\pi (\nu_- + \frac{1}{2})}}{(2\pi z)^{\frac{1}{2}}} \left(1 + O\left( \frac{1}{z} \right) \right).
\end{equation}
where $\mathrm{sg}$ is the sign function defined by $\mathrm{sg}(x) = 1$ if $x > 0$, $0$ if $x = 0$ and $-1$ if $x < 0$. 
Thus, this estimate is true in the complex plane except near the axis $\R^-$. However, since the application $z \mapsto \alpha_R(z) \sqrt{X} I_{-\nu_-}(zX)$ is entire and 
even (see (\ref{defI})) we can extend the asymptotics on the whole complex plane. This is also true for $F_{R2}$, $F_{R3}$ and $F_{R4}$.\\
\noindent
We also observe that, in the asymptotics for large $z$ in the complex plane, we have
\[ F_{R1}(X,\lambda,z) \sim F_{R4}(X,\lambda,\overline{z})^{\star}\]
and
\[ F_{R2}(X,\lambda,z) \sim F_{R3}(X,\lambda,\overline{z})^{\star}.\]
These are the symmetries given in \cite{DN} and which are not true in general here due to the mass $m$ of the Dirac fields. We expected to find these symmetries in the asymptotics for large $z$ because the most 
important term in the potential for such $z$ is $za(x)$, thus the mass term $b$ has no influence as $|z|$ becomes large.
\end{remark}
\noindent
Using the asymptotic (\ref{asInu}), we can prove the following Theorem concerning the asymptotics of the Jost functions for $z \to +\infty$, $z$ real.

\begin{theorem}\label{asJostdroiteR}
 The Jost functions from the right satisfy the following asymptotics for $z \to +\infty$, $z$ real, uniformly on each compact subset of $]0,A[$,
 \[ F_{R1}(X,\lambda,z) = \frac{2^{-\nu_-}}{\sqrt{2\pi}} \left( \frac{\kappa_-}{a_-} \right)^{i \frac{(\lambda-c_-)}{\kappa_-}} \Gamma(1-\nu_-) z^{-i\frac{(\lambda-c_-)}{\kappa_-}} e^{zX} \begin{pmatrix}
                                                                                                                                                                                               1 & O\left( \frac{1}{z} \right) \\ O\left( \frac{1}{z} \right) & 1
                                                                                                                                                                                              \end{pmatrix} \left( 1 + O \left( \frac{1}{z} \right) \right),\]  
 \[ F_{R2}(X,\lambda,z) = i \frac{2^{-\mu_-}}{\sqrt{2\pi}} \left( \frac{\kappa_-}{a_-} \right)^{-i \frac{(\lambda-c_-)}{\kappa_-}} \Gamma(1-\mu_-) z^{i\frac{(\lambda-c_-)}{\kappa_-}} e^{zX} \begin{pmatrix}
                                                                                                                                                                                               O\left( \frac{1}{z} \right) &1 \\ -1 & O\left( \frac{1}{z} \right) 
                                                                                                                                                                                              \end{pmatrix}
 \left( 1 + O \left( \frac{1}{z} \right) \right),\] 
 \[ F_{R3}(X,\lambda,z) = i \frac{2^{-\nu_-}}{\sqrt{2\pi}} \left( \frac{\kappa_-}{a_-} \right)^{i \frac{(\lambda-c_-)}{\kappa_-}} \Gamma(1-\nu_-) z^{-i\frac{(\lambda-c_-)}{\kappa_-}} e^{zX} \begin{pmatrix}
                                                                                                                                                                                                O\left( \frac{1}{z} \right) &1 \\ -1 & O\left( \frac{1}{z} \right)
                                                                                                                                                                                              \end{pmatrix} \left( 1 + O \left( \frac{1}{z} \right) \right),\] 
 \[ F_{R4}(X,\lambda,z) = \frac{2^{-\mu_-}}{\sqrt{2\pi}} \left( \frac{\kappa_-}{a_-} \right)^{-i \frac{(\lambda-c_-)}{\kappa_-}} \Gamma(1-\mu_-) z^{i\frac{(\lambda-c_-)}{\kappa_-}} e^{zX} \begin{pmatrix}
                                                                                                                                                                                               1 & O\left( \frac{1}{z} \right) \\ O\left( \frac{1}{z} \right) & 1
                                                                                                                                                                                              \end{pmatrix} \left( 1 + O \left( \frac{1}{z} \right) \right).\] 
\end{theorem}

\begin{remark}
 Corresponding asymptotics are also true for $z \to - \infty$ by parity/imparity. 
\end{remark}

We start proving Theorem \ref{asJostdroite} by an estimate on the Green kernels of the corresponding inhomogenous equations. This is done 
using good estimates of the modified Bessel functions given in \cite{Leb}.
We know, thanks to Proposition \ref{eqJostX}, that the components of Jost functions $F_{Ri,j}$ satisfy the coupled differential equations
\begin{equation}\label{eqF_R1,1}
 u''(X,\lambda,z) + q(X,\lambda) u(X,\lambda,z) = z^2 u(X,\lambda,z) + h_{i,j,k}(X,\lambda,z),
\end{equation}
where
 \[q(X,\lambda) = \frac{\omega_-}{X^2} + q_-(X,\lambda), \quad \mathrm{with} \quad \omega_{-} = \frac{(c_--\lambda)^2}{\kappa_-^2} -i \frac{(c_--\lambda)}{\kappa_-} \quad \mathrm{and} \quad q_{-}(X,\lambda) = O(1), \quad X \to 0,\]  
and
\[h_{i,j,k}(X,\lambda,z) = \frac{f_{i,j,k}(X,\lambda,z)}{a(X)^2}\]
satisfies
 \[h_{i,j,k}(X,\lambda,z) = O(1), \quad \mathrm{when} \quad X \to 0.\]
We begin by studying the homogeneous equation
\begin{equation}\label{eqhomo}
 u'' + \frac{\omega_-}{X^2}u = z^2 u
\end{equation}
which is a Bessel equation. We choose a different fondamental system of solutions of (\ref{eqhomo}) according to the block we study.
For the block $F_{R1}$ we choose $(I_{-\nu_-},K_{-\nu_-})$ whereas for the block $F_{R2}$ we choose
$(I_{\nu_-},K_{\nu_-})$. Concerning $F_{R1,j}$, $j \in \{1,2\}$, we set
\begin{eqnarray}\label{u_0,1}
 u_{0,1j}(X) &=& \alpha_{1j} \sqrt{X} I_{-\nu_-}(zX) + \beta_{1j} \sqrt{X} K_{-\nu_-}(zX),
\end{eqnarray}
whereas for $F_{R2,j}$, $j \in \{1,2\}$, we set
\begin{eqnarray}\label{u_0,2}
u_{0,2j}(X) &=& \alpha_{2j} \sqrt{X} I_{\nu_-}(zX) + \beta_{2j} \sqrt{X} K_{\nu_-}(zX).
\end{eqnarray}
Now, we solve the equation (\ref{eqF_R1,1}) by perturbation. We rewrite (\ref{eqF_R1,1}) as
\[ u'' + \frac{\omega_-}{X^2}u - z^2 u = - q_- (X,\lambda)u + h_{i,j,k}(X,\lambda,z).\]
The Green kernel for the block $F_{R1}$ is defined by 
\[G_1(t,X,z) = \sqrt{tX}(I_{-\nu_-}(zt)K_{-\nu_-}(zX)-I_{-\nu_-}(zX)K_{-\nu_-}(zt)),\]
whereas the Green kernel for the block $F_{R2}$ is defined by
\[G_2(t,X,z) = \sqrt{tX}(I_{\nu_-}(zt)K_{\nu_-}(zX)-I_{\nu_-}(zX)K_{\nu_-}(zt)).\]
The general solution of (\ref{eqF_R1,1}) is then
\begin{equation}\label{eqintu}
 u(X,\lambda,z) = u_{0,ij}(X,\lambda,z) + \int_{0}^{X} G_i(t,X,z)(q_-(t,\lambda)u(t,\lambda,z) + h_{i,j,k}(t,\lambda,z))\, \mathrm dt,
\end{equation}
where the integral term makes sense thanks to Lemma \ref{termesourceb}. To prove Theorem \ref{asJostdroite} we need the following estimates on the Green kernels $G_1$ 
and $G_2$. The proof of this Proposition is based on good estimates on the modified Bessel functions (see (\ref{asInu})) and its derivatives 
(see \cite{Iro} eq. (2.17) or \cite{DKN} Proposition 3.1 for a similar proof).

\begin{prop}\label{estimG}
 If $\vert z \vert \geq 1$, for $i \in \{1,2\}$, for all $X \in ]0,A[$ and for all $t \in ]0,X[$,
\[ \vert G_{i}(t,X,z) \vert \leq \frac{Ce^{\vert \mathrm{Re}(z) \vert (X-t)}}{(1+\vert zX \vert^{\frac{1}{2}})(1+\vert zt \vert^{\frac{1}{2}})}.\]
\end{prop}

\begin{remark}
 It suffices to prove Proposition \ref{estimG} for $z$ such that $\mathrm{Re}(z) \geq 0$. Indeed, using the definition of the modified Bessel 
functions (\ref{defI}) and (\ref{defK}) we know that
\begin{eqnarray*}
 G_2(t,X,z) &=& \frac{\pi \sqrt{tX}}{2 \sin(\nu_- \pi)}(I_{\nu_-}(zt) (I_{-\nu_-}(zX) - I_{\nu_-}(zX)) - I_{\nu_-}(zX)(I_{-\nu_-}(zt) - I_{\nu_-}(zt)))\\
&=& \frac{\pi \sqrt{tX}}{2 \sin(\nu_-  \pi)}(I_{\nu_-}(zt) I_{-\nu_-}(zX) - I_{\nu_-}(zX)I_{-\nu_-}(zt)).
\end{eqnarray*}
 Finally, using the definition of $I_{\nu}$ given previously in (\ref{defI}), we obtain that the application $z \mapsto G_2(t,X,z)$ is even in the variable $z\in \mathbb{C}$. Similarly, the application $z \mapsto G_1(t,X,z)$ is even in $z \in \mathbb{C}$.
\end{remark}

\subsection{Asymptotics of the Jost functions from the right when X tends to 0}
\noindent
To prove Theorem \ref{asJostdroite} we need the asymptotics of the Jost functions from the right when the Liouville variable $X$ tends to $0$. These asymptotics allow us in a second time to find explicitely the principal 
term of the components of Jost functions in terms of modified Bessel functions. To obtain these asymptotics we use 
the fact that we can write the Faddeev blocks as series and the link between Faddeev matrices and the Jost functions. First, we do the analysis for the block $F_{R1}$.

\begin{lemma}\label{estimF_R1}
 When $X \to 0$, for any fixed $z \in \mathbb{C}$,
\[F_{R1}(X,\lambda,z) = \left(\left( \frac{\kappa_-}{a_-} \right)^{i \left(\frac{\lambda-c_-}{\kappa_-}\right)} X^{i \left(\frac{\lambda-c_-}{\kappa_-}\right)} + O(X^2)\right)  \left( \begin{array}{cc} 1 & O(X^2) \\O(X^2)  & 1 \end{array} \right) .\]
\end{lemma}

\begin{proof}
Thanks to (\ref{serieM_R1}) we know that
\[ \Vert \hat{M}_{R1}(x,\lambda,z) - \hat{M}_{R1}^0(x,\lambda,z) \Vert = O \left( \left(\int_{-\infty}^x \Vert  k(s,z) \Vert \, \mathrm ds \right)^{2} \right).\]
We recall that
\[k(x,z)= e^{2iC^{-}(x)} \left( \begin{array}{cc} -ib(x) & za(x) \\-za(x)  &ib(x) \end{array} \right).\]
Moreover, according to the asymptotics, given in Lemma \ref{aspot}, we have $b(x) = O(a(x))$ when $x \to -\infty$. Then $\Vert k(s) \Vert \leq C a(x)$ for $x \to -\infty$. 
Thus, since $\hat{M}_{R1}^0(x,\lambda,z) = I_2$,
\[ \Vert \hat{M}_{R1}(x,\lambda,z) - I_2 \Vert = O(X^2), \quad X \to 0.\]
Concerning $F_{R1}$, we know that
\[F_{R1} (x,\lambda,z) = e^{i(\lambda x-C^{-}(x))} \hat{M}_{R1}(x,\lambda,z),\]
then
\[ \Vert F_{R1}(x,\lambda,z) - F_{R1}^0(x,\lambda,z) \Vert = O(X^2), \quad X \to 0,\]
where 
\[F_{R1}^0 (x,\lambda,z) = e^{i(\lambda x-C^{-}(x))} I_2.\]
Moreover,
\begin{equation}\label{estimexpX}
 e^{i \lambda h(X)} = \left( \frac{\kappa_-}{a_-} \right)^{\frac{i\lambda}{\kappa_-}} X^{\frac{i\lambda}{\kappa_-}}+O(X^2), \quad X \to 0.
\end{equation}
Indeed, thanks to the asymptotic of the potential $a$, given in Lemma \ref{aspot},
\[ X = g(x) = \int_{-\infty}^{x} a(t) \, \mathrm dt = \frac{a_-}{\kappa_-}e^{\kappa_- x} + O(e^{3\kappa_- x}), \quad x \to - \infty,\]
then,
\[e^{i \lambda h(X)} = e^{i\lambda \left( \frac{1}{\kappa_-} \ln \left( \frac{\kappa_-}{a_-} X \right) + O(X^2) \right)} = \left( \frac{\kappa_-}{a_-} \right)^{\frac{i\lambda}{\kappa_-}} X^{\frac{i\lambda}{\kappa_-}}(1+O(X^2)) = \left( \frac{\kappa_-}{a_-} \right)^{\frac{i\lambda}{\kappa_-}} X^{\frac{i\lambda}{\kappa_-}}+O(X^2), \quad X \to 0.\]
Furthermore,
\begin{eqnarray*}
 C^-(x) &=& \int_{-\infty}^x (c(s)-c_-) \, \mathrm ds +c_-x\\
&=& \int_{-\infty}^x (c_-e^{2\kappa_- s} + O(e^{4\kappa_- s})) \, \mathrm ds +c_-x\\
&=& c_-x + O(X^2), \quad X \to 0.
\end{eqnarray*}
Thus,
\begin{eqnarray*}
 e^{i(\lambda h(X)-C^{-}(h(X)))} &=& e^{ih(X)(\lambda -c_-)} + O(X^2)\\
&=& \left( \frac{\kappa_-}{a_-} \right)^{i \left(\frac{\lambda-c_-}{\kappa_-}\right)} X^{i \left(\frac{\lambda-c_-}{\kappa_-}\right)} + O(X^2), \quad X \to 0.
\end{eqnarray*}
\end{proof}

Thanks to these estimates we can find the principal term of the Jost function in terms of modified Bessel functions.
Concerning the coefficient $F_{R1,1}$, we first have to find $\alpha_{11}$ and $\beta_{11}$  defined in (\ref{u_0,1}) to obtain that the solution $u_{11}$ of the equation (\ref{eqintu}) 
\[ u_{11}(X,\lambda,z) = u_{0,11}(X,\lambda,z) + \int_{0}^{X} G_1(t,X,z)(q_-(t,\lambda)u_{11}(t,\lambda,z) + h_{1,1,3}(t,\lambda,z))\, \mathrm dt\]
is $F_{R1,1}$. Concerning $F_{R1,2}$ we have to find $\alpha_{12}$ and $\beta_{12}$ to obtain that the solution $u_{12}$ of (\ref{eqintu}) 
\[ u_{12}(X,\lambda,z) = u_{0,12}(X,\lambda,z) + \int_{0}^{X} G_1(t,X,z)(q_-(t,\lambda)u_{12}(t,\lambda,z) + h_{1,2,4}(t,\lambda,z))\, \mathrm dt\]
is $F_{R1,2}$. Using the asymptotics 
of $F_{R1,1}(X,\lambda,z)$ and $F_{R1,2}(X,\lambda,z)$ when $X \to 0$ we prove the following Proposition.

\begin{prop}\label{propoF_R1,1}
For any fixed $z$, for any $X \in ]0,A[$,
\[F_{R1,1}(X,\lambda,z) = F_{R1,1}^- (X,\lambda,z) + \int_{0}^{X} G_1(t,X,z)(q_-(t,\lambda)F_{R1,1}(t,\lambda,z) + h_{1,1,3}(t,\lambda,z))\, \mathrm dt,\]
where,
\[ F_{R1,1}^- (X,\lambda,z) = \frac{z^{\nu_-}}{2^{\nu_-}} \left( \frac{\kappa_-}{a_-} \right)^{i \left( \frac{\lambda-c_-}{\kappa_-} \right)} \Gamma(1-\nu_-) \sqrt{X} I_{-\nu_-}(zX).\]
Similarly,
\[F_{R1,2}(X,\lambda,z) = \int_{0}^{X} G_1(t,X,z)(q_-(t,\lambda)F_{R1,2}(t,\lambda,z) + h_{1,2,4}(t,\lambda,z))\, \mathrm dt.\]
\end{prop}

\begin{proof}
 We search $\alpha_{11}$ and $\beta_{11}$ such that
\[F_{R1,1}(X,\lambda,z) = F_{R1,1}^- (X,\lambda,z) + \int_{0}^{X} G_1(t,X,z)(q_-(t,\lambda)F_{R1,1}(t,\lambda,z) + h_{1,1,3}(t,\lambda,z))\, \mathrm dt,\]
where
\[ F_{R1,1}^- (X,\lambda,z) = \alpha_{11} \sqrt{X} I_{-\nu_-}(zX) + \beta_{11} \sqrt{X} K_{-\nu_-}(zX).\]
We recall (see \cite{Leb} eqs $(5.7.1)$ and $(5.7.2)$) that,
\begin{equation}\label{estimI_nu}
 \alpha \sqrt{X} I_{-\nu_-}(zX) \sim  \frac{\alpha z^{-\nu_-}}{\Gamma(1-\nu_-)2^{-\nu_-}} X^{i \left(\frac{\lambda-c_-}{\kappa_-}\right)}, \quad X \to 0
\end{equation}
and
\begin{equation}\label{estimK_nu}
 \beta \sqrt{X} K_{-\nu_-}(zX) \sim  \frac{-\beta \pi}{2\sin(\nu_- \pi)} \frac{z^{\nu_-}}{\Gamma(1+\nu_-)2^{\nu_-}} X^{1-i \left(\frac{\lambda-c_-}{\kappa_-}\right)}, \quad X \to 0.
\end{equation}
Thus, for any fixed $z \in \mathbb{C}$, when $X \to 0$,
\begin{eqnarray*}
 G_1(t,X,z) &=& \sqrt{tX}(I_{-\nu_-}(zt)K_{-\nu_-}(zX)-I_{-\nu_-}(zX)K_{-\nu_-}(zt))\\
&=& \sqrt{t}I_{-\nu_-}(zt) \sqrt{X} K_{-\nu_-}(zX)-\sqrt{X}I_{-\nu_-}(zX) \sqrt{t}K_{-\nu_-}(zt)\\
&\sim& \frac{- \pi}{\sin(\nu_- \pi) \Gamma(1-\nu_-) \Gamma(1+\nu_-)}\left( t^{i \left(\frac{\lambda-c_-}{\kappa_-}\right)} X^{1-i \left(\frac{\lambda-c_-}{\kappa_-}\right)} - X^{i \left(\frac{\lambda-c_-}{\kappa_-}\right)} t^{1-i \left(\frac{\lambda-c_-}{\kappa_-}\right)}\right).
\end{eqnarray*}
Finally, thanks to Lemma \ref{termesourceb},
\[ \int_{0}^{X} G_1(t,X,z)\underbrace{(q_-(t)F_{R1,1}(t,\lambda,z) + h_{1,1,3}(t,\lambda,z))}_{=O(1)}\, \mathrm dt = O(X^2), \quad X \to 0.\]
Moreover, thanks to Lemma \ref{estimF_R1},
\[F_{R1,1}(X,\lambda,z) = \left( \frac{\kappa_-}{a_-} \right)^{i \left(\frac{\lambda-c_-}{\kappa_-}\right)} X^{i \left(\frac{\lambda-c_-}{\kappa_-}\right)} + O(X^2), \quad X \to 0.\]
This permits to conclude, using again (\ref{estimI_nu}) and (\ref{estimK_nu}), that
\[ \alpha_{11} = \frac{z^{\nu_-}}{2^{\nu_-}} \left( \frac{\kappa_-}{a_-} \right)^{i \left( \frac{\lambda-c_-}{\kappa_-} \right)} \Gamma(1-\nu_-) \quad \quad \mathrm{and} \quad \quad \beta_{11}=0.\]
The same argument allows us to prove that
\[ \alpha_{12} = 0 \quad \quad \mathrm{and} \quad \quad \beta_{12}=0.\]
\end{proof}

We do the same analysis for the block $F_{R2}$. Since the study is similar we just give the corresponding results without proofs.

\begin{lemma}\label{estimF_R2}
 When $X \to 0$,
\[F_{R2,1}(X,\lambda,z) = \frac{b_-}{2a_- \nu_-} \left( \frac{\kappa_-}{a_-} \right)^{-i \left( \frac{\lambda-c_-}{\kappa_-} \right)} X^{1-i \left( \frac{\lambda-c_-}{\kappa_-} \right)} + O(X^3)\]
and
\[F_{R2,2}(X,\lambda,z) = iz \frac{1}{2\nu_-} \left( \frac{\kappa_-}{a_-} \right)^{-i \left(\frac{\lambda-c_-}{\kappa_-}\right)} X^{1-i \left(\frac{\lambda-c_-}{\kappa_-}\right)} + O(X^3).\]
\end{lemma}


The analogous of Proposition \ref{propoF_R1,1} for the second block is the following.

\begin{prop}\label{propoF_R2}
For any fixed $z$ and for any $X \in ]0,A[$,
\[F_{R2,1}(X,\lambda,z) = F_{R2,1}^- (X,\lambda,z) + \int_{0}^{X} G_2(t,X,z)(q_-(t,\lambda)F_{R2,1}(t,\lambda,z) + h_{2,1,3}(t,\lambda,z))\, \mathrm dt,\]
where
\[ F_{R2,1}^- (X,\lambda,z) = \frac{b_-}{2a_-} \left(\frac{z}{2}\right)^{-\nu_-}\left( \frac{\kappa_-}{a_-} \right)^{-i \left( \frac{\lambda-c_-}{\kappa_-} \right)} \Gamma(1-\mu_-) \sqrt{X} I_{\nu_-}(zX).\]
Similarly,
\[F_{R2,2}(X,\lambda,z) = F_{R2,2}^- (X,\lambda,z) + \int_{0}^{X} G_2(t,X,z)(q_-(t,\lambda)F_{R2,2}(t,\lambda,z) + h_{2,2,4}(t,\lambda,z))\, \mathrm dt,\]
where
\[ F_{R2,2}^- (X,\lambda,z) = i  \left(\frac{z}{2}\right)^{\mu_-}  \left( \frac{\kappa_-}{a_-} \right)^{-i \left( \frac{\lambda-c_-}{\kappa_-} \right)} \Gamma(1-\mu_-) \sqrt{X} I_{\nu_-}(zX).\]
\end{prop}


\subsection{Improvement of the first estimates on Jost functions}\label{partDuh}
\noindent
In our computation of the asymptotics of the coefficients of Jost functions from the right (respectively from the left) for large $z$ in the complex plane 
(see next Section) we need estimates on these functions of the form
\[ |F_{Ri,j}(X,\lambda,z)| \leq Ce^{|\mathrm{Re}(z)|X}, \quad |F_{Li,j}(X,\lambda,z)| \leq Ce^{|\mathrm{Re}(z)|(A-X)}, \quad \forall (i,j) \in \{1,2,3,4\}^2,\]
for all $z \in \C$ uniformly for $X \in ]0,X_1[$ (respectively $X \in ]X_1,A[$), $X_1 \in ]0,A[$ fixed.
At this time, thanks to Lemma \ref{estimM} and the link between the Faddeev matrices 
and the Jost functions given by (\ref{lienFMR}), we know that,
\[ |F_{Ri,j}(X,\lambda,z)| \leq Ce^{|z|X}, \quad |F_{Li,j}(X,\lambda,z)| \leq Ce^{|z|X}, \quad \forall (i,j) \in \{1,2,3,4\}^2,\]
for all $z \in \C$ uniformly for $X \in ]0,X_1[$ (respectively $X \in ]X_1,A[$), $X_1 \in ]0,A[$ fixed. 
These estimates are not enough for our purpose. Using the Phragm\'en-Lindel\"{o}f's Theorem, we observe that it would be sufficient to prove that the components of the Jost functions are bounded on $i \mathbb{R}$. We shall now prove this claim. We first note that the scalar equations obtained on the components of Jost functions in the variable $X$ in Lemma \ref{eqJostX} are
coupled, due to the form of the rest $h_{i,j,k}$. We thus transform each of these coupled scalar equations in one vectorial equation and study them by perturbation using an iterative method.

\begin{lemma}\label{bornsuriR}
 For all $(i,j) \in \{1,2,3,4\}^2$ and for all $X_1 \in ]0,A[$ the function $z \mapsto F_{Ri,j}(X,\lambda,z)$ is bounded on $i\mathbb{R}$ uniformly in $X \in ]0,X_1[$ and 
$z \mapsto F_{Li,j}(X,\lambda,z)$ is bounded on $i\mathbb{R}$ uniformly in $X \in ]X_1,A[$.
\end{lemma}

\begin{proof}
 Since the proofs are similar in the other cases, we just do the proof for the couple $(F_{R1,1},F_{R1,3})$. We start with the second order scalar equations given in Lemma \ref{eqJostX}:
\[ F_{R1,1}''(X,\lambda,z) +q(X,\lambda) F_{R1,1}(x,\lambda,z) = z^2 F_{R1,1}(X,\lambda,z) + h_{1,1,3}(X,\lambda,z),\]
\[ F_{R1,3}''(X,\lambda,z) +q(X,\lambda) F_{R1,3}(x,\lambda,z) = z^2 F_{R1,3}(X,\lambda,z) + h_{1,3,1}(X,\lambda,z),\]
where (see Lemmas \ref{eqJost1} and \ref{eqJostX}, Equation (\ref{deffij}) and Lemma \ref{termesourceb})
\begin{eqnarray*}
h_{1,1,3}(X,\lambda,z) &=& c_1(X,\lambda,z) F_{R1,1}'(X,\lambda,z) + c_2(X,\lambda,z) F_{R1,3}'(X,\lambda,z) \\
& &  + c_3(X,\lambda,z) F_{R1,1}(X,\lambda,z) + c_4(X,\lambda,z) F_{R1,3}(X,\lambda,z),
\end{eqnarray*}
\begin{eqnarray*}
h_{1,3,1}(X,\lambda,z) &=& c_1(X,\lambda,z) F_{R1,3}'(X,\lambda,z) + c_2(X,\lambda,z) F_{R1,1}'(X,\lambda,z) \\
& & + c_3(X,\lambda,z) F_{R1,3}(X,\lambda,z) + c_4(X,\lambda,z) F_{R1,1}(X,\lambda,z),
\end{eqnarray*}
and the functions $c_i$ are given in the proof of Lemma \ref{eqJost1} (see (\ref{defci})). We know that when $X \to 0$
\begin{eqnarray}\label{esticX}
  c_i(X,\lambda,z) = O(X), \quad \forall i \in \{1,2\}, \quad c_i(X,\lambda,z) = O(1), \quad \forall i \in \{3,4\},
\end{eqnarray}
and when $z \to \infty$ in the complex plane
\begin{eqnarray}\label{esticz}
  c_i(X,\lambda,z) = O \left( \frac{1}{z} \right) \quad \forall i \in \{1,2,3,4\}.
\end{eqnarray}
We now transform this pair of scalar equations of second order in a single vectorial equation of first order. We set,
\[U(X,\lambda,z) = \left( \begin{array}{c} F_{R1,1}(X,\lambda,z) \\ F_{R1,3}(X,\lambda,z) \\ F_{R1,1}'(X,\lambda,z) \\ F_{R1,3}'(X,\lambda,z) \end{array} \right),\]
and we obtain
\[ U'(X,\lambda,z) = A(X,\lambda,z) U(X,\lambda,z) + B(X,\lambda,z) U(X,\lambda,z),\]
where
\[A(X,\lambda,z) = \left( \begin{array}{cccc} 0 & 0 & 1 & 0\\ 0& 0 & 0& 1 \\  z^2 - \frac{\omega_-}{X^2} &0&0&0 \\ 0&z^2 - \frac{\omega_-}{X^2}&0&0 \end{array} \right), \quad 
B(X,\lambda,z) = \left( \begin{array}{cccc} 0 & 0 & 0 & 0\\ 0& 0 & 0& 0 \\ c_3 - q_- & c_4 & c_1 & c_2\\ c_4 & c_3 - q_- & c_2 & c_1 \end{array} \right).\]
We use the notation
\[U(X,\lambda,z) = \left( \begin{array}{c} U_1(X,\lambda,z) \\ U_2(X,\lambda,z) \\ U_3(X,\lambda,z) \\ U_4(X,\lambda,z) \end{array} \right),\]
and we begin by studying the homogeneous equation
\[ U'(X,\lambda,z) = A(X,\lambda,z) U(X,\lambda,z).\]
We note that this equation is equivalent to the pair of modified Bessel equations
\[ U_1''(X,\lambda,z) + \frac{\omega_-}{X^2} U_1 (X,\lambda,z) = z^2 U_1 (X,\lambda,z),\]
\[ U_3''(X,\lambda,z) + \frac{\omega_-}{X^2} U_3 (X,\lambda,z) = z^2 U_3 (X,\lambda,z).\]
For the basis of the space of solutions of these scalar equations we choose $(\sqrt{X}I_{\nu_-}(zX),\sqrt{X}I_{-\nu_-}(zX))$ ($\nu_- \notin \mathbb{Z}$). In the following of the proof we set $f_{\nu} = \sqrt{X}I_{\nu_-}(zX)$ 
and $g_{\nu} = \sqrt{X}I_{-\nu_-}(zX)$.
Thus, we can choose the basis
formed by the four vectors
\[V_1 = \left( \begin{array}{c} f_{\nu} \\0 \\ f_{\nu}' \\ 0 \end{array} \right), \quad V_2 = \left( \begin{array}{c} 0 \\ f_{\nu} \\ 0 \\ f_{\nu}' \end{array} \right), \quad
V_3 = \left( \begin{array}{c} g_{\nu} \\0 \\ g_{\nu}' \\ 0 \end{array} \right), \quad V_2 = \left( \begin{array}{c} 0 \\ g_{\nu} \\ 0 \\ g_{\nu}' \end{array} \right)\]
and we set $C(X,\lambda,z)$ the matrix in which the $\mathrm{i}^{\mathrm{th}}$ column is $V_i$. Using the properties on the Wronskian of two modified Bessel functions (see \cite{Leb}), we can prove the following Lemma.

\begin{lemma}\label{detC}
 \[\det (C(X,\lambda,z)) = W(f_{\nu},g_{\nu})^2 = \frac{4 \sin (\nu \pi)^2}{\pi^2} =: \Delta.\]
\end{lemma}

We can now study the inhomogeneous equation. We write the solution of the homogeneous equation under the form
\[U_0(X,\lambda,z) = \left( \begin{array}{c} \alpha_1 f_{\nu} + \beta_1 g_{\nu} \\ \alpha_3 f_{\nu} + \beta_3 g_{\nu} \\ \alpha_1 f_{\nu}' + \beta_1 g_{\nu}' \\ \alpha_3 f_{\nu}' + \beta_3 g_{\nu}' \end{array} \right),\]
where $\alpha_i$ and $\beta_i$, $i \in \{1,3\}$ are chosen so that the Jost functions satisfy the prescribed asymptotics at $X \to 0$. We can now write a generalized Duhamel's formula
\[ U(X,\lambda,z) = U_0(X,\lambda,z) + \int_{0}^{X} R(X,t,\lambda,z)B(t,\lambda,z)U(t,\lambda,z)  \, \mathrm dt,\]
where $R(X,t, \lambda,z)$ is the resolvent of the homogeneous problem for $t \in ]0,X[$. We can write $R(X,t,\lambda,z) = C(X,\lambda,z) C(t,\lambda,z)^{-1}$ and we obtain
\begin{eqnarray}\label{Duh}
 U(X,\lambda,z) = U_0(X,\lambda,z) + C(X,\lambda,z) \int_{0}^{X} C(t,\lambda,z)^{-1}B(t,\lambda,z)U(t,\lambda,z)  \, \mathrm dt.
\end{eqnarray}
The integral term makes sense thanks to the known asymptotics of the Jost function from the right and their derivatives when $X \to 0$, the 
good estimates on the functions $c_i$ given by (\ref{esticX}) and the formula (\ref{defI}).
We will use an iterative method to solve this integral equation. We define
\[ U^{0}(X,\lambda,z) = U_0(X,\lambda,z),\]
\begin{eqnarray}\label{ite}
 U^{k+1}(X,\lambda,z) = C(X,\lambda,z) \int_{0}^{X} C(t,\lambda,z)^{-1}B(t,\lambda,z)U^{k}(t,\lambda,z)  \, \mathrm dt, \quad \forall k \in \mathbb{N}.
\end{eqnarray}
Our aim is now to prove the following Lemma.

\begin{lemma}\label{estimUk}
 For all $X_{1} \in ]0,A[$ there exists a constant $C \in \mathbb{R}$ such that for all $X \in ]0,X_1[$ and for large pure imaginary complex $z$ we have
\[ |U_i^k(X,\lambda,z)| \leq \alpha_k \frac{5^k C^{3k+1}}{|\Delta|^k} X^{2k}, \quad i \in \{1,2\},\]
\[ |U_i^k(X,\lambda,z)| \leq \alpha_k \frac{5^k C^{3k+1} |z|}{|\Delta|^k} X^{2k-1}, \quad i \in \{3,4\},\]
where
\[ \alpha_0 = 1 \quad \mathrm{and} \quad \alpha_{k+1} = \left( \frac{1}{2k+1} + \frac{1}{2k+2} \right) \alpha_k, , \quad \forall k \in \mathbb{N}.\]
\end{lemma}

\begin{proof}
We begin the proof by the case $k=0$. Thanks to the estimates (\ref{asInu}) and the formula (\ref{defI}) we know that for all $X_{0} \in ]0,A[$ 
there exists a constant $C \in \mathbb{R}$ such that for all $X \in ]0,X_1[$ and for large $z = iy$, $y \in \mathbb{R}$:
\begin{eqnarray}\label{estimB}
 | f_{\nu}(X,\lambda,z) | \leq \frac{CX}{|z|^{\frac{1}{2}}}, \quad | g_{\nu}(X,\lambda,z) | \leq \frac{C}{|z|^{\frac{1}{2}}},
\end{eqnarray}
\begin{eqnarray}\label{estimdB}
 | f_{\nu}'(X,\lambda,z) | \leq C|z|^{\frac{1}{2}}, \quad | g_{\nu}'(X,\lambda,z) | \leq \frac{C|z|^{\frac{1}{2}}}{X}.
\end{eqnarray}
We recall that
\[U^0(X,\lambda,z) = \left( \begin{array}{c} \alpha_1 f_{\nu} + \beta_1 g_{\nu} \\ \alpha_3 f_{\nu} + \beta_3 g_{\nu} \\ \alpha_1 f_{\nu}' + \beta_1 g_{\nu}' \\ \alpha_3 f_{\nu}' + \beta_3 g_{\nu}' \end{array} \right).\]
We know that the constants $\alpha_i$ and $\beta_i$ don't depend on $X$ but they depend on $z$. Indeed, thanks to Lemma \ref{propoF_R1,1} we have
\[ \alpha_1 = \alpha_ 3 = \beta_3 = 0, \quad \beta_1 = \left( \frac{z}{2} \right)^{\nu_-} \left( \frac{\kappa_-}{a_-} \right)^{i\left( \frac{\lambda-c_-}{\kappa_-} \right)} \Gamma(1-\nu-).\]
Thus, using the estimates (\ref{estimB}) and (\ref{estimdB}) we obtain the proof of Lemma \ref{estimUk} in the case $k=0$.
We suppose that the result is true for an integer $k$ and we prove the result of the Lemma for the first and the third components, since the study is the same for the other components.
Thanks to (\ref{ite}) we know that (for a sake of simplicity we omit the parameters of the functions),
\begin{eqnarray*}
 U_1^{k+1}(X,\lambda,z) \Delta &=& -f_{\nu} \int_{0}^{X} g_{\nu}( (c_3-q_-) U_1^k + c_4 U_2^k + c_1 U_3^k + c_2 U_4^k)  \, \mathrm dt \\
 && + g_{\nu} \int_{0}^{X} f_{\nu}( (c_3-q_-) U_1^k + c_4 U_2^k + c_1 U_3^k + c_2 U_4^k)  \, \mathrm dt
\end{eqnarray*}
and
\begin{eqnarray*}
 U_3^{k+1}(X,\lambda,z) \Delta &=& -f_{\nu}' \int_{0}^{X} g_{\nu}( (c_3-q_-) U_1^k + c_4 U_2^k + c_1 U_3^k + c_2 U_4^k)  \, \mathrm dt \\
 && + g_{\nu}' \int_{0}^{X} f_{\nu}( (c_3-q_-) U_1^k + c_4 U_2^k + c_1 U_3^k + c_2 U_4^k)  \, \mathrm dt.
\end{eqnarray*}
Using the hypothesis we can estimate the integral terms. Indeed,
\begin{eqnarray*}
  \left| \int_{0}^{X} g_{\nu}( (c_3-q_-) U_1^k + c_4 U_2^k + c_1 U_3^k + c_2 U_4^k)  \, \mathrm dt \right| &\leq& 5C^2 \alpha_k \frac{5^k C^{3k+1}}{|\Delta|^k} \frac{X^{2k+1}}{2k+1}
\end{eqnarray*}
and
\begin{eqnarray*}
 \left| \int_{0}^{X} f_{\nu}( (c_3-q_-) U_1^k + c_4 U_2^k + c_1 U_3^k + c_2 U_4^k)  \, \mathrm dt \right| \leq 5C^2 \alpha_k \frac{5^k C^{3k+1}}{|\Delta|^k} \frac{X^{2k+2}}{2k+2}.
\end{eqnarray*}
Thus, using the estimates (\ref{estimB}) and (\ref{estimdB}), we can prove the Lemma for $U_1^{k+1}$ and $U_3^{k+1}$:
\begin{eqnarray*}
 |\Delta| U_1^{k+1}(X,\lambda,z) &\leq& \frac{5^{k+1} C^{3(k+1)+1}}{|\Delta|^k} \alpha_{k+1} X^{2(k+1)}
\end{eqnarray*}
and
\[ |\Delta| U_3^{k+1}(X,\lambda,z) \leq \frac{5^{k+1} C^{3(k+1)+1} |z|}{|\Delta|^k} \alpha_{k+1} X^{2(k+1)-1}.\]
Thus, the Lemma is proved.
\end{proof}

To use the iterative method we need a sommability result on the term
\[\frac{5^k C^{3k+1}}{|\Delta|^k} \alpha_k X^{2k},\]
which is given by the following Lemma.

\begin{lemma}\label{somalpha}
 \[\alpha_k \leq \frac{3^k}{(2k-1)!!},\]
where $(2k-1)!!$ denotes the product of all the odd number lower than $2k-1$.
\end{lemma}

\begin{proof}
 It suffices to note that
\[ \alpha_k = \prod_{j=1}^k \frac{4j-1}{2j(2j-1)} \leq \prod_{j=1}^k \left( \frac{2}{2j-1} + \frac{1}{(2j-1)^2} \right) \leq \prod_{j=1}^k \frac{3}{2j-1} = \frac{3^k}{(2k-1)!!}.\]
\end{proof}

Thanks to Lemmas \ref{estimUk} and \ref{somalpha} we obtain that for all $X_1 \in ]0,A[$, there exists a constant $C \in \mathbb{R}$ such that for all $X \in ]0,X_1[$ and 
for large pure imaginary complex $z$ we have 
for $i \in \{ 1,2\}$,
 \[ \left| \sum_{k \geq 0} U_i^k (X,\lambda,z) \right| \leq \sum_{k \geq 0} \alpha_k \frac{5^k C^{3k+1}}{|\Delta|^k} X^{2k} \leq \sum_{k \geq 0}  \frac{\left( \frac{15}{|\Delta|} \right)^k C^{3k+1} X^{2k}}{(2k-1)!!} < +\infty\]
and for $i \in \{ 3,4\}$,
 \[ \left| \sum_{k \geq 0} U_i^k (X,\lambda,z) \right| \leq \sum_{k \geq 0}  \frac{\left( \frac{15}{|\Delta|} \right)^k C^{3k+1} |z| X^{2k-1}}{(2k-1)!!} < +\infty.\]
We set
\[ U(X,\lambda,z) = \sum_{k \geq 0} U^k(X,\lambda,z),\]
which is the solution of the integral equation. We recall that
\[ U_1(X,\lambda,z) = F_{R1,1}(X,\lambda,z), \quad U_2(X,\lambda,z) = F_{R1,3}(X,\lambda,z),\]
and we note that we have shown that for all $X_1 \in ]0,A[$, the functions $F_{R1,1}$ and $F_{R1,3}$ are bounded for large $z \in i \mathbb{R}$ uniformly in $X \in ]0,X_1[$. 
Since the Jost functions are continuous in $z$, 
we conclude that $F_{R1,1}$ and $F_{R1,3}$ are bounded on $i \mathbb{R}$ uniformly in $X \in ]0,X_1[$.
\end{proof}

\begin{theorem}\label{estimJ}
 For all $X_1 \in ]0,A[$, there is a constant $C$ such that for all $X \in ]0,X_1[$ and for all $z$ in the complex plane,
\[ |F_{Ri,j}(X,\lambda,z)| \leq C e^{|\mathrm{Re}(z)|X}, \quad |F_{Ri,j}'(X,\lambda,z)| \leq \frac{C|z|}{X} e^{|\mathrm{Re}(z)|X}, \quad \forall (i,j) \in \{1,2,3,4\}^2.\]
 For all $X_1 \in ]0,A[$, there is a constant $C$ such that for all $X \in ]X_1,A[$ and for all $z$ in the complex plane,
\[ |F_{Li,j}(X,\lambda,z)| \leq C e^{|\mathrm{Re}(z)|X},\quad |F_{Li,j}'(X,\lambda,z)| \leq \frac{C|z|}{X} e^{|\mathrm{Re}(z)|X} \quad \forall (i,j) \in \{1,2,3,4\}^2.\]
\end{theorem}

\begin{proof}
 Concerning the Jost functions from the right, thanks to Lemma \ref{estimM} we know that
\[ |F_{Ri,j}(X,\lambda,z)| \leq C e^{|z|X},\quad \forall (i,j) \in \{1,2,3,4\}^2.\]
 Moreover, thanks to Lemma \ref{bornsuriR} we know that
\[ |F_{Ri,j}(X,\lambda,z)| \leq C,\quad \forall z \in i \mathbb{R}, \quad \forall (i,j) \in \{1,2,3,4\}^2.\]
 Thus, since the Jost functions are entire (see Section 5), using the Phragm\'en-Lindel\"{o}f's Theorem we conclude that the first estimate of Theorem \ref{estimJ} is satisfied. The same proof is also true for the Jost functions from the left.\\
To prove the estimate for the derivative of the Jost functions we just recall that these functions satisfy the equation (\ref{eqJ}) and we conclude using the estimates on the Jost functions.
\end{proof}

\begin{coro}\label{estimh}
  For all $X_1 \in ]0,A[$, there is a constant $C$ such that for all $X \in ]0,X_1[$ and for all $z$ in the complex plane,
\[ |h_{i,j,k}(X,\lambda,z)| \leq C e^{|\mathrm{Re}(z)|X}, \quad \forall i, \in \{1,2,3,4\}, \quad \forall (j,k) \in \{(1,3),(2,4),(3,1),(4,2)\}.\]
\end{coro}

\begin{proof}
 We only prove that $h_{1,1,3}(X,\lambda,z)$ satisfies this estimate since the other cases are similar. We recall that
\begin{eqnarray*}
h_{1,1,3}(X,\lambda,z) &=& c_1(X,\lambda,z) F_{R1,1}'(X,\lambda,z) + c_2(X,\lambda,z) F_{R1,3}'(X,\lambda,z) \\
& &  + c_3(X,\lambda,z) F_{R1,1}(X,\lambda,z) + c_4(X,\lambda,z) F_{R1,3}(X,\lambda,z),
\end{eqnarray*}
where the function $c_i$ satisfy (\ref{esticX}) and (\ref{esticz}).
Therefore, thanks to Theorem \ref{estimJ}, we obtain the good estimate.
\end{proof}

\subsection{Asymptotics of the Jost function from the right for large complex z}
\noindent
We begin by studying the first block. Since the other coefficients can be treated using exactly the same argument, we just give the proof 
of Theorem \ref{asJostdroite} for the coefficients $F_{R1,1}$ and $F_{R1,2}$.
We prove the following Proposition:

\begin{prop}\label{approxFR_1}
 For all $X_1 \in ]0,A[$ there exists a constant C such that for large $z$ in the complex plane and for all $X \in ]0,X_1[$,
\[\vert F_{R1,1}(X,\lambda,z) - \alpha_{11} \sqrt{X}I_{-\nu_-}(zX) \vert \leq C \frac{e^{|\mathrm{Re}(z)|X}}{\vert z \vert}\]
and
\[ F_{R1,2}(X,\lambda,z)  = O \left( \frac{e^{|\mathrm{Re}(z)|X}}{\vert z \vert} \right).\]
\end{prop}

\begin{remark}
 Since
\[I_{\nu}(z) = \sum_{k=0}^{\infty} \frac{\left( \frac{z}{2} \right)^{\nu + 2k}}{\Gamma(k+\nu+1) k!}, \quad \vert z \vert < \infty, \quad \vert \arg(z) \vert < \pi\]
and
\[ \alpha_{11} = \frac{z^{\nu_-}}{2^{\nu_-}} \left( \frac{\kappa_-}{a_-} \right)^{i \left( \frac{\lambda-c_-}{\kappa_-} \right)} \Gamma(1-\nu_-),\]
there is no problem to extend the function $z \mapsto \alpha_{11} \sqrt{X}I_{-\nu_-}(zX)$ on $\mathbb{R}^-$ by symmetry since this function is even.
We expected to have this symmetry for large $z$ since in the massless case (\cite{DN})
 the first Jost function from the right is even.
\end{remark}

\begin{proof}
 Thanks to Proposition \ref{propoF_R1,1}, we know that
\[F_{R1,1}(X,\lambda,z) = \alpha_{11} \sqrt{X} I_{-\nu_-}(zX) + \int_{0}^{X} G_1(t,X,z)(q_-(t,\lambda)F_{R1,1}(X,\lambda,z) + h_{1,1,3}(t,\lambda,z))\, \mathrm dt,\]
where
\[ \alpha_{11} = \frac{z^{\nu_-}}{2^{\nu_-}} \left( \frac{\kappa_-}{a_-} \right)^{i \left( \frac{\lambda-c_-}{\kappa_-} \right)} \Gamma(1-\nu_-) \quad \quad \mathrm{and} \quad \quad \nu_- = \frac{1}{2}-i \frac{(\lambda-c_-)}{\kappa_-}.\]
Using Proposition \ref{estimG} and Theorem \ref{estimJ} we obtain that 
 for all $X_1 \in ]0,A[$ there exists a constant C such that for large $z$ in the complex plane and for all $X \in ]0,X_1[$,
\begin{eqnarray*}
 \left| \int_{0}^{X} G_1(t,X,z)q_-(t,\lambda)F_{R1,1}(X,\lambda,z)\, \mathrm dt \right| &\leq& C\frac{e^{|\mathrm{Re}(z)|X}}{1+\vert zX \vert^{\frac{1}{2}}} \int_{0}^{X} \frac{1}{1+\vert zt \vert^{\frac{1}{2}}}\, \mathrm dt\\
&\leq& C \frac{e^{|\mathrm{Re}(z)|X}}{\vert z \vert}.
\end{eqnarray*}
Thus, for all $X_1 \in ]0,A[$ there exists a constant C such that for large $z$ in the complex plane and for all $X \in ]0,X_1[$,
\[ \left| F_{R1,1}(X,\lambda,z) - \alpha_{11} \sqrt{X} I_{-\nu_-}(zX) + \int_{0}^{X} G_1(t,X,z)h_{1,1,3}(t,\lambda,z)\, \mathrm dt \right| \leq C \frac{e^{|\mathrm{Re}(z)|X}}{\vert z \vert}.\]
Moreover, thanks to Proposition \ref{estimG} and Corollary \ref{estimh} we have, for large $z$ in the complex plane
\[\left| \int_{0}^{X} G_1(t,X,z)h_{1,1,3}(t,\lambda,z)\, \mathrm dt \right| \leq C \frac{e^{|\mathrm{Re}(z)|X}}{\vert z \vert}, \quad \forall X \in ]0,X_1[.\]
We conclude that for all $X_1 \in ]0,A[$ there exists a constant C such that for large $z$ in the complex plane,
\[\vert F_{R1,1}(X,\lambda,z) - \alpha_{11} \sqrt{X}I_{-\nu_-}(zX) \vert \leq C \frac{e^{|\mathrm{Re}(z)|X}}{\vert z \vert}, \quad \forall X \in ]0,X_1[.\]
By the same argument we obtain that for all $X_1 \in ]0,A[$ there exists a constant C such that for large $z$ in the complex plane,
\[ \vert F_{R1,2}(X,\lambda,z) \vert  \leq C \frac{e^{|\mathrm{Re}(z)|X}}{\vert z \vert}, \quad \forall X \in ]0,X_1[.\]
\end{proof}


Concerning the second block we obtain the following Proposition. Since the analysis is similar we omit the proof.

\begin{prop}\label{approxFR_2}
For all $X_1 \in ]0,A[$ there exists a constant C such that for large $z$ in the complex plane,
\[ \vert F_{R2,1}(X,\lambda,z) \vert  = C \frac{e^{|\mathrm{Re}(z)|X}}{\vert z \vert}, \quad  \forall X \in ]0,X_1[\]
and
\[\vert F_{R2,2}(X,\lambda,z) - \alpha_{22} \sqrt{X}I_{\nu_-}(zX) \vert \leq C \frac{e^{|\mathrm{Re}(z)|X}}{\vert z \vert}, \quad  \forall X \in ]0,X_1[.\]
\end{prop}

\begin{remark}
 Since
\[I_{\nu}(z) = \sum_{k=0}^{\infty} \frac{\left( \frac{z}{2} \right)^{\nu + 2k}}{\Gamma(k+\nu+1) k!}, \quad \vert z \vert < \infty, \quad \vert \arg(z) \vert < \pi,\]
\[ \alpha_{21} = \frac{b_-}{2a_-} \left(\frac{z}{2} \right)^{-\nu_-}  \left( \frac{\kappa_-}{a_-} \right)^{-i \left( \frac{\lambda-c_-}{\kappa_-} \right)} \frac{\Gamma(1+\nu_-)}{\nu_-}\]
and
\[ \alpha_{22} = i \frac{z}{2} \left(\frac{z}{2} \right)^{-\nu_-}  \left( \frac{\kappa_-}{a_-} \right)^{-i \left( \frac{\lambda-c_-}{\kappa_-} \right)} \frac{\Gamma(1+\nu_-)}{\nu_-}.\]
there is no problem to extend the functions $z \mapsto \alpha_{21} \sqrt{X}I_{\nu_-}(zX)$ and $z \mapsto \alpha_{22} \sqrt{X}I_{\nu_-}(zX)$ on $\mathbb{R}^-$ by symmetry since these functions have parity properties.
We note that the function $z \mapsto \alpha_{22} \sqrt{X}I_{\nu_-}(zX)$ is odd. We expected to have this symmetry for large $z$ since, in the massless case (see \cite{DN}), the second component of Jost function from the right is odd.
\end{remark}

\subsection{Asymptotics of the Jost functions from the left for large z}
\noindent
As for the Jost functions from the right we can prove that the Jost functions from the left satisfy the following estimates.

\begin{theorem}\label{asJostgauche}
 We set
\[ \alpha_L(z) = \left( \frac{z}{2} \right)^{\nu_+} \left( - \frac{\kappa_+}{a_+} \right)^{i \left( \frac{\lambda-c_+}{\kappa_+} \right)} \Gamma(1-\nu_+)\]
and
\[ \beta_L(z) = \left( \frac{z}{2} \right)^{\mu_+} \left(- \frac{\kappa_+}{a_+} \right)^{-i \left( \frac{\lambda-c_+}{\kappa_+} \right)} \Gamma(1-\mu_+).\]
For large $z$ in the complex plane (see remark below), the Jost functions from the left satisfy the following estimates, uniformly on each compact subset of $]0,A[$,
\[ \left| \left| F_{L1}(X,\lambda,z) - \alpha_L(z) \sqrt{A-X} I_{-\nu_+}(z(A-X)) \begin{pmatrix}  1 & 0 \\ 0 & 1   \end{pmatrix} \right|\right| = O \left( \frac{e^{| \mathrm{Re}(z)| (A-X)}}{\vert z \vert} \right),\]
\[ \left| \left| F_{L2}(X,\lambda,z) + i \beta_L(z) \sqrt{A-X} I_{\nu_+}(z(A-X)) \begin{pmatrix}  0 &1 \\ -1 & 0   \end{pmatrix} \right| \right| = O \left( \frac{e^{| \mathrm{Re}(z) | (A-X)}}{|z|} \right) ,\]
\[ \left| \left| F_{L3}(X,\lambda,z) +i \overline{\beta_L(\bar{z})} \sqrt{A-X} I_{\mu_+}(z(A-X)) \begin{pmatrix}  0 &1 \\ -1 & 0   \end{pmatrix} \right| \right| = O \left( \frac{e^{| \mathrm{Re}(z) | (A-X)}}{|z|} \right) ,\]
\[ \left| \left| F_{L4}(X,\lambda,z) - \overline{\alpha_L(\bar{z})} \sqrt{A-X} I_{-\mu_+}(z(A-X))  \begin{pmatrix}  1 & 0 \\ 0 & 1   \end{pmatrix} \right|\right| = O \left( \frac{e^{| \mathrm{Re}(z)| (A-X)}}{\vert z \vert} \right).\]
\end{theorem}

\begin{remark}
 As in the study of the Jost function from the right we have to be careful because the asymptotics of the modified Bessel functions are given in the whole complex plane 
 except near the axis $\R^-$. However using parity/imparity properties (given by the definition of $\alpha_L(z)$ 
 and $\beta_L(z)$) we can extend the asymptotics to the whole complex plane.
\end{remark}

\noindent
As a consequence, using the asymptotic (\ref{asInu}),
we can prove the following asymptotics of the Jost functions from the left for $z \to +\infty$, $z$ real.

\begin{theorem}\label{asJostgaucheR}
 The Jost functions from the left satisfy the following asymptotics for $z \to +\infty$, $z$ real, uniformly on each compact subset of $]0,A[$,
 \[ F_{L1}(X,\lambda,z) = \frac{2^{-\nu_+}}{\sqrt{2\pi}} \left(- \frac{\kappa_+}{a_+} \right)^{i \frac{(\lambda-c_+)}{\kappa_+}} \Gamma(1-\nu_+) z^{-i\frac{(\lambda-c_+)}{\kappa_+}} e^{z(A-X)} \begin{pmatrix}
                                                                                                                                                                                                  1 & O \left( \frac{1}{z} \right) \\ O \left( \frac{1}{z} \right) & 1
                                                                                                                                                                                                 \end{pmatrix}
\left( 1 + O \left( \frac{1}{z} \right) \right),\] 
 \[ F_{L2}(X,\lambda,z) = -i \frac{2^{-\mu_+}}{\sqrt{2\pi}} \left(- \frac{\kappa_+}{a_+} \right)^{-i \frac{(\lambda-c_+)}{\kappa_+}} \Gamma(1-\mu_+) z^{i\frac{(\lambda-c_+)}{\kappa_+}} e^{z(A-X)} \begin{pmatrix}
                                                                                                                                                                                                  O \left( \frac{1}{z} \right) & 1 \\ -1 & O \left( \frac{1}{z} \right)
                                                                                                                                                                                                 \end{pmatrix}\left( 1 + O \left( \frac{1}{z} \right) \right),\] 
 \[ F_{L3}(X,\lambda,z) = -i \frac{2^{-\nu_+}}{\sqrt{2\pi}} \left(- \frac{\kappa_+}{a_+} \right)^{i \frac{(\lambda-c_+)}{\kappa_+}} \Gamma(1-\nu_+) z^{-i\frac{(\lambda-c_+)}{\kappa_+}} e^{z(A-X)}\begin{pmatrix}
                                                                                                                                                                                                  O \left( \frac{1}{z} \right) & 1 \\ -1 & O\left( \frac{1}{z} \right)
                                                                                                                                                                                                 \end{pmatrix} \left( 1 + O \left( \frac{1}{z} \right) \right),\]
 \[ F_{L4}(X,\lambda,z) = \frac{2^{-\mu_+}}{\sqrt{2\pi}} \left(- \frac{\kappa_+}{a_+} \right)^{-i \frac{(\lambda-c_+)}{\kappa_+}} \Gamma(1-\mu_+) z^{i\frac{(\lambda-c_+)}{\kappa_+}} e^{z(A-X)}\begin{pmatrix}
                                                                                                                                                                                                  1 & O \left( \frac{1}{z} \right) \\ O \left( \frac{1}{z} \right) & 1                                                                                                                                                                                    \end{pmatrix} \left( 1 + O \left( \frac{1}{z} \right) \right).\] 
\end{theorem}

\begin{remark}
 Corresponding asymptotics are also true for $z \to - \infty$ by parity/imparity. 
\end{remark}

\subsection{Asymptotics of the matrix of scattering data}
\noindent
We recall that the matrix $\hat{A}_L(\lambda,z)$ is defined by
\[\hat{F}_L(x,\lambda,z) = \hat{F}_R(x,\lambda,z) \hat{A}_{L}(\lambda,z),\]
and, thanks to the equation (\ref{lienALJ}) we see that
\[ \hat{A}_{L}(\lambda,z) = \Gamma^1 \hat{F}_R(x,\lambda,\bar{z})^{\star} \Gamma^1 \hat{F}_L(x,\lambda,z).\]
Moreover
\[F_{R} (x,\lambda,z) = e^{-i \Gamma^1 C^{-}(x)} \hat{F}_R(x,\lambda,z) \quad \text{and} \quad F_L (x,\lambda,z) = e^{-i \Gamma^1 C^{-}(x)} \hat{F}_L(x,\lambda,z),\]
thus,
\[ \hat{A}_{L}(\lambda,z) = \Gamma^1 F_R(x,\lambda,\bar{z})^{\star} \Gamma^1 F_L(x,\lambda,z).\]
Then, the blocks of the matrix $\hat{A}_{L}(\lambda,z)$ satisfy for all $X \in ]0,A[$ the relations
\[\begin{cases}
\hat{A}_{L1}(\lambda,z) = F_{R1}(X,\lambda,\bar{z})^{\star} F_{L1}(X,\lambda,z) - F_{R3}(X,\lambda,\bar{z})^{\star} F_{L3}(X,\lambda,z),\\
\hat{A}_{L2}(\lambda,z) = F_{R1}(X,\lambda,\bar{z})^{\star} F_{L2}(X,\lambda,z) - F_{R3}(X,\lambda,\bar{z})^{\star} F_{L4}(X,\lambda,z),\\
\hat{A}_{L3}(\lambda,z) = -F_{R2}(X,\lambda,\bar{z})^{\star} F_{L1}(X,\lambda,z) + F_{R4}(X,\lambda,\bar{z})^{\star} F_{L3}(X,\lambda,z),\\
\hat{A}_{L4}(\lambda,z) = -F_{R2}(X,\lambda,\bar{z})^{\star} F_{L2}(X,\lambda,z) + F_{R4}(X,\lambda,\bar{z})^{\star} F_{L4}(X,\lambda,z).
\end{cases}\]

\begin{remark}
 We immediatly see, thanks to the symmetries on the asymptotics of the Jost functions, that for $z$ large enough in the complex 
 plane we have $\hat{A}_{L1}(\lambda,z) \sim \hat{A}_{L4}(\lambda,\bar{z})^{\star}$ and $\hat{A}_{L2}(\lambda,z) \sim \hat{A}_{L3}(\lambda,\bar{z})^{\star}$ which are the symmetries 
 true for any $z$ in the massless case studied in \cite{DN}.
\end{remark}

Thanks to Theorems \ref{asJostdroite} and \ref{asJostgauche} and the asymptotics (\ref{asInu}) we obtain that if we set
\[ \alpha(z) =  \left( \frac{\kappa_-}{a_-} \right)^{-i \left( \frac{\lambda-c_-}{\kappa_-} \right)}  \left( - \frac{\kappa_+}{a_+} \right)^{i \left( \frac{\lambda-c_+}{\kappa_+} \right)} \Gamma(1-\mu_-)  \Gamma(1-\nu_+) \left( \frac{z}{2} \right)^{\mu_- + \nu_+}\]
and
\[ \beta(z) =  \left( \frac{\kappa_-}{a_-} \right)^{-i \left( \frac{\lambda-c_-}{\kappa_-} \right)}  \left( - \frac{\kappa_+}{a_+} \right)^{-i \left( \frac{\lambda-c_+}{\kappa_+} \right)} \Gamma(1-\mu_-)  \Gamma(1-\mu_+) \left( \frac{z}{2} \right)^{\mu_- + \mu_+},\]
the blocks of the matrix $\hat{A}_{L}$ satisfy the following asymptotics for large $z$ in the complex plane:
\begin{eqnarray*}
&&\left| \left| \hat{A}_{L1}(\lambda,z) - \alpha(z) \sqrt{X(A-X)} ( I_{-\mu_-}(zX) I_{-\nu_+}(z(A-X)) + I_{\nu_-}(zX)I_{\mu_+}(z(A-X)) )   \begin{pmatrix}  1 & 0 \\ 0 & 1   \end{pmatrix} \right|\right| \\
&&\quad \quad \quad \quad \quad \quad  \quad \quad \quad = O \left( \frac{e^{| \mathrm{Re}(z)| A}}{\vert z \vert} \right),
\end{eqnarray*}
\begin{eqnarray*}
&& \left| \left| \hat{A}_{L2}(\lambda,z) + i \beta(z) \sqrt{X(A-X)} (I_{\nu_-}(zX)I_{-\mu_+}(z(A-X)) + I_{-\mu_-}(zX)I_{\nu_+}(z(A-X))) \begin{pmatrix}  0 &1 \\ -1 & 0   \end{pmatrix} \right| \right| \\
&&\quad \quad \quad \quad \quad \quad  \quad \quad \quad = O \left( \frac{e^{| \mathrm{Re}(z) | A}}{|z|} \right) ,
\end{eqnarray*}
\begin{eqnarray*}
&& \left| \left| \hat{A}_{L3}(\lambda,z) + i \overline{\beta(\bar{z})} \sqrt{X(A-X)} (I_{\mu_-}(zX)I_{-\nu_+}(z(A-X)) + I_{-\nu_-}(zX)I_{\mu_+}(z(A-X))) \begin{pmatrix}  0 &1 \\ -1 & 0   \end{pmatrix} \right| \right| \\
&&\quad \quad \quad \quad \quad \quad  \quad \quad \quad = O \left( \frac{e^{| \mathrm{Re}(z) | A}}{|z|} \right) ,
\end{eqnarray*}
\begin{eqnarray*}
&& \left| \left| \hat{A}_{L4}(\lambda,z) - \overline{\alpha(\bar{z})} \sqrt{X(A-X)} ( I_{-\nu_-}(zX) I_{-\mu_+}(z(A-X)) + I_{\mu_-}(zX)I_{\nu_+}(z(A-X)) )  \begin{pmatrix}  1 & 0 \\ 0 & 1   \end{pmatrix} \right|\right| \\
&&\quad \quad \quad \quad \quad \quad  \quad \quad \quad = O \left( \frac{e^{| \mathrm{Re}(z)| A}}{\vert z \vert} \right).
\end{eqnarray*}

Now, using the asymptotics of the modified Bessel functions (\ref{asInu}), we obtain the asymptotics of the matrix $\hat{A}_L(\lambda,z)$ for large $z$ in the complex plane 
(we can extend these asymptotics to the half line $\mathbb{R}^-$ by parity).

\begin{theorem}\label{asA_L}
The blocks of the matrix $\hat{A}_{L}$ satisfy the following estimates for large $z$ in the complex plane:
$$\begin{array}{rl}
\hat{A}_{L1}(\lambda,z) =&  \frac{1}{2\pi} \left( - \frac{\kappa_+}{a_+} \right)^{i \frac{(\lambda-c_+)}{\kappa_+}} \left( \frac{\kappa_-}{a_-} \right)^{-i \frac{(\lambda-c_-)}{\kappa_-}} \Gamma(1-\nu_+) \Gamma(1-\mu_-) \left( \frac{z}{2} \right)^{i\left(\frac{(\lambda-c_-)}{\kappa_-}-\frac{(\lambda-c_+)}{\kappa_+} \right)} \\
& \times  \left( e^{zA} + e^{-zA} e^{-\mathrm{sg}(\mathrm{Im}(z)) \pi \left( \frac{\lambda-c_+}{\kappa_+} - \frac{\lambda-c_-}{\kappa_-} \right)} \right) \begin{pmatrix}  1 & O \left( \frac{1}{z} \right) \\ O \left( \frac{1}{z} \right) & 1  \end{pmatrix} \left( 1 + O \left( \frac{1}{z} \right) \right),\\
\hat{A}_{L2}(\lambda,z) =& \frac{-i}{2\pi} \left( - \frac{\kappa_+}{a_+} \right)^{-i \frac{(\lambda-c_+)}{\kappa_+}} \left( \frac{\kappa_-}{a_-} \right)^{-i \frac{(\lambda-c_-)}{\kappa_-}} \Gamma(1-\mu_+) \Gamma(1-\mu_-) \left( \frac{z}{2} \right)^{i\left(\frac{(\lambda-c_-)}{\kappa_-}+\frac{(\lambda-c_+)}{\kappa_+} \right)} \\
& \times \left( e^{zA} - e^{-zA} e^{\mathrm{sg}(\mathrm{Im}(z)) \pi \left( \frac{\lambda-c_+}{\kappa_+} - \frac{\lambda-c_-}{\kappa_-} \right)} \right) \begin{pmatrix}  O \left( \frac{1}{z} \right) & 1 \\ -1 & O \left( \frac{1}{z} \right)   \end{pmatrix}\left( 1 + O \left( \frac{1}{z} \right) \right),\\
\hat{A}_{L3}(\lambda,z) =& \frac{-i}{2\pi} \left( - \frac{\kappa_+}{a_+} \right)^{i \frac{(\lambda-c_+)}{\kappa_+}} \left( \frac{\kappa_-}{a_-} \right)^{i \frac{(\lambda-c_-)}{\kappa_-}} \Gamma(1-\nu_+) \Gamma(1-\nu_-) \left( \frac{z}{2} \right)^{-i\left(\frac{(\lambda-c_-)}{\kappa_-}+\frac{(\lambda-c_+)}{\kappa_+} \right)} \\
& \times \left( e^{zA} - e^{-zA} e^{-\mathrm{sg}(\mathrm{Im}(z)) \pi \left( \frac{\lambda-c_+}{\kappa_+} + \frac{\lambda-c_-}{\kappa_-} \right)} \right) \begin{pmatrix}  O \left( \frac{1}{z} \right) & 1 \\ -1 & O \left( \frac{1}{z} \right)   \end{pmatrix} \left( 1 + O \left( \frac{1}{z} \right) \right),\\
\hat{A}_{L4}(\lambda,z) =& \frac{1}{2\pi} \left( - \frac{\kappa_+}{a_+} \right)^{-i \frac{(\lambda-c_+)}{\kappa_+}} \left( \frac{\kappa_-}{a_-} \right)^{i \frac{(\lambda-c_-)}{\kappa_-}} \Gamma(1-\mu_+) \Gamma(1-\nu_-) \left( \frac{z}{2} \right)^{-i\left(\frac{(\lambda-c_-)}{\kappa_-}-\frac{(\lambda-c_+)}{\kappa_+} \right)} \\
& \times \left( e^{zA} + e^{-zA} e^{\mathrm{sg}(\mathrm{Im}(z)) \pi \left( \frac{\lambda-c_+}{\kappa_+} - \frac{\lambda-c_-}{\kappa_-} \right)} \right) \begin{pmatrix}  1 & O \left( \frac{1}{z} \right) \\ O \left( \frac{1}{z} \right) & 1  \end{pmatrix} \left( 1 + O \left( \frac{1}{z} \right) \right).
\end{array}$$
\end{theorem}

\begin{remark}
 For large $z$ in $\mathbb{R}^+$ we obtain the analoguous of the asymptotics given in \cite{DN}, Theorem 4.19, in the massless case.
\end{remark}

\section{The Complex Angular Momentum method}
\noindent
In this Section, we follow Section 3 of \cite{DN} and use the idea of the complex angular momentum method. 
First we note that the Jost functions $F_L(x,\lambda,z)$ and $F_R(x,\lambda,z)$ are analytic in the $z \in \C$ variable. Indeed, for instance, we recall that,
\[F_{R} (x,\lambda,z) = e^{-i \Gamma^1 C^{-}(x)} \hat{M}_R(x,\lambda,z) e^{i \lambda \Gamma^1 x}.\]
Moreover, using series expansion (see (\ref{serieM_R1})) we obtain the analyticity of the blocks of $\hat{M}_R$ since the terms of the series
are polynomial in the variable $z$. Thus the Jost functions are also analytic.\\
The application $z \mapsto \hat{A}_L(\lambda,z)$ is also analytic on the complex plane $\mathbb{C}$. Indeed, by definition
\[F_L(x,\lambda,z) = F_R(x,\lambda,z) \hat{A}_{L}(\lambda,z).\]
and $\det(F_R)= 1$. Thus the coefficients of the matrix $\hat{A}_{L}(\lambda,z)$ are combinations of the components of the Jost functions. Then, these coefficients are 
analytic functions of the variable $z$.

We now prove that the coefficients $\hat{A}_{Li,j}$ of the matrix $\hat{A}_L$ (see Section 3.2 for the notation) belong to the Nevanlinna class when restricted to the half 
plane $\Pi^+ = \{ z \in \mathbb{C}, \, \mathrm{Re}(z) > 0 \}$. As an application we prove the uniqueness results mentionned in Introduction.

Recall first that the Nevanlinna class $N(\Pi^+)$ is defined as the set of all analytic functions $f(z)$ on $\Pi^+$ that satisfy the estimate
\[ \underset{0 < r < 1}{\sup} \int_{-\pi}^{\pi} \ln^+ \left| f \left( \frac{1-re^{i\varphi}}{1+re^{i\varphi}} \right) \right| \, \mathrm d\varphi < \infty,\]
where $\ln^+(x) = \ln(x)$ if $\ln(x) \geq 0$ and $0$ if $\ln(x) < 0$. We shall use the following result proved in \cite{Ram}.

\begin{lemma}\label{Nev}
 Let $h \in H(\Pi^+)$ be a holomorphic function in $\Pi^+$ satisfying
\[\vert h(z) \vert \leq C e^{A \mathrm{Re}(z)}, \quad \forall z \in \Pi^{+},\]
where $A$ and $C$ are two constants. Then $h \in N(\Pi^+)$.
\end{lemma}


As a consequence of Lemma \ref{Nev} and Theorem \ref{asA_L}, we thus get,

\begin{coro}\label{ALN}
 For each $\lambda \in \mathbb{R}$ fixed, the applications $z \mapsto \hat{A}_{Li,j}(\lambda,z)_{\vert \Pi^+}$ belong to $N(\Pi^+)$ for $(i,j) \in \{1,2,3,4\}^2$.
\end{coro}

\begin{proof}
 This is an immediate consequence of Theorem \ref{asA_L}.
\end{proof}


We now recall the following result proved in \cite{Ram}, Theorem 1.3.

\begin{theorem}[\cite{Ram}, Thm. 1.3]\label{thmNev}
 Let $h \in N(\Pi^+)$ satisfying $h(n) = 0$ for all $n \in \mathcal{L}$ where $\mathcal{L} \subset \mathbb{N}^{\star}$ with $\sum_{n \in \mathcal{L}} \frac{1}{n} = \infty$. Then $h \equiv 0$ in $N(\Pi^+)$.
\end{theorem}

%
%

In other words, thanks to Corollary \ref{ALN} and Theorem \ref{thmNev}, the scattering data $\hat{A}_{Li,j}(\lambda,z)$ are uniquely determined as functions of $z \in \mathbb{C}$ from their values on a subset $\mathcal{L}$ of the integers that satisfies the M\"untz condition 
$\sum_{n \in \mathcal{L}} \frac{1}{n} = \infty$. In our proof of the main Theorem we need the following result.

\begin{coro}\label{corL}
  Consider two dS-RN black holes and denote by $\hat{L}$, $\hat{R}$ and $\tilde{\hat{L}}$, $\tilde{\hat{L}}$ the corresponding reflection coefficients. Let $\mathcal{L} \subset \mathbb{N}^{\star}$ satisfying 
$\sum_{n \in \mathcal{L}} \frac{1}{n} = \infty$. Assume that the following equality holds
\[\hat{L}(\lambda,n) = \tilde{\hat{L}}(\lambda,n), \quad \forall n \in \mathcal{L} \quad( \mathrm{respectively} \quad \hat{R}(\lambda,n) = \tilde{\hat{R}}(\lambda,n), \quad \forall n \in \mathcal{L}).\]
Then, if we denote by $P$ the set of poles of $\hat{L}$ (respectively $Q$ the set of poles of $\hat{R}$) we have,
\[\hat{L}(\lambda,z) = \tilde{\hat{L}}(\lambda,z), \quad \forall z \in \mathbb{C} \setminus P \quad( \mathrm{respectively} \quad \hat{R}(\lambda,z) = \tilde{\hat{R}}(\lambda,z), \quad \forall z \in \mathbb{C} \setminus Q).\]
\end{coro}

\begin{proof}
Since the proof is similar, we just show the result for the reflection coefficient $\hat{L}$. We recall (see (\ref{defL})) that
 \[ \hat{L}(\lambda,n) = \hat{A}_{L3}(\lambda,n) \hat{A}_{L1}(\lambda,n)^{-1} = -\hat{A}_{R4}(\lambda,n)^{-1} \hat{A}_{R3}(\lambda,n).\]
Thus, thanks to the hypothesis, we immediatly obtain
\[\tilde{\hat{A}}_{R4}(\lambda,n)  \hat{A}_{L3}(\lambda,n)  = -\tilde{\hat{A}}_{R3}(\lambda,n) \hat{A}_{L1}(\lambda,n), \quad \forall n \in \mathcal{L}.\]
Using Corollary \ref{ALN} and Theorem \ref{thmNev}, we deduce from this equality that
\[ \tilde{\hat{A}}_{R4}(\lambda,z) \hat{A}_{L3}(\lambda,z)  = -\tilde{\hat{A}}_{R3}(\lambda,z) \hat{A}_{L1}(\lambda,z), \quad \forall z \in \C.\]
Finally, for all $z \in \C \setminus P$ (we note that thanks to the previous equality $P = \tilde{P}$, where $\tilde{P}$ is the set of poles of $\tilde{\hat{L}}$),
\[ \hat{L}(\lambda,z) = \hat{A}_{L3}(\lambda,z) \hat{A}_{L1}(\lambda,z)^{-1} = -\tilde{\hat{A}}_{R3}(\lambda,z) \tilde{\hat{A}}_{R4}(\lambda,z)^{-1} = \tilde{\hat{L}}(\lambda,z).\]
\end{proof}

\section{Proof of the main Theorem}\label{preuveL}
\noindent
The aim of this Section is to prove Theorem \ref{mainthm}. Let us first assume that the assertion $(i)$ is satisfied. Hence we have to prove the uniqueness of the parameters $(M,Q,\Lambda)$ of a dS-RN black hole from
the knowledge of the reflection coefficients of the partial scattering operators, $L(\lambda,n)$, for a fixed energy $\lambda \in \mathbb{R}$ and for all $n \in \mathcal{L} \subset \mathbb{N}^{\star}$ 
satisfying the M\"{u}ntz condition $\sum_{n \in \mathcal{L}} \frac{1}{n} = \infty$.

Consider thus two dS-RN black holes with parameters $(M,Q,\Lambda)$ and $(\tilde{M},\tilde{Q},\tilde{\Lambda})$ respectively. We shall denote 
by $a(x)$, $b(x)$ and $c(x)$ and $\tilde{a}(x)$, $\tilde{b}(x)$ and $\tilde{c}(x)$ the corresponding potentials appearing in the Dirac equation (\ref{eqDir}) (see (\ref{ham})). Recall that they satisfy 
the hypotheses of the Section \ref{Dirpb}, (\ref{defpot}), and the estimates of Lemma \ref{aspot}.
We assume that
 \[ L(\lambda,n) = \tilde{L}(\lambda,n), \quad \forall n \in \mathcal{L}.\]
Thus (see Proposition \ref{LienSS}),
\[\hat{L}(\lambda,n) = \tilde{\hat{L}}(\lambda,n), \quad \forall n \in \mathcal{L}.\]

\begin{lemma}\label{lemA=Atilde}
 Suppose that
 \[\hat{L}(\lambda,n) = \tilde{\hat{L}}(\lambda,n), \quad \forall n \in \mathcal{L}.\]
 Then
 \[A:= \int_{\mathbb{R}} a(t)\, \mathrm dt = \int_{\mathbb{R}} \tilde{a}(t)\, \mathrm dt =: \tilde{A}.\]
\end{lemma}

\begin{proof}
 We recall that (see (\ref{eqAL}) and (\ref{defL})),
 \[ \hat{L}(\lambda,n) = \hat{A}_{L3}(\lambda,n) \hat{A}_{L1}(\lambda,n)^{-1}\]
 and
\[ \hat{A}_{L1}(\lambda,n)^{\star} \hat{A}_{L1}(\lambda,n) = I_{2} + \hat{A}_{L3}(\lambda,n)^{\star} \hat{A}_{L3}(\lambda,n).\]
Note that, thanks to their asymptotics given in Theorem \ref{asA_L}, the blocks $\hat{A}_{Li}(\lambda,n)$, $i \in \{1,2,3,4\}$, 
of the matrix $\hat{A}_{L}(\lambda,n)$, are invertible if $n$ is large enough. \\
Thus, for $n$ large enough,
\begin{eqnarray*}
 \hat{A}_{L1}(\lambda,n)^{\star} &=& \hat{A}_{L1}(\lambda,n)^{-1} + \hat{A}_{L3}(\lambda,n)^{\star} \hat{L}(\lambda,n)\\
 &=& \hat{A}_{L1}(\lambda,n)^{-1} + \hat{A}_{L3}(\lambda,n)^{\star} \tilde{\hat{L}}(\lambda,n)\\
 &=& \hat{A}_{L1}(\lambda,n)^{-1} + \hat{A}_{L3}(\lambda,n)^{\star} \left((\tilde{\hat{A}}_{L3}(\lambda,n)^{\star})^{-1}( \tilde{\hat{A}}_{L1}(\lambda,n)^{\star} -  \tilde{\hat{A}}_{L1}(\lambda,n)^{-1})\right).
\end{eqnarray*}
 Thus,
 \[(\hat{A}_{L3}^{\star})^{-1} \hat{A}_{L1}^{\star} -  (\hat{A}_{L3}^{\star})^{-1}\hat{A}_{L1}^{-1} = (\tilde{\hat{A}}_{L3}^{\star})^{-1} \tilde{\hat{A}}_{L1}^{\star} -  (\tilde{\hat{A}}_{L3}^{\star})^{-1}\tilde{\hat{A}}_{L1}^{-1}.\]
 Moreover, the hypothesis implies that, 
 \[(L(\lambda,n)^{\star})^{-1} = (\tilde{L}(\lambda,n)^{\star})^{-1},\]
 i.e.
 \[(\hat{A}_{L3}^{\star})^{-1} \hat{A}_{L1}^{\star} = (\tilde{\hat{A}}_{L3}^{\star})^{-1} \tilde{\hat{A}}_{L1}^{\star}.\]
 Thus,
 \[ (\hat{A}_{L3}^{\star})^{-1}\hat{A}_{L1}^{-1} = (\tilde{\hat{A}}_{L3}^{\star})^{-1}\tilde{\hat{A}}_{L1}^{-1}.\]
 Finally, using the asymptotics given in Theorem \ref{asA_L}, we obtain that $A = \tilde{A}$.
 \end{proof}

Hence, we can define the diffeomorphisms $h, \tilde{h}: ]0,A[ \rightarrow \mathbb{R}$ as the inverse of the Liouville 
transforms $g$ and $\tilde{g}$ in which we use the potentials $a(x)$ and $\tilde{a}(x)$ respectively.
Now, as in \cite{DN}, we follow a strategy inspired by \cite{FY}. Let us introduce for $X \in ]0,A[$ the $4 \times 4$ matrix 
\[ P(X,\lambda,z) = \begin{pmatrix} P_1(X,\lambda,z) & P_2(X,\lambda,z) \\ P_3(X,\lambda,z) & P_4(X,\lambda,z) \end{pmatrix}, \quad P_j(X,\lambda,z) \in M_2(\mathbb{C}), \quad j \in \{1,2,3,4\},\]
defined by
\[P(X,\lambda,z)  \tilde{F}_R(\tilde{h}(X),\lambda,z) = F_R(h(X),\lambda,z).\]
We first prove that the matrix $P(X,\lambda,z)$ is constant equal to $\pm I_4$ using the Phragm\'en-Lindel\"{o}f's Theorem and Liouville's Theorem thanks to the good 
estimates on the coefficients of the matrix $P(X,\lambda,z)$ as well as their analyticity with respect to $z$. After that, we obtain two equalities on scalar functions 
depending on the potentials and we deduce from the explicit form of the potentials the uniqueness of the parameters $M$, $Q$ and $\Lambda$.

\subsection{Study of the matrix P}\label{P=I}
\noindent
We first recall that
\[F_{R} (x,\lambda,z) = e^{-i \Gamma^1 C^{-}(x)} \hat{F}_R(x,\lambda,z).\]
Thus,
\begin{eqnarray*}
 P(X,\lambda,z) &=& F_R(h(X),\lambda,z) \tilde{F}_R(\tilde{h}(X),\lambda,z)^{-1} \\
&=& e^{-i \Gamma^1 C^{-}(h(X))} \hat{F}_R(x,\lambda,z) \tilde{\hat{F}}_R(\tilde{h}(X),\lambda,z)^{-1} e^{i \Gamma^1 \tilde{C}^{-}(\tilde{h}(X))}.
\end{eqnarray*}
We know (see (\ref{eqJ})) that
\[ \tilde{\hat{F}}_R(\tilde{h}(X),\lambda,z)^{-1} = \Gamma^1 \tilde{\hat{F}}_R(\tilde{h}(X),\lambda,z)^{\star} \Gamma^1.\]
Thus,
\begin{eqnarray*}
 P(X,\lambda,z) &=& e^{-i \Gamma^1 C^{-}(h(X))} \hat{F}_R(h(X),\lambda,z) \Gamma^1 \tilde{\hat{F}}_R(\tilde{h}(X),\lambda,z)^{\star} \Gamma^1 e^{i \Gamma^1 \tilde{C}^{-}(\tilde{h}(X))}\\
&=& e^{-i \Gamma^1 C^{-}(h(X))} \begin{pmatrix} \hat{F}_{R1} \tilde{\hat{F}}_{R1}^{\star} - \hat{F}_{R2} \tilde{\hat{F}}_{R2}^{\star}  & -\hat{F}_{R1} \tilde{\hat{F}}_{R3}^{\star}+\hat{F}_{R2} \tilde{\hat{F}}_{R4}^{\star} \\ \hat{F}_{R3} \tilde{\hat{F}}_{R1}^{\star} - \hat{F}_{R4} \tilde{\hat{F}}_{R2}^{\star} & -\hat{F}_{R3} \tilde{\hat{F}}_{R3}^{\star} + \hat{F}_{R4} \tilde{\hat{F}}_{R4}^{\star}  \end{pmatrix}e^{i \Gamma^1 \tilde{C}^{-}(\tilde{h}(X))}\\
&=& \begin{pmatrix} F_{R1} \tilde{F}_{R1}^{\star} - F_{R2} \tilde{F}_{R2}^{\star}  & -F_{R1} \tilde{F}_{R3}^{\star}+F_{R2} \tilde{F}_{R4}^{\star} \\ F_{R3} \tilde{F}_{R1}^{\star} - F_{R4} \tilde{F}_{R2}^{\star} & -F_{R3} \tilde{F}_{R3}^{\star} + F_{R4} \tilde{F}_{R4}^{\star}  \end{pmatrix}.
\end{eqnarray*}

\begin{lemma}\label{Pholo}
 For all $(i,j) \in \{1,2,3,4\}^2$, the applications $z \mapsto P_{i,j}(X,\lambda,z)$ are analytic on $\mathbb{C}$ and of exponential type (i.e. that there exist some constants $c$ and $C$ such that $\vert P_{i,j}(X,\lambda,z) \vert \leq c e^{C\vert z \vert}$).
\end{lemma}

\begin{proof}
 Since the Jost functions are analytic on the whole complex plane, the Lemma is an easy consequence of the previous equality. The fact that the components $P_{i,j}(X,\lambda,z)$ are of exponential type is a consequence of Theorems \ref{asJostdroite}, \ref{asJostgauche} and \ref{asA_L}. 
\end{proof}

We just study $P_1$ and $P_2$ since the study of $P_3$ and $P_4$ is similar. We know that 
\[P_1(X,\lambda,z) = F_{R1} \tilde{F}_{R1}^{\star} - F_{R2} \tilde{F}_{R2}^{\star}\]
and
\[P_2(X,\lambda,z) = -F_{R1} \tilde{F}_{R3}^{\star}+F_{R2} \tilde{F}_{R4}^{\star}.\]
We now use
\begin{eqnarray*}
 F_L(x,\lambda,z) &=& F_R(x,\lambda,z) \hat{A}_L(\lambda,z) \\
&=& \begin{pmatrix} F_{R1}\hat{A}_{L1} + F_{R2}\hat{A}_{L3} & F_{R1} \hat{A}_{L2} + F_{R2}\hat{A}_{L4} \\ F_{R3}\hat{A}_{L1} + F_{R4}\hat{A}_{L3} & F_{R3} \hat{A}_{L2} + F_{R4}\hat{A}_{L4} \end{pmatrix}.
\end{eqnarray*}
Thus, since
\[F_{L1} = F_{R1}\hat{A}_{L1} + F_{R2}\hat{A}_{L3},\]
we obtain
\[F_{R1} = (F_{L1} - F_{R2}\hat{A}_{L3})\hat{A}_{L1}^{-1} = F_{L1}\hat{A}_{L1}^{-1} - F_{R2}\hat{L}(\lambda,z)\]
and, thanks to
\[F_{L2} = F_{R1} \hat{A}_{L2} + F_{R2}\hat{A}_{L4},\]
we also obtain
\[F_{R2} = (F_{L2} - F_{R1}\hat{A}_{L2})\hat{A}_{L4}^{-1} = F_{L2}\hat{A}_{L4}^{-1} - F_{R1}\hat{A}_{L2} \hat{A}_{L4}^{-1}.\]
Note that, thanks to their asymptotics the matrix $A_{Li}$ and $A_{Ri}$, for $i \in \{1,2,3,4\}$, are invertible for real $z$ sufficiently large.
Then, for $z$ real large enough,
\[\tilde{F}_{R2}^{\star} = (\tilde{\hat{A}}_{L4}^{-1})^{\star} \tilde{F}_{L2}^{\star} - (\tilde{\hat{A}}_{L4}^{-1})^{\star} \tilde{\hat{A}}_{L2}^{\star} \tilde{F}_{R1}^{\star}.\]
Moreover, thanks to Equations (\ref{eqALR}) and (\ref{eqALR2}),
\[ (\tilde{\hat{A}}_{L4}^{-1})^{\star} \tilde{\hat{A}}_{L2}^{\star} = -\tilde{\hat{A}}_{R4}^{-1} \tilde{\hat{A}}_{R3} = \tilde{\hat{L}}.\]
Finally,
\begin{eqnarray*}
 P_1(X,\lambda,z) &=& F_{R1} \tilde{F}_{R1}^{\star} - F_{R2} \tilde{F}_{R2}^{\star} \\
&=& ( F_{L1}\hat{A}_{L1}^{-1} - F_{R2}\hat{L}(\lambda,z)) \tilde{F}_{R1}^{\star} - F_{R2} ((\tilde{\hat{A}}_{L4}^{-1})^{\star} \tilde{F}_{L2}^{\star} - \tilde{\hat{L}}(\lambda,z) \tilde{F}_{R1}^{\star})\\
&=&  F_{L1}\hat{A}_{L1}^{-1} \tilde{F}_{R1}^{\star} - F_{R2} (\tilde{\hat{A}}_{L4}^{-1})^{\star} \tilde{F}_{L2}^{\star} + F_{R2} (\tilde{\hat{L}}(\lambda,z) - \hat{L}(\lambda,z))\tilde{F}_{R1}^{\star}.
\end{eqnarray*}
Thus, since $\hat{L}$ and $\tilde{\hat{L}}$ have no singularities on $\mathbb{R}$, Corollary \ref{corL} shows that $\hat{L}(\lambda,z) = \tilde{\hat{L}}(\lambda,z)$ for all $z \in \mathbb{R}$. Then, we obtain that for all 
$z \in \mathbb{R}$ and for all $X \in ]0,A[$,
\begin{eqnarray}\label{expresP1}
 P_1(X,\lambda,z) =  F_{L1}\hat{A}_{L1}^{-1} \tilde{F}_{R1}^{\star} - F_{R2} (\tilde{\hat{A}}_{L4}^{-1})^{\star} \tilde{F}_{L2}^{\star}.
\end{eqnarray}
For $P_2$, we obtain similarly
\begin{eqnarray}\label{expresP2}
P_2(X,\lambda,z) = - F_{L1}\hat{A}_{L1}^{-1} \tilde{F}_{R3}^{\star} + F_{R2} (\tilde{\hat{A}}_{L4}^{-1})^{\star} \tilde{F}_{L4}^{\star}.
\end{eqnarray}

Thus, Theorems \ref{asJostdroite}, \ref{asJostgauche} and \ref{asA_L} (and Section \ref{partDuh}) show that these applications are bounded on 
$\mathbb{R}$ and $i \mathbb{R}$. Finally applying the Phragm\'{e}n-Lindel\"{o}f's Theorem (\cite{Bo}, Theorem 1.4.2) on each quadrant of the complex plane, we deduce that, $z \mapsto P_{i,j}(X,\lambda,z)$ is bounded on $\mathbb{C}$. 
By Liouville's Theorem, we thus obtained that the applications $z \mapsto P_{i,j}(X,\lambda,z)$ are constants on $\mathbb{C}$. Contrary to \cite{DN} we can't use the evaluation on $z = 0$ because we don't have explicit formula for
$P(X,\lambda,0)$. To obtain that $P(X,\lambda,z) = \pm I_4$ we use the asymptotics of the Jost functions. First, by definition
\[P(X,\lambda)  \tilde{F}_R(\tilde{h}(X),\lambda,z) = F_R(h(X),\lambda,z),\]
so
\begin{equation}\label{eqP_1P_21}
P_1(X,\lambda)  \tilde{F}_{R1}(\tilde{h}(X),\lambda,z) + P_2(X,\lambda)  \tilde{F}_{R3}(\tilde{h}(X),\lambda,z) = F_{R1}(h(X),\lambda,z)
\end{equation}
and
\begin{equation}\label{eqP_1P_22}
P_1(X,\lambda)  \tilde{F}_{R2}(\tilde{h}(X),\lambda,z) + P_2(X,\lambda)  \tilde{F}_{R4}(\tilde{h}(X),\lambda,z) = F_{R2}(h(X),\lambda,z).
\end{equation}
Thanks to the Theorem \ref{asJostdroiteR} we know that for every fixed $X \in ]0,A[$, for large $z$, $z$ real,
\[F_{R1}(h(X),\lambda,z) = \frac{2^{-\nu_-}}{\sqrt{2\pi}} \left( \frac{\kappa_-}{a_-} \right)^{i \frac{(\lambda-c_-)}{\kappa_-}} \Gamma(1-\nu_-) z^{-i\frac{(\lambda-c_-)}{\kappa_-}} e^{zX} \begin{pmatrix}
                                                                                                                                                                                               1 & O\left( \frac{1}{z} \right) \\ O\left( \frac{1}{z} \right) & 1
                                                                                                                                                                                              \end{pmatrix} \left( 1 + O \left( \frac{1}{z} \right) \right)\]
and
\[ F_{R3}(h(X),\lambda,z) = i \frac{2^{-\nu_-}}{\sqrt{2\pi}} \left( \frac{\kappa_-}{a_-} \right)^{i \frac{(\lambda-c_-)}{\kappa_-}} \Gamma(1-\nu_-) z^{-i\frac{(\lambda-c_-)}{\kappa_-}} e^{zX} \begin{pmatrix}
                                                                                                                                                                                                O\left( \frac{1}{z} \right) &1 \\ -1 & O\left( \frac{1}{z} \right)
                                                                                                                                                                                              \end{pmatrix}
 \left( 1 + O \left( \frac{1}{z} \right) \right).\]
Thus, setting
\[ \alpha = \frac{2^{-\nu_-}}{\sqrt{2\pi}} \left( \frac{\kappa_-}{a_-} \right)^{i \frac{(\lambda-c_-)}{\kappa_-}} \Gamma(1-\nu_-) \quad \mathrm{and} \quad \tilde{\alpha} = \frac{2^{-\tilde{\nu}_-}}{\sqrt{2\pi}} \left( \frac{\tilde{\kappa}_-}{\tilde{a}_-} \right)^{i \frac{(\lambda-\tilde{c}_-)}{\tilde{\kappa}_-}} \Gamma(1-\tilde{\nu}_-),\]
we obtain, using (\ref{eqP_1P_21}) and just keeping the main terms, that for every fixed $X \in ]0,A[$, for large real $z$,
\[ \left( \tilde{\alpha} P_1(X,\lambda) I_2+ i \tilde{\alpha} P_2(X,\lambda) \begin{pmatrix} 0 & 1 \\ -1 & 0 \end{pmatrix} \right) z^{i \left( \frac{\lambda-c_-}{\kappa_-} - \frac{\lambda-\tilde{c}_-}{\tilde{\kappa}_-} \right)}  \sim \alpha I_2.\]
We deduce that
\begin{eqnarray}\label{egexpo}
 \frac{\lambda-c_-}{\kappa_-} = \frac{\lambda-\tilde{c}_-}{\tilde{\kappa}_-},
\end{eqnarray}
then
\begin{eqnarray}\label{egnumu}
 \nu_- = \tilde{\nu}_- \quad \mathrm{and} \quad \tilde{\mu} = \mu.
\end{eqnarray}
Thus,
\begin{eqnarray*}
\tilde{\alpha} &=& \frac{2^{-\tilde{\nu}_-}}{\sqrt{2\pi}} \left( \frac{\tilde{\kappa}_-}{\tilde{a}_-} \right)^{i \frac{(\lambda-\tilde{c}_-)}{\tilde{\kappa}_-}} \Gamma(1-\tilde{\nu}_-)\\
&=& \alpha \left( \frac{\tilde{\kappa}_-}{\tilde{a}_-} \frac{a_-}{\kappa_-} \right)^{i \frac{(\lambda-\tilde{c}_-)}{\tilde{\kappa}_-}}.
\end{eqnarray*}
Then, for all $X \in ]0,A[$,
\begin{equation}\label{eq1P}
  \left( P_1(X,\lambda) + i P_2(X,\lambda) \begin{pmatrix} 0 & 1 \\ -1 & 0 \end{pmatrix} \right) \left( \frac{\tilde{\kappa}_-}{\tilde{a}_-} \frac{a_-}{\kappa_-} \right)^{i \frac{(\lambda-c_-)}{\kappa_-}}  = I_2.
\end{equation}
Similarly, using (\ref{eqP_1P_22}), for all $X \in ]0,A[$,
\begin{equation}\label{eq2P}
 \left( P_1(X,\lambda) + iP_2(X,\lambda)\begin{pmatrix} 0 & 1 \\ -1 & 0 \end{pmatrix}  \right) \left( \frac{\tilde{\kappa}_-}{\tilde{a}_-} \frac{a_-}{\kappa_-} \right)^{-i \frac{(\lambda-c_-)}{\kappa_-}}  = I_2.
\end{equation}
Thus, using Equations (\ref{eq1P}) and (\ref{eq2P}),
\begin{eqnarray}\label{egpm1}
 \left( \frac{\tilde{\kappa}_-}{\tilde{a}_-} \frac{a_-}{\kappa_-} \right)^{-i \frac{(\lambda-c_-)}{\kappa_-}} = \left( \frac{\tilde{\kappa}_-}{\tilde{a}_-} \frac{a_-}{\kappa_-} \right)^{i \frac{(\lambda-c_-)}{\kappa_-}} = \pm 1.
\end{eqnarray}
Finally, using Equations (\ref{egexpo})-(\ref{egnumu}) and (\ref{egpm1}), the asymptotics of the Jost functions and the fact that $z \mapsto P(X,\lambda,z)$ is constant,
we obtain, thanks to Equations (\ref{expresP1}) and (\ref{expresP2}), that
\[P_1(X,\lambda,z) = \pm I_2 \quad \quad \mathrm{and} \quad \quad P_2(X,\lambda,z) = 0.\]
Similarly we show that
\[P_3(X,\lambda,z) = 0 \quad \quad \mathrm{and} \quad \quad P_4(X,\lambda,z) = \pm I_2 .\]
In addition,
\[P_1(X,\lambda,z) = \left( \frac{\tilde{\kappa}_-}{\tilde{a}_-} \frac{a_-}{\kappa_-} \right)^{-i \frac{(\lambda-c_-)}{\kappa_-}}I_2 = P_{4}(X,\lambda,z).\]
Finally, for all $z \in \mathbb{C}$ and for all $X \in ]0,A[$,
\[P(X,\lambda,z) = \pm I_4.\]

\subsection{Proof of Theorem \ref{mainthm} under the first assumption}\label{conclu}
\noindent
We recall that we work with the operator
\[H = \Gamma^1 D_x + za(x) \Gamma^2 + b(x) \Gamma^0 + c(x).\]
Since we work with the Liouville's variable $X$, we introduce the operator
\[L = a(X) \Gamma^1 D_X + za(X) \Gamma^2 + b(X) \Gamma^0 + c(X)\]
where we use the notations $a(X) = a(h(X))$, $b(X) = b(h(X))$ and $c(X) = c(h(X))$. Note that 
\[ H F(x,\lambda,z) = \lambda F(x,\lambda,z) \Leftrightarrow L F(h(X),\lambda,z) = \lambda F(h(X),\lambda,z).\]
Thus, by definition of the Jost functions,
\[\Gamma^1 D_X(F_R(h(X),\lambda,z)) = \left( -z \Gamma^2 - \frac{b(X)}{a(X)} \Gamma^0  - \frac{c(X)}{a(X)}  + \frac{\lambda}{a(X)} \right)F_R(h(X),\lambda,z)\]
and similarly
\[\Gamma^1 D_X(\tilde{F}_R(\tilde{h}(X),\lambda,z)) = \left( -z \Gamma^2 - \frac{\tilde{b}(X)}{\tilde{a}(X)} \Gamma^0  - \frac{\tilde{c}(X)}{\tilde{a}(X)}  + \frac{\lambda}{\tilde{a}(X)} \right)\tilde{F}_R(\tilde{h}(X),\lambda,z).\]
Moreover, in Section \ref{P=I}, we have shown that
\[F_R(h(X),\lambda,z) = \pm \tilde{F}_R(\tilde{h}(X),\lambda,z), \quad \forall X \in ]0,A[.\]
Then,
\[\left( \left(\frac{\tilde{b}(X)}{\tilde{a}(X)} - \frac{b(X)}{a(X)} \right) \Gamma^0 + \left(\frac{\tilde{c}(X)}{\tilde{a}(X)} - \frac{c(X)}{a(X)} \right) + \lambda \left(\frac{1}{a(X)} - \frac{1}{\tilde{a}(X)} \right)\right) F_R(h(X),\lambda,z) =0, \quad \forall X \in ]0,A[.\]
Thanks to the definition of the matrix
\[ \Gamma^0 = \begin{pmatrix} 0 & \sigma \\ \sigma^{\star} & 0 \end{pmatrix} \quad \quad \mathrm{with} \quad \quad \sigma = \begin{pmatrix} -i & 0 \\ 0 & i \end{pmatrix},\]
we easily obtain
\[ \begin{pmatrix} \left( \left(\frac{\tilde{c}}{\tilde{a}} - \frac{c}{a} \right) + \lambda \left( \frac{1}{a}-\frac{1}{\tilde{a}} \right) \right) I_2 & \left( \frac{\tilde{b}}{\tilde{a}} - \frac{b}{a} \right) \sigma \\ \left( \frac{\tilde{b}}{\tilde{a}} - \frac{b}{a} \right)\sigma^{\star} & \left( \left(\frac{\tilde{c}}{\tilde{a}} - \frac{c}{a} \right) + \lambda \left( \frac{1}{a}-\frac{1}{\tilde{a}} \right) \right) I_2 \end{pmatrix} F_R(h(X),\lambda,z) = 0.\]
Since the Jost function form the right $F_R$ is invertible we deduce that,
\begin{eqnarray}\label{eq1}
 \frac{ \tilde{c} - \lambda}{\tilde{a}}(\tilde{h}(X)) = \frac{ c - \lambda}{a}(h(X)), \quad \forall X \in ]0,A[
\end{eqnarray}
and
\begin{eqnarray}\label{eq2}
\frac{\tilde{b}}{\tilde{a}}(\tilde{h}(X)) = \frac{b}{a}(h(X)), \quad \forall X \in ]0,A[.
\end{eqnarray}
This is the first statement of Theorem \ref{mainthm}. We now use the explicit form of the potentials. We recall that, using the notation $r = r(h(X))$,
\[  a(x) = \frac{\sqrt{F(r)}}{r}, \quad b(x) = m \sqrt{F(r)}, \quad c(x) = \frac{qQ}{r}.\]
Thus, Equation (\ref{eq2}) gives us that
\[ \tilde{r}(\tilde{h}(X)) = r(h(X)).\]
Moreover, using Equation (\ref{eq1}),
\[  \frac{F(r(h(X)))}{(qQ-\lambda r(h(X)))^2} = \frac{\tilde{F}(r(h(X)))}{(q\tilde{Q}-\lambda r(h(X)))^2}.\]
Thus,
\[ F(r) (q\tilde{Q}-\lambda r)^2  = \tilde{F}(r) (qQ-\lambda r)^2.\]
Finally, using the definition of $F$,
\[F(r) = 1 - \frac{2M}{r} + \frac{Q^2}{r^2} - \frac{\Lambda r^2}{3}\]
and identifying the coefficients of $r^6$, $r^5$ and $r^3$ we obtain respectively
\[ \frac{\lambda^2 \Lambda}{3} = \frac{\lambda^2 \tilde{\Lambda}}{3},\]
\[ 2 \lambda q \tilde{Q} \frac{\Lambda}{3} = 2 \lambda q Q \frac{\tilde{\Lambda}}{3}\]
and
\[2 \lambda q \tilde{Q} + 2 M\lambda^2 = 2 \lambda q Q + 2 \tilde{M} \lambda^2.\]
If $\lambda \neq 0$, these equalities allow us to conclude,
\[ M = \tilde{M}, \quad Q = \tilde{Q}, \quad \Lambda = \tilde{\Lambda}.\]
If $\lambda = 0$ ($q \neq 0$) we have
\[ F(r) \tilde{Q}^2  = \tilde{F}(r) Q^2.\]
Thus, using the definition of $F$ we easily obtain that
\[ M = \tilde{M}, \quad Q^2 = \tilde{Q}^2, \quad \Lambda = \tilde{\Lambda}.\]
Then,
\[ F(r) = \tilde{F}(r),\]
and using (\ref{eq1}) we thus obtain $Q = \tilde{Q}$.\\
This conclude the proof of Theorem \ref{mainthm}.
\hfill $\square$

\subsection{Proof of the main Theorem under the second assumption}\label{preuveR}
\noindent
The aim of this Section is to prove Theorem \ref{mainthm} if the assertion $(ii)$ is satisfied, i.e. to prove the uniqueness of the parameters $(M,Q,\Lambda)$ of a dS-RN black hole from
the knowledge of the reflection coefficient of the partial scattering operators, $R(\lambda,n)$, for a fixed energy $\lambda \in \mathbb{R}$ and for all $n \in \mathcal{L} \subset \mathbb{N}^{\star}$ 
satisfying the M\"{u}ntz condition $\sum_{n \in \mathcal{L}} \frac{1}{n} = \infty$. The strategy is exactly the same as the one used previously, the only difference comes 
from the fact that the equalities $R(\lambda,n) = \tilde{R}(\lambda,n)$ do not exactly imply that $\hat{R}(\lambda,n) = \tilde{\hat{R}}(\lambda,n)$.

Consider thus two dS-RN black holes with parameters $(M,Q,\Lambda)$ and $(\tilde{M},\tilde{Q},\tilde{\Lambda})$ respectively. We assume that
 \[ R(\lambda,n) = \tilde{R}(\lambda,n), \quad \forall n \in \mathcal{L}.\]
Thus, using the link between the scattering operator $S$ and the scattering matrix $\hat{S}$ given in Proposition \ref{LienSS} we deduce from this 
equality that
\[\hat{R}(\lambda,n)e^{-2i\beta} =  \tilde{\hat{R}}(\lambda,n)e^{-2i\tilde{\beta}}, \quad \forall n \in \mathcal{L},\]
where $\beta$ is the constant defined in (\ref{beta}).

\begin{lemma}
 Suppose that
\[\hat{R}(\lambda,n)e^{-2i\beta} =  \tilde{\hat{R}}(\lambda,n)e^{-2i\tilde{\beta}}, \quad \forall n \in \mathcal{L}.\]
 Then
 \[A:= \int_{\mathbb{R}} a(t)\, \mathrm dt = \int_{\mathbb{R}} \tilde{a}(t)\, \mathrm dt =: \tilde{A}.\]
\end{lemma}

\begin{proof}
 The proof of this Lemma is strictly the same as the proof of Lemma \ref{lemA=Atilde}.
 We use the equations (\ref{defR}), (\ref{eqALR}), (\ref{eqALR2}) and (\ref{equnit}) to obtain that
 \[ \hat{R}(\lambda,n) = -\hat{A}_{L1}(\lambda,n)^{-1} \hat{A}_{L2}(\lambda,n)\]
 and
\[ \hat{A}_{L1}(\lambda,n) \hat{A}_{L1}(\lambda,n)^{\star} - \hat{A}_{L2}(\lambda,n) \hat{A}_{L2}(\lambda,n)^{\star} = I_{2}.\]
Thus,
\begin{eqnarray*}
 \hat{A}_{L1}(\lambda,n)^{\star} &=& \hat{A}_{L1}(\lambda,n)^{-1} - \hat{R}(\lambda,n)\hat{A}_{L2}(\lambda,n)^{\star} \\
 &=& \hat{A}_{L1}(\lambda,n)^{-1} - e^{2i(\beta - \tilde{\beta})}\tilde{\hat{R}}(\lambda,n)\hat{A}_{L2}(\lambda,n)^{\star} \\
 &=& \hat{A}_{L1}(\lambda,n)^{-1} +e^{2i(\beta - \tilde{\beta})} \left(( \tilde{\hat{A}}_{L1}(\lambda,n)^{\star} -  \tilde{\hat{A}}_{L1}(\lambda,n)^{-1})(\tilde{\hat{A}}_{L2}(\lambda,n)^{\star})^{-1} \right) \hat{A}_{L2}(\lambda,n)^{\star}.
\end{eqnarray*}
 Then,
 \[ \hat{A}_{L1}^{\star}(\hat{A}_{L2}^{\star})^{-1} -  \hat{A}_{L1}^{-1}(\hat{A}_{L2}^{\star})^{-1} = e^{2i(\beta - \tilde{\beta})} \left(  \tilde{\hat{A}}_{L1}^{\star}(\tilde{\hat{A}}_{L2}^{\star})^{-1} -  \tilde{\hat{A}}_{L1}^{-1}(\tilde{\hat{A}}_{L2}^{\star})^{-1} \right).\]
 Moreover, the hypothesis implies that, 
 \[(\hat{R}(\lambda,n)^{\star})^{-1} = e^{2i(\beta - \tilde{\beta})} (\tilde{\hat{R}}(\lambda,n)^{\star})^{-1},\]
 i.e.
 \[ \hat{A}_{L1}^{\star} (\hat{A}_{L2}^{\star})^{-1} = e^{2i(\beta - \tilde{\beta})} \tilde{\hat{A}}_{L1}^{\star} (\tilde{\hat{A}}_{L2}^{\star})^{-1}.\]
 Thus,
 \[ \hat{A}_{L1}^{-1}(\hat{A}_{L2}^{\star})^{-1} = e^{2i(\beta - \tilde{\beta})} \tilde{\hat{A}}_{L1}^{-1}(\tilde{\hat{A}}_{L2}^{\star})^{-1}.\]
 Finally, using the asymptotics given by Theorem \ref{asA_L}, we obtain that $A = \tilde{A}$.
 \end{proof}

Hence, as in the previous Subsection, we can introduce, a matrix $P(X,\lambda,z)$ for $X \in ]0,A[$. 
However, due to the presence of the term $e^{2i(\beta - \tilde{\beta})}$, the definition of this matrix is a little bit different (but is 
the definition given in the Section 5 of \cite{DN} if we set $c = \beta - \tilde{\beta}$). Indeed, we define $P(X,\lambda,z)$ by
\[P(X,\lambda,z)  \tilde{F}_R(\tilde{h}(X),\lambda,z) = F_R(h(X),\lambda,z) e^{i (\beta - \tilde{\beta}) \Gamma^1}.\]
The strategy is now exactly the same as the one used previously: we have to prove that the matrix $P(X,\lambda,z)$ is constant equal to $\pm I_4$. 
After some calculations we obtain that
\[P_1(X,\lambda,z) =  e^{i(\beta - \tilde{\beta})} F_{R2}\hat{A}_{R2}^{-1} \tilde{F}_{L1}^{\star} - e^{-i(\beta - \tilde{\beta})} F_{L2} (\tilde{\hat{A}}_{R3}^{-1})^{\star} \tilde{F}_{R1}^{\star}\]
and
\[P_2(X,\lambda,z) = - e^{i(\beta - \tilde{\beta})} F_{R2}\hat{A}_{R2}^{-1} \tilde{F}_{L3}^{\star} + e^{-i(\beta - \tilde{\beta})} F_{L2} (\tilde{\hat{A}}_{R3}^{-1})^{\star} \tilde{F}_{R3}^{\star}.\]
Thanks to these equalities and the asymptotics given in Theorems \ref{asJostdroite}, \ref{asJostgauche} and \ref{asA_L}, we can apply the Phragm\'{e}n-Lindel\"{o}f's and Liouville's 
Theorems to obtain that $z \mapsto P(X,\lambda,z)$ is constant on $\mathbb{C}$. Thus, as previously, we use the asymptotics of the Jost functions on the real line 
given by Theorems \ref{asJostdroiteR}, \ref{asJostgaucheR} and \ref{asA_L}, to obtain that
\[P(X,\lambda,z) = \pm I_4.\]
Finally, to finish the proof of Theorem \ref{mainthm} in the case $(ii)$ we use the same procedure than the one in Section \ref{conclu}.
\hfill $\square$

\vspace{0,5cm}

\noindent
\textit{Acknowledgments:} This paper was initiated by T.Daud\'e and F.Nicoleau during the PhD of the author. 
The author wants to deeply thank T.Daud\'e and F.Nicoleau for their help and their encouragement.

{}

\end{document}